%% file: tpca_glauber_main.tex
\documentclass[11pt]{article}

\input{preamble}
\title{Averaged Annealed Markov Chains Achieve the Algorithmic Threshold for Tensor PCA, Part I: Glauber Dynamics}
\author{%
  Aniket Das \qquad Tselil Schramm\\[0.5em]
  {\normalsize Stanford University}\\
  {\normalsize\texttt{\{aniketd,tselil\}@stanford.edu}}
}
\date{}

\begin{document}

\hypersetup{pageanchor=false}
\pagenumbering{gobble}
\maketitle
\input{abstract}
\newpage
\tableofcontents
\vfill
\paragraph{Statement on AI Usage.}
All ideas and technical contributions in this paper, as well as the forthcoming companion paper, were biologically generated. Claude Code Opus 5 (and later, Opus 5.5) was used for language editing, proofreading, and transcribing a hand-drawn version of \cref{fig:phase-diagram} into TikZ code. The authors bear sole responsibility for the correctness of this manuscript.
\label{par:ai-disclosure}
\newpage
\pagenumbering{arabic}
\hypersetup{pageanchor=true}

\input{intro}
\input{proof_sketch}
\input{preliminaries}
\input{averaged_differentials}
\input{ergodic_averaging}
\input{strong_recovery_odd}

\input{strong_recovery_every_order}

\printbibliography

\clearpage
\appendix
\input{appendix}

\end{document}

%% file: preamble.tex
\usepackage{tpca_glauber}
\usepackage{macros}

\showmeetnotestrue

\newif\ifparagraphindent
\paragraphindentfalse
\ifparagraphindent\else
\fi

%% file: abstract.tex
\begin{abstract}
\normalsize

In the tensor principal components analysis (tensor PCA) problem, we are asked to recover an unknown vector in $\{\pm1\}^N$ from a noisy observation of its $k$-th tensor power, for a fixed $k\ge3$. 
Tensor PCA exhibits a well-studied information-computation gap: global algorithms such as spectral methods and sum-of-squares work above an algorithmic threshold, below which we know matching lower bounds (against these and some other restricted models of computation). 
In contrast, local algorithms such as Markov chains were only known to succeed from a random initialization at a polynomially larger signal-to-noise ratio, and were believed to be inferior to global algorithms for tensor PCA. 
We show that Glauber dynamics, when supplemented with annealing and ergodic averaging, achieves polynomial-time strong recovery above the algorithmic threshold from a uniformly random initialization, demonstrating that local Markov chains can succeed all the way down to the algorithmic threshold from an uninformative cold start.

\end{abstract}

%% file: intro.tex
\section{Introduction}
\label{sec:introduction}

In this work, we study local Markov chains for the tensor principal components analysis (tensor PCA) problem in the ``spiked tensor'' model, which involves recovering an unknown vector from a noisy observation of its tensor power. Originally introduced by \citet{richard2014statistical} as a generalization of classical spiked matrix models in statistics and machine learning~\citep{johnstone2001distribution,baik2005phase}, tensor PCA is a canonical planted problem in average-case complexity theory \citep{hopkins2015tensor,hopkins2016fast,hopkins2017power,bhattiprolu2017certificates,brennan2020reducibility,brennan2021queries,dudeja2021statistical,kunisky2022notes,potechin2022subexponential,kunisky2024cumulants,wein2025kikuchi,kothari2026smooth,bresler2026average}. It is also widely studied in high-dimensional statistics~\citep{lesieur2017statistical,perry2020statistical,jagannath2020statistical} and in the theory of spin glasses in statistical physics \citep{auffinger2013random,benarous2019landscape,ros2019complex,biroli2020iron,arous2020algorithmic}.

This paper is the first in a two-part work and focuses on tensor PCA over the signed Boolean hypercube $\Cube\doteq\{\pm1\}^N$. Our forthcoming companion paper~\citep{das2026averaged} addresses the closely related spherical tensor PCA problem. For a fixed integer $k \geq 3$ and an unknown planted vector $\xopt\in\Cube$, we observe the symmetric order-$k$ tensor
\begin{align}
    Y=\lambda\sqrt{N}\,(\plant)^{\otimes k}+G, \qquad \plant\doteq\frac{\xopt}{\sqrt{N}}
    \label{eq:tensor-pca-observation}
\end{align}
where $\lambda\ge0$ denotes the signal-to-noise ratio and $G \in (\R^N)^{\otimes k}$ denotes the standard symmetrized Gaussian tensor, i.e., the symmetrization of a $k$-tensor with i.i.d.\ $\cN(0,1)$ entries.
The goal is to recover $\xopt$ (modulo sign for even $k$) from a \emph{single observation} of the noisy tensor $Y$.
Two notions of recovery are standard in the literature. A (possibly randomized) estimator $\widehat x=\widehat x(Y)\in\Cube$ for $\xopt$ is said to achieve \emph{weak} (resp. \emph{strong}) \emph{recovery} if its correlation $N^{-1}\dotp{\widehat x}{\xopt}$ is strictly bounded away from zero (resp. $=1-o_N(1)$). For even $k$, the correlation is instead defined as $N^{-1}\abs{\dotp{\widehat x}{\xopt}}$.

The statistical limits of recovery are well understood. Focusing on weak recovery, there exists a sharp threshold $\lambda_{\mathrm{IT}}=\Theta_k(1)$ such that weak recovery is information-theoretically impossible for ${\lambda<\lambda_{\mathrm{IT}}}$ and achievable for $\lambda>\lambda_{\mathrm{IT}}$~\citep{lesieur2017statistical,perry2020statistical,jagannath2020statistical} (strong recovery requires $\lambda = \omega_N(1)$ at any diverging rate ~\citep{lesieur2017statistical,jagannath2020statistical,richard2014statistical}).
By contrast, the current \emph{algorithmic threshold} is polynomially larger, at ${\lambda_{\mathrm{ALG}}\asymp N^{\tfrac{k-2}{4}}}$. For ${\lambda\gtrsim\lambda_{\mathrm{ALG}}}$, strong recovery is known to be achievable in polynomial time by a variety of algorithms, ranging from tensor unfolding~\citep{richard2014statistical,hopkins2015tensor} and partial traces~\citep{hopkins2016fast,anandkumar2017homotopy,lovig2025almost} to the sum-of-squares (SoS) hierarchy~\citep{hopkins2015tensor} and other related spectral methods~\citep{wein2025kikuchi,kothari2026smooth}. On the other hand, for $\lambda\ll\lambda_{\mathrm{ALG}}$, no polynomial-time algorithm is believed to exist, even for weak recovery.\footnote{A natural candidate for recovery is the maximum likelihood estimator, namely the maximizer of $\dotp{Y}{x^{\otimes k}}$ over $ x\in \Cube$. Maximizing such polynomials, even approximately, is $\NP$-hard in the worst case~\citep{hastad2001some,hillar2013most,barak2012hypercontractivity}, but this has little bearing on structured average-case tensors of the form in \cref{eq:tensor-pca-observation}.}
Consequently, tensor PCA is conjectured to exhibit a \emph{statistical-computational gap} of order $N^{\tfrac{k-2}{4}}$ (see \cref{fig:phase-diagram}). While this longstanding conjecture remains open, it is supported by lower bounds against many classes of algorithms, such as the SoS hierarchy~\citep{hopkins2015tensor,hopkins2017power,potechin2022subexponential}, low-degree polynomials~\citep{kunisky2022notes,kunisky2024cumulants}, statistical query algorithms~\citep{dudeja2021statistical,brennan2021queries}, and average-case reductions from other conjecturally hard planted problems~\citep{brennan2020reducibility,bresler2026average}.

\begin{figure}[t]
\centering
\input{phase_diagram}
\caption{Phase diagram for the recovery problem in tensor PCA in terms of $\lambda$.}
\label{fig:phase-diagram}
\end{figure}
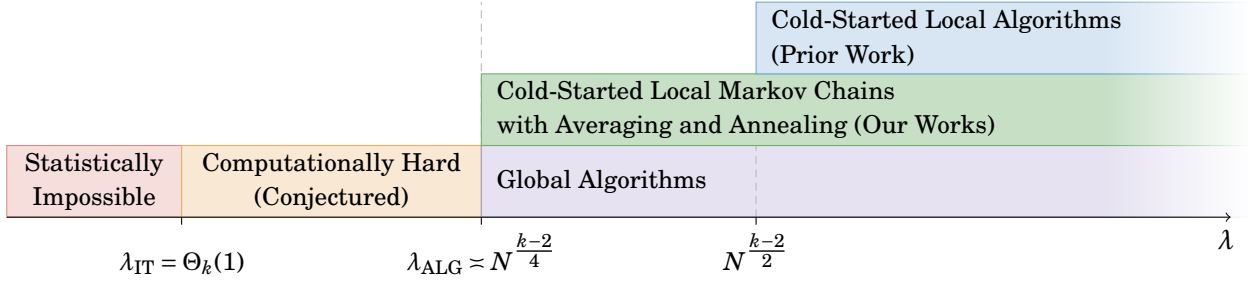

Beyond this classical statistical-computational gap, tensor PCA is also believed to exhibit a finer-grained gap between \emph{global} algorithms, such as the SoS hierarchy, tensor unfolding, and other spectral methods discussed above, and \emph{local} Markov chains, such as Glauber and Langevin dynamics. While global algorithms achieve strong recovery in polynomial time for ${\lambda\gtrsim\lambda_{\mathrm{ALG}}}$, local Markov chains initialized from an uninformative cold start (e.g., a uniformly random configuration) are only known to succeed for ${\lambda\gg N^{\tfrac{k-2}{2}}}$~\citep{arous2020algorithmic}, a polynomial factor of $N^{\tfrac{k-2}{4}}$ above the algorithmic threshold. 
Indeed, this gap has been recognized since the seminal work of \citet{richard2014statistical}, who predicted that local iterative algorithms, such as tensor power iteration and approximate message passing (AMP), when initialized at random, fail to recover the planted signal whenever ${\lambda\ll N^{\tfrac{k-2}{2}}}$~\citep[see also][]{huang2022power,wu2024sharp,lesieur2017statistical}. To our knowledge, the only local iterative algorithms that provably succeed above the algorithmic threshold (modulo polylogarithmic factors) are those warm-started by a global method~\citep{richard2014statistical,anandkumar2017homotopy,lovig2025almost}.

Our work focuses on Glauber dynamics, a canonical local Markov chain that is widely used for Bayesian inference in statistical physics, theoretical computer science, and statistics~\citep{zdeborova2016statistical,liu2024locally,geman1984stochastic}. Given an inverse temperature $\beta$, Glauber dynamics (\cref{alg:continuous-time-glauber}) asymptotically converges to the Gibbs measure ${\mu_{Y,\beta}(x)\propto\exp(\beta N^{-\tfrac{k-1}{2}}\dotp{Y}{x^{\otimes k}})}$~\citep{levin2017markov}. This property motivates its use in inference problems, since at $\beta=\lambda$ this Gibbs measure is precisely the posterior distribution of $\xopt$ under a uniform prior, a sample from which attains the information-theoretic limits for both weak and strong recovery~\citep{nishimori2001statistical,lesieur2017statistical}. However, proving polynomial-time convergence of Glauber dynamics to $\mu_{Y,\lambda}$ appears out of reach of current techniques. Nevertheless, this does not necessarily rule out polynomial-time recovery by Glauber dynamics, since recovery only requires the dynamics to reach configurations correlated with $\xopt$, a much weaker requirement than convergence to the posterior. This motivates the central question of our work.

\begin{center}
\emph{``Can polynomial-time algorithms based on canonical local Markov chains, such as Glauber dynamics initialized from a cold start, achieve recovery down to the algorithmic threshold?''}
\end{center}
For local Markov chains run at a fixed temperature, prior works conjecture a negative answer. Specifically, \citet[Section~1]{arous2020algorithmic} and \citet[Section~2(b)]{bandeira2023free} conjecture that such chains initialized uniformly at random remain trapped near zero correlation over polynomial timescales whenever ${\lambda\ll N^{\tfrac{k-2}{2}}}$. However, currently available rigorous lower bounds are restricted to disorder-dependent Gibbs-type initializations, which trap local Markov chains in a band of near-zero correlation for subexponential timescales via free-energy barriers~\citep{arous2020algorithmic,bandeira2023free}. Such initializations are, in some sense, arguably closer to the worst case than to typical uninformative cold starts such as a uniformly random configuration, for which the conjecture has remained open~\citep[Remark~2.10]{bandeira2023free}.

\subsection{Our Results}
\label{subsec:our-results}
\begin{algorithm}[t]
\caption{Annealed Glauber Dynamics with Ergodic Averaging}
\label[algorithm]{alg:averaged-annealed-glauber}
\textbf{Input:} Hamiltonian $H:\Cube\to\R$, initial configuration $\tilX_0\in\Cube$, inverse temperatures $\beta_0,\beta>0$, and simulation times $T_{\mathrm{av}},T>0$.
\makeatletter
\providecommand*{\theHALG@line}{}
\renewcommand*{\theHALG@line}{\thealgorithm.\arabic{ALG@line}}
\makeatother
\begin{algorithmic}[1]
    \State Run continuous time Glauber dynamics (\cref{alg:continuous-time-glauber}) with Hamiltonian $H$, inverse temperature $\beta_0$, initial configuration $\tilX_0$, and time $T_{\mathrm{av}}$ to obtain $(\tilX_t)_{0\le t\le T_{\mathrm{av}}}$.
    \State Compute the ergodic average:
    \begin{align*}
        \Xbar_{T_{\mathrm{av}}}\doteq\frac{1}{T_{\mathrm{av}}}\int_0^{T_{\mathrm{av}}}\tilX_t\,dt.
    \end{align*}
    \State Run continuous time Glauber dynamics with Hamiltonian $H$, inverse temperature $\beta$, initial configuration $X_0=\sgn(\Xbar_{T_{\mathrm{av}}})$, and time $T$ to obtain $(X_t)_{0\le t\le T}$.
\end{algorithmic}
\textbf{Output:} $X_T$.
\end{algorithm}

Our main result is a simple, fully local algorithm for tensor PCA, achieving polynomial-time strong recovery for all $\lambda \gtrsim \lambda_{\mathrm{ALG}}$.
This is perhaps surprising, given the so-far observed gap between local and non-local algorithms (and in light of \cite{richard2014statistical,arous2020algorithmic}). 
Our algorithm is continuous-time Glauber dynamics from a uniform initialization, with a simple twist.
\mbox{\Cref{alg:averaged-annealed-glauber}} first runs high temperature Glauber dynamics and computes the ergodic average of its trajectory. It then rounds this average to the hypercube and uses the resulting configuration as a warm start for Glauber dynamics at a much lower, constant temperature. As described in the following theorem, we prove that this algorithm succeeds at polynomial-time strong recovery above the algorithmic threshold whenever the tensor order $k$ is odd.

\begin{theorem}[Informal]
\label{thm:glauber-odd-k-informal}
Let $k\ge3$ be odd, and let $H(x)\doteq N^{-\tfrac{k-1}{2}}\dotp{Y}{x^{\otimes k}}$ denote the tensor PCA Hamiltonian with $\lambda\gtrsim\lambda_{\mathrm{ALG}}$. Then, with a uniformly random initialization and suitable inverse temperatures, \cref{alg:averaged-annealed-glauber} run with $T_{\mathrm{av}}\gtrsim N^{2k+o(1)}$ and $\log N\lesssim T\lesssim\poly(N)$ strongly recovers the planted signal with probability at least $1-\exp(-N^{1/2-o(1)})$.
\end{theorem}

When the tensor order $k$ is even, the ergodic average no longer provides an informative warm start. However, this is easily remedied by using ergodic second moments instead. To this end, \cref{alg:annealed-glauber-ergodic-second-moments}, which arises as a simple modification of \cref{alg:averaged-annealed-glauber}, initializes low temperature Glauber dynamics using the top eigenvector of the ergodic second moment of ${F(x)=(x,1)\in\R^{N+1}}$ along the trajectory of high temperature Glauber dynamics. We prove that the resulting algorithm achieves polynomial-time strong recovery above the algorithmic threshold for every tensor order $k\ge3$, irrespective of its parity.

\begin{algorithm}[t]
\caption{Annealed Glauber Dynamics with Augmented Ergodic Second Moment }
\label[algorithm]{alg:annealed-glauber-ergodic-second-moments}
\textbf{Input:} Hamiltonian $H:\Cube\to\R$, initial configuration $\tilX_0\in\Cube$, inverse temperatures $\beta_0,\beta>0$, and simulation times $T_{\mathrm{av}},T>0$.
\makeatletter
\providecommand*{\theHALG@line}{}
\renewcommand*{\theHALG@line}{\thealgorithm.\arabic{ALG@line}}
\makeatother
\begin{algorithmic}[1]
    \State Run continuous time Glauber dynamics (\cref{alg:continuous-time-glauber}) with Hamiltonian $H$, inverse temperature $\beta_0$, initial configuration $\tilX_0$, and time $T_{\mathrm{av}}$ to obtain $(\tilX_t)_{0\le t\le T_{\mathrm{av}}}$.
    \State Compute the augmented ergodic second moment:
    \begin{align*}
        \bar M_{T_{\mathrm{av}}}\doteq\frac{1}{T_{\mathrm{av}}}\int_0^{T_{\mathrm{av}}}F(\tilX_t)^{\otimes2}\,dt,
        \qquad
        F(x)\doteq(x,1)\in\R^{N+1}.
    \end{align*}
    \State Compute a top unit eigenvector $(v,a)\in\R^N\times\R$ of $\bar M_{T_{\mathrm{av}}}$, oriented so that $a\ge0$.
    \State Run continuous time Glauber dynamics with Hamiltonian $H$, inverse temperature $\beta$, initial configuration $X_0=\sgn(v)$, and time $T$ to obtain $(X_t)_{0\le t\le T}$.
\end{algorithmic}
\textbf{Output:} $X_T$.
\end{algorithm}

\begin{theorem}[Informal]
\label{thm:glauber-every-k-informal}
Let $k\ge3$, and let $H(x)\doteq N^{-\tfrac{k-1}{2}}\dotp{Y}{x^{\otimes k}}$ denote the tensor PCA Hamiltonian with $\lambda\gtrsim\lambda_{\mathrm{ALG}}$. Then, with a uniformly random initialization and suitable inverse temperatures, \cref{alg:annealed-glauber-ergodic-second-moments} run with $T_{\mathrm{av}}\gtrsim N^{2k+2+o(1)}$ and $\log N\lesssim T\lesssim\poly(N)$ strongly recovers the planted signal with probability at least $1-\exp(-N^{1/2-o(1)})$.
\end{theorem}

\begin{remark}
\label{rem:second-moment-algorithm-locality}
We note that the warm start in \cref{alg:annealed-glauber-ergodic-second-moments} is computed via an ergodic statistic of the underlying local Markov chain itself. This differs from the prior warm-started local algorithms of~\citep{richard2014statistical,anandkumar2017homotopy,lovig2025almost}, which obtain their initializations from the observed tensor $Y$ via a global algorithm. Furthermore, because this ergodic statistic is accumulated online, the eigenvector computation of Step~3 in \cref{alg:annealed-glauber-ergodic-second-moments} can itself be carried out locally, by evolving Oja's flow~\citep{oja1982simplified,kumar2023streaming} alongside the high temperature trajectory,
\begin{align*}
    \dot z_t=\eta_t\dotp{F(\tilX_t)}{z_t}F(\tilX_t),
    \qquad
    F(x)=(x,1)\in\R^{N+1},
    \qquad
    z_0\sim\Unif(\bS^N),
\end{align*}
with step sizes $\eta_t>0$, and replacing $(v,a)$ with $z_{T_{\mathrm{av}}}$ normalized to unit norm and oriented so that its last coordinate is nonnegative. Since the trajectory of continuous-time Glauber dynamics is piecewise constant (see \cref{subsec:cts-glauber} and~\citep[Chapter~20]{levin2017markov}), this ODE can be solved in closed form between consecutive updates of the chain, and over each such interval it amounts to a rank-one update of $z_t$ determined by the current configuration. The resulting algorithm maintains only the current configuration and a single auxiliary vector, much like the running average in \cref{alg:averaged-annealed-glauber}. For simplicity, we retain the explicit eigenvector computation, since analyzing this variant would divert attention from our main ideas.
\end{remark}

We highlight that \cref{alg:averaged-annealed-glauber,alg:annealed-glauber-ergodic-second-moments} are agnostic to the Hamiltonian, i.e., they are defined for any arbitrary $H:\Cube\to\R$ and make no use of the specific form of the tensor PCA Hamiltonian. Furthermore, our proofs rely on only a few structural properties of the problem. To this end, we expect our approach to extend to other planted problems as well as domains beyond the hypercube. In fact, our forthcoming companion paper~\citep{das2026averaged} carries these ideas over to spherical tensor PCA, and proves that spherical Langevin dynamics, when supplemented with annealing and ergodic averaging, succeeds at polynomial-time strong recovery above the algorithmic threshold.

\subsection{Prior Work}
\label{subsec:prior-work}

\paragraph{The algorithmic threshold for Tensor PCA.} The spiked tensor formulation of tensor PCA was introduced by \citet{richard2014statistical}, who in the same work showed that tensor unfolding attains the algorithmic threshold for even $k$, and conjectured that it does so for odd $k$ as well. Beyond this, they predicted that local iterative algorithms such as tensor power iteration and AMP fall short of the algorithmic threshold, succeeding only above ${\lambda\gtrsim N^{\tfrac{k-2}{2}}}$ when initialized at random. Subsequently, \citet{hopkins2015tensor} used the SoS hierarchy to develop a variant of tensor unfolding that attains the algorithmic threshold for odd $k$, and proved that degree-$4$ SoS fails below this threshold (modulo logarithmic factors) for $k=3$. Building on these SoS proofs, \citet{hopkins2016fast} developed efficient spectral algorithms based on partial traces for $k=3$. Below the algorithmic threshold, the concurrent works of \citet{bhattiprolu2017certificates} and \citet{raghavendra2017strongly} showed that degree-$D$ SoS, which runs in time $N^{O(D)}$, succeeds whenever ${\lambda\gtrsim(N/D)^{\tfrac{k-2}{4}}}$, modulo logarithmic factors. This tradeoff was later recovered for even $k$ by \citet{wein2025kikuchi} via the Kikuchi hierarchy, and subsequently sharpened for every order by \citet{kothari2026smooth}~\citep[see also][]{hastings2020classical,li2025smooth}. These results are complemented by lower bounds in restricted models of computation. Specifically, \citet{hopkins2017power} and \citet{potechin2022subexponential} showed that degree-$N^{\Omega(\eps)}$ SoS fails whenever ${\lambda\le N^{\tfrac{k-2}{4}-\eps}}$. \citet{kunisky2022notes} proved that degree-$D$ polynomials fail to distinguish the spiked tensor $Y$ from the pure noise tensor $G$ whenever ${\lambda\ll(N/D)^{\tfrac{k-2}{4}}}$, matching the tradeoff above. For odd $k$, \citet{kunisky2024cumulants} extended this lower bound to weak recovery. \citet{dudeja2021statistical} and \citet{brennan2021queries} concurrently showed that statistical query algorithms fail below the algorithmic threshold in a multi-sample formulation of tensor PCA. Beyond such lower bounds for restricted classes of algorithms, \citet{brennan2020reducibility} established conditional hardness of tensor PCA below the algorithmic threshold, assuming a hypergraph version of the planted clique conjecture. Recently, \citet{bresler2026average} gave average-case reductions from planted $k$-XOR to tensor PCA.

\paragraph{Local algorithms.} Turning to local algorithms, \citet{huang2022power}, and subsequently \citet{wu2024sharp}, showed that power iteration from a random initialization succeeds when ${\lambda\gtrsim N^{\tfrac{k-2}{2}}}$, confirming the prediction of \citet{richard2014statistical}. \citet{arous2020algorithmic} proved that gradient flow and Langevin dynamics at any fixed temperature likewise succeed from a random initialization once ${\lambda\gg N^{\tfrac{k-2}{2}}}$. Conversely, the same authors~\citep[Section~1]{arous2020algorithmic}, and later \citet[Section~2(b)]{bandeira2023free}, conjectured that below this scale, local Markov chains initialized uniformly at random fail to recover the signal in polynomial time, and proved that such chains remain trapped near zero correlation for subexponential timescales when started from disorder-dependent Gibbs-type initializations. Their conjecture is supported by the annealed \mbox{Kac--Rice} computations of \citet{benarous2019landscape}, who argued heuristically that local ascent methods started at random remain stuck in uninformative local maxima over polynomial timescales unless ${\lambda\gtrsim N^{\tfrac{k-2}{2}}}$. A parallel line of work studies local algorithms warm-started by a global method. This line of inquiry dates back to \citet{richard2014statistical}, who showed that power iteration succeeds above the algorithmic threshold when initialized from a sufficiently correlated warm start, such as the output of tensor unfolding. \citet{anandkumar2017homotopy} showed that for $k=3$, a partial trace of $Y$ also suffices as such a warm start. Subsequently, \citet{lovig2025almost} designed a randomized greedy algorithm that succeeds above the algorithmic threshold for any odd $k\ge3$ when initialized from a sign-rounded partial trace. Beyond this, \citet{biroli2020iron} proposed averaged gradient descent, an iterative method which follows the gradient averaged over many random replicas at a fixed overlap with the current iterate, an averaging they describe as global (for even $k$, an eigenvector of an averaged Hessian replaces the gradient). They conjectured that this method succeeds above the algorithmic threshold, and provided evidence through \mbox{Kac--Rice} heuristics and numerical simulations for $k=3$. In contrast, our algorithm averages the trajectory of a single high temperature Markov chain, and comes with rigorous guarantees.

\paragraph{Markov Chains beyond mixing.}
Broadly, our work fits into the growing line of research that analyzes local Markov chains for statistical estimation, even when standard tools for proving rapid mixing do not apply. \citet{liu2024locally} introduced locally stationary distributions and used them to show that high temperature Glauber dynamics achieves weak recovery in the spiked Wigner and stochastic block models. In a follow-up work, \citet{rajaraman2025markov} refined the spiked Wigner guarantee via a connection to AMP, and conjectured that annealed Glauber dynamics samples from the posterior in polynomial time~\citep[Conjecture~3.5]{rajaraman2025markov}. Surprisingly, \citet{gheissari2026mixing} showed that for spiked Wigner matrices, Langevin dynamics initialized uniformly at random mixes in logarithmic time at large constant signal-to-noise ratios, despite its exponential worst-case mixing time, and thereby achieves weak recovery. For planted clique, \citet{chen2023almost} proved that the Metropolis process of \citet{jerrum1992large} fails to recover even almost-linear planted cliques, whereas \citet{gheissari2025finding} constructed a specialized low temperature Markov chain on all vertex subsets that recovers planted cliques of size of order $\sqrt{n}$ in $n$-vertex graphs from a carefully chosen initialization, the full vertex set. For sparse tensor PCA, \citet{chen2024low} showed that low temperature Metropolis succeeds only when the signal strength is significantly above the associated algorithmic threshold, and established a matching bottleneck for low temperature local Markov chains based on the overlap gap property. \citet{bandeira2022franz} proved that below the low-degree threshold, free-energy barriers confine local Markov chains started at suitable low-correlation states in many Gaussian additive models~\citep[see also][]{benarous2023free}, and asked whether an appropriately designed Markov chain can nevertheless reach this threshold~\citep[full version, Section~2.2]{bandeira2022franz}. Our results answer this question in the affirmative for tensor PCA, once Glauber dynamics is supplemented with annealing and ergodic averaging.

\subsection{Discussion and Open Problems}
\label{subsec:discussion-open-problems}

Markov chain Monte Carlo (MCMC) methods are ubiquitous in statistical inference, yet rigorous guarantees for them have largely hinged on proofs of rapid mixing. Generic criteria for rapid mixing apply to tensor PCA only at sufficiently high temperature~\citep[e.g.,][]{eldan2022spectral,chen2022localization,anari2024trickle}, and to our knowledge, no available technique establishes polynomial-time mixing of Glauber dynamics for tensor PCA at the range of low temperatures relevant to recovery. Even so, our algorithms based on local Markov chains achieve polynomial-time strong recovery from a cold start all the way down to the algorithmic threshold, and our analysis sidesteps the question of mixing altogether. In the same spirit, \citet{liu2024locally} proved that Glauber dynamics achieves weak recovery in the spiked Wigner model without bounding its mixing time. Since the spiked Wigner model exhibits no statistical-computational gap~\citep{baik2005phase,deshpande2017asymptotic}, they asked whether Glauber dynamics, possibly with annealing, achieves computationally optimal recovery in planted problems that exhibit an information-computation tradeoff, such as the stochastic block model with many communities and planted constraint satisfaction problems (CSPs); see~\citep[Problem~1.22]{liu2024locally}. Given the close connection between tensor PCA and planted $k$-XOR~\citep{bresler2026average}, our results may be regarded as a modest step toward an affirmative answer; our use of ergodic averaging adds an interesting twist. Our forthcoming companion paper~\citep{das2026averaged} provides further evidence in this direction, as it establishes the analogous guarantee for Langevin dynamics in spherical tensor PCA. In this sense, the two works together redeem local Markov chains, perhaps the most natural algorithms that statistical physics offers for planted problems, by demonstrating that standard MCMC heuristics such as annealing and ergodic averaging~\citep{kirkpatrick1983optimization,metropolis1953equation,geyer1992practical} suffice for these chains to provably reach the algorithmic threshold.

Our results leave open several natural questions, the first of which is whether averaged annealed Glauber dynamics attains the algorithmic threshold in other planted problems of interest. Since \cref{alg:averaged-annealed-glauber,alg:annealed-glauber-ergodic-second-moments} are agnostic to the Hamiltonian, they can be run on such problems over the hypercube, possibly after replacing the ergodic statistics in their first phase with functionals better suited to the problem at hand. An affirmative answer for planted CSPs would bear more directly on the question raised in~\citep[Problem~1.22]{liu2024locally}. We conjecture that this is indeed the case for planted $k$-XOR, given its close connection to tensor PCA~\citep{bresler2026average}.

A related question is whether ergodic averaging is necessary at all, or whether annealing alone, namely Glauber dynamics run under a suitable temperature schedule from a uniformly random initialization, can reach the algorithmic threshold. We believe that averaging is indeed necessary whenever ${\lambda\ll N^{\tfrac{k-2}{2}}}$. At the high temperatures at which the ergodic average encodes an informative warm start, the state of Glauber dynamics at any given time remains essentially uncorrelated with $\xopt$. This is because the signal manifests merely as a minute bias in each coordinate, which only becomes visible after averaging over the trajectory. On the other hand, at lower temperatures, local Markov chains initialized from such nearly uniform states are conjectured to remain trapped near zero correlation in this regime~\citep{arous2020algorithmic,bandeira2023free}. Proving that annealing alone fails would establish this conjecture, which appears to be beyond the reach of existing lower-bound arguments based on free-energy barriers~\citep{arous2020algorithmic,bandeira2023free}.

A compelling direction of research is to understand the fine-grained tradeoffs between runtime and signal strength achievable by local Markov chains near the algorithmic threshold. For global algorithms, \citet{kothari2026smooth} and \citet{li2025smooth} showed that for every fixed constant $c>0$, a running time of $N^{O_{c,k}(1)}$ suffices at ${\lambda=c\lambda_{\mathrm{ALG}}}$~\citep[cf.][]{bhattiprolu2017certificates,raghavendra2017strongly}. In contrast, our algorithms are only guaranteed to succeed when ${\lambda\ge C\lambda_{\mathrm{ALG}}}$ for a sufficiently large constant $C$. Establishing analogous fine-grained tradeoffs for averaged annealed Glauber dynamics remains an interesting open problem.

The analysis of the second phase of our algorithms is restricted to a constant inverse temperature $\beta\in(0,\infty)$, whereas the posterior ${\mu_{Y,\lambda}(x)\propto\exp(\lambda N^{-\tfrac{k-1}{2}}\dotp{Y}{x^{\otimes k}})}$ corresponds to $\beta=\lambda$, which diverges with $N$. It would be interesting to determine the recovery guarantees achievable by Glauber dynamics run at such diverging inverse temperatures (possibly supplemented with annealing and ergodic averaging) and more ambitiously, whether it can sample from the posterior in polynomial time whenever ${\lambda\gtrsim\lambda_{\mathrm{ALG}}}$. In some sense, this would be a natural extension of the recent work of~\citet{gheissari2026mixing}, which showed that Langevin dynamics for the spiked Wigner model (corresponding to $k=2$) mixes in $O(\log N)$ time from a uniformly random initialization when the signal-to-noise ratio is a sufficiently large constant.

Lastly, while the second phase of our algorithms runs for time $O(\log N)$, translating to $O(N^k\log N)$ operations, their first phase runs for time $T_{\mathrm{av}}=N^{2k+o(1)}$, which amounts to $N^{3k+o(1)}$ operations. This is slower than the global spectral algorithms for tensor PCA~\citep[e.g.,][]{hopkins2015tensor,hopkins2016fast,wein2025kikuchi}, which run in time $O(N^k)$. We believe that this gap is an artifact of our analysis, which requires a polynomially high temperature in the first phase, and hence a polynomially long averaging time. To this end, we ask whether averaged annealed Glauber dynamics can achieve a near-linear running time, matching that of the global spectral algorithms.

\subsection{Acknowledgements}
\label{subsec:ack}
We thank Brice Huang for helpful discussions and feedback.
The authors acknowledge the support of the NSF AI Institute for Foundations of Machine Learning (IFML) and of T.S.'s NSF CAREER award \# 2143246.

\subsection*{Organization}

We sketch the key ideas behind our proofs in \cref{sec:proof-sketch}, and review the necessary notation and preliminaries in \cref{sec:preliminaries}. In \cref{sec:averaged-discrete-differentials}, we carry out a signal-noise decomposition of the average discrete gradient and Hessian of the tensor PCA Hamiltonian, and \cref{sec:ergodic-averaging} shows that ergodic averages of high temperature Glauber dynamics recover these averaged differentials. We then analyze the second phase of \cref{alg:averaged-annealed-glauber} in \cref{sec:strong-recovery-odd}, proving a formal version of \cref{thm:glauber-odd-k-informal}. Finally, we extend this analysis to \cref{alg:annealed-glauber-ergodic-second-moments} in \cref{sec:strong-recovery-every-order}, proving a formal version of \cref{thm:glauber-every-k-informal}.

%% file: phase_diagram.tex
\definecolor{phaseImpossible}{HTML}{D66B63}
\definecolor{phaseHard}{HTML}{E39A3B}
\definecolor{phaseGlobal}{HTML}{8C7BC0}
\definecolor{phaseLocal}{HTML}{4F8FCB}
\definecolor{phaseOurs}{HTML}{4E9A4A}
\begin{tikzpicture}[
    x=0.01\linewidth, y=1mm,
    bandtext/.style={font=\small, inner sep=0pt, align=center},
    ticktext/.style={font=\small, anchor=base, inner sep=0pt},
    guide/.style={gray!55, densely dashed},
  ]
  \def\xIT{14}
  \def\xALG{38}
  \def\xLOC{60}
  \def\xFade{88}
  \def\xEnd{100}
  \def\closedbar#1#2#3#4#5{%
    \filldraw[draw=#1!80, fill=#1!22] (#2,#4) rectangle (#3,#5);}
  \def\openbar#1#2#3#4#5{%
    \fill[#1!#2] (#3,#4) rectangle (\xFade,#5);
    \shade[left color=#1!#2, right color=white] (\xFade,#4) rectangle (\xEnd,#5);
    \draw[#1!80] (\xFade,#4) -- (#3,#4) -- (#3,#5) -- (\xFade,#5);
    \foreach \i in {1,...,20} {
      \pgfmathsetmacro\fadeop{1-\i/20}
      \draw[#1!80, draw opacity=\fadeop]
        ({\xFade+(\i-1)*(\xEnd-\xFade)/20},#4) -- ({\xFade+\i*(\xEnd-\xFade)/20},#4)
        ({\xFade+(\i-1)*(\xEnd-\xFade)/20},#5) -- ({\xFade+\i*(\xEnd-\xFade)/20},#5);
    }}
  \draw[guide] (\xALG,0) -- (\xALG,28.5);
  \closedbar{phaseImpossible}{0}{\xIT}{0}{9.5}
  \closedbar{phaseHard}{\xIT}{\xALG}{0}{9.5}
  \openbar{phaseGlobal}{22}{\xALG}{0}{9.5}
  \draw[guide] (\xLOC,0) -- (\xLOC,19);
  \openbar{phaseOurs}{32}{\xALG}{9.5}{19}
  \openbar{phaseLocal}{22}{\xLOC}{19}{28.5}
  \node[bandtext] at ({\xIT/2},4.75) {Statistically\\Impossible};
  \node[bandtext] at ({(\xIT+\xALG)/2},4.75) {Computationally Hard\\(Conjectured)};
  \node[bandtext, anchor=west] at ({\xALG+1.2},4.75) {Global Algorithms};
  \node[bandtext, anchor=west, align=left] at ({\xALG+1.2},14.25)
    {Cold-Started Local Markov Chains \\ with Averaging and Annealing (Our Works)};
  \node[bandtext, anchor=west, align=left] at ({\xLOC+1.2},23.75)
    {Cold-Started Local Algorithms\\(Prior Work)};
  \draw[->] (0,0) -- (98,0);
  \node[anchor=north, inner sep=0pt] at (97.6,-1.6) {$\lambda$};
  \foreach \x in {\xIT,\xALG,\xLOC} \draw (\x,0) -- (\x,-1);
  \node[ticktext] at (\xIT,-7) {$\lambda_{\mathrm{IT}}=\Theta_k(1)$};
  \node[ticktext] at (\xALG,-7) {$\lambda_{\mathrm{ALG}}\asymp N^{\tfrac{k-2}{4}}$};
  \node[ticktext] at (\xLOC,-7) {$N^{\tfrac{k-2}{2}}$};
\end{tikzpicture}

%% file: proof_sketch.tex
\section{Proof Sketch}
\label{sec:proof-sketch}

Before presenting our key ideas, we first describe the tensor PCA Hamiltonian in more detail and fix some notation. A natural approach to recovery is to solve the Maximum Likelihood Estimation (MLE) problem, defined as follows:
\begin{align}
    \widehat{x}_{\mathrm{MLE}}
    \in \argmax_{x\in\Cube}\dotp{Y}{x^{\otimes k}},
    \label{eq:tensor-pca-mle}
\end{align}
We approach the objective in \Cref{eq:tensor-pca-mle} indirectly, by running Glauber Dynamics with the corresponding Hamiltonian\footnote{Contrary to the usual convention in statistical physics, we define the Hamiltonian as the objective to be maximized.} $H:\Cube\to\R$:
\begin{align}
    H(x)&\doteq N^{-\tfrac{k-1}{2}}\dotp{Y}{x^{\otimes k}}
        =N\lambda m(x)^k+H_0(x),\label{eq:tensor-pca-hamiltonian}\\
    m(x)&\doteq\frac{1}{N}\dotp{\xopt}{x}
        =\frac{1}{\sqrt{N}}\dotp{\plant}{x},
        \qquad H_0(x)\doteq N^{-\tfrac{k-1}{2}}\dotp{G}{x^{\otimes k}},
        \label{eq:tensor-pca-overlap-noise}
\end{align}
where $m:\Cube\to[-1,1]$ is the normalized correlation with the planted vector and $H_0$ is the Ising pure $k$-spin Hamiltonian~\citep{panchenko2013sherrington,talagrand2011meanfield2}, a centered isotropic Gaussian process on $\Cube$ with covariance
\begin{align}
    \label{eq:kspin-covariance}
    \E[H_0(x)H_0(y)]
    =N\left(\frac{\dotp{x}{y}}{N}\right)^k,
    \qquad x,y\in\Cube.
\end{align}
We refer to $N\lambda m(x)^k$ and $H_0$ as the signal and the disorder, respectively, and let $\pi$ denote the uniform measure on $\Cube$.

We now sketch the main ideas behind our results. We discuss our analysis of \cref{alg:averaged-annealed-glauber} for odd $k$ in \cref{subsec:sketch-curvature-entropy,subsec:sketch-averaging-disorder,subsec:sketch-ergodic-averaging,subsec:sketch-reverse-planting}, and the ideas behind our analysis of \cref{alg:annealed-glauber-ergodic-second-moments} in \cref{subsec:sketch-any-order}.

\subsection{Escaping the Equator}
\label{subsec:sketch-curvature-entropy}

Let $(X_t)_{t\ge0}$ denote the trajectory of Glauber dynamics for $H$ at a constant inverse temperature $\beta>0$, and let $\partial_iH(x)\doteq\frac12\bigl[H(x_{\neg i},x_i=+1)-H(x_{\neg i},x_i=-1)\bigr]$ denote the local field at site $i$, i.e., the discrete partial derivative of $H$. As we shall describe in \cref{subsec:cts-glauber}, each time step of Glauber dynamics selects a uniformly random site $i\in[N]$ and replaces the spin $x_i$ with a fresh spin of mean $\tanh(\beta\partial_iH(x))$. Using this description, one can show that, modulo stochastic fluctuations, the correlation $m_t\doteq m(X_t)$ evolves according to the ODE
\begin{align}
    \frac{dm_t}{dt}
    \approx\frac{1}{N}\sum_{i=1}^N\xopt_i\tanh\bigl(\beta\partial_iH(X_t)\bigr)-m_t.
    \label{eq:sketch-overlap-drift}
\end{align}
Let $\dgrad H\doteq(\partial_1H,\ldots,\partial_NH)$ denote the discrete gradient of $H$. Since $\dgrad(N\lambda m^k)(x)\approx k\lambda m(x)^{k-1}\xopt$, we obtain the following by further approximating $\tanh(z) \approx z$.
\begin{align}
    \frac{dm_t}{dt}
    \approx\underbrace{\beta k\lambda\,m_t^{k-1}}_{\substack{\text{signal}\\\text{drift}}}
    -\underbrace{m_t}_{\substack{\text{entropic}\\\text{drift}}}
    +\underbrace{\frac{\beta}{N}\dotp{\xopt}{\dgrad H_0(X_t)}}_{\substack{\text{disorder}\\\text{drift}}}.
    \label{eq:sketch-overlap-ode}
\end{align}
The disorder drift is a highly intricate random function of the trajectory, since $X_t$ itself depends on $H_0$ through the entire history of the dynamics. Controlling this term is perhaps the most challenging obstacle in our analysis, and we shall return to it later in this section.

Setting aside the disorder drift for now, the evolution of $m_t$ is essentially a competition between the entropic and signal drifts. The entropic drift, which arises because each time step discards the current value of the resampled spin, drags $m_t$ towards the equator $\{x\in\Cube:\abs{m(x)}\lesssim N^{-1/2}\}$, where the uniform measure on $\Cube$ concentrates. The signal drift is induced by the local fields of the signal, which tries to align each spin $x_i$ with $\xopt_i$. The ratio of this signal drift to the entropic drift is $\beta k\lambda m_t^{k-2}$, which increases with $m_t$. We therefore expect the correlation to grow at an accelerating rate if $\beta k\lambda m_0^{k-2}>1$, and to be dragged back to the equator otherwise. In other words, Glauber dynamics initialized at correlation $m_0$ should escape the equator and recover the planted signal efficiently if and only if $\lambda m_0^{k-2}\gtrsim1$, or equivalently, $m_0\gtrsim\lambda^{-1/(k-2)}$. A similar heuristic picture was proposed for spherical tensor PCA~\citep{arous2020algorithmic}, and it is also supported by landscape complexity analyses~\citep{deoliveira1997landscape,benarous2019landscape,ros2019complex}, which show that spurious local maxima are confined to the band $\abs{m}\lesssim\lambda^{-1/(k-2)}$ around the equator.

A uniformly random initialization has correlation of order $N^{-1/2}$, which would require $\lambda\gtrsim N^{\frac{k-2}{2}}$ to escape the equator. This matches the cold start threshold for local algorithms discussed in \mbox{\cref{sec:introduction}}. At $\lambda\asymp\lambda_{\mathrm{ALG}}$, however, Glauber dynamics requires an initial correlation of order at least ${\lambda_{\mathrm{ALG}}^{-1/(k-2)}\asymp N^{-1/4}}$, even in this idealized picture that ignores the disorder. Our first obstacle is therefore to find a \emph{natural} warm start, which (i) achieves correlation $\Omega(N^{-1/4})$ when $\lambda\asymp\lambda_{\mathrm{ALG}}$, and (ii) can be approximated from the trajectory of Glauber dynamics itself. As we shall now discuss, averaging the discrete gradient of $H$ over the hypercube yields a warm start that satisfies both desiderata.

\subsection{Warm Start via the Average Discrete Gradient}
\label{subsec:sketch-averaging-disorder}

Recall that our goal is to construct a warm start from the trajectory of Glauber dynamics itself. Since Glauber dynamics depends on $H$ only through its local fields, so should any such warm start computed from its trajectory. This merits a closer analysis of the local fields, which decompose as $\dgrad H(x)\approx k\lambda m(x)^{k-1}\xopt+\dgrad H_0(x)$, and therefore form a noisy observation of $\xopt$, since $m(x)^{k-1}\ge0$ for odd $k$. However, at a uniformly random configuration, the signal contribution is of order $\lambda N^{-\frac{k-1}{2}}$ in each coordinate, and is thus dwarfed by the disorder contribution, whose coordinates are typically of order one by Gaussian concentration.

Fortunately, although $H_0$ has a rugged landscape with exponentially many local maxima~\citep{deoliveira1997landscape,auffinger2013random}, its values at typical pairs of configurations are nearly independent, since its covariance only depends on the overlap $\dotp{x}{y}/N$. Indeed, for independent $x,y\sim\pi$, the central limit theorem suggests that $\dotp{x}{y}/\sqrt N$ is asymptotically distributed as $\cN(0,1)$, and hence
\begin{align*}
    \E\bigl[H_0(x)H_0(y)\bigr]
    =N\Bigl(\frac{\dotp{x}{y}}{N}\Bigr)^{k}
    =N^{1-\frac{k}{2}}\Bigl(\frac{\dotp{x}{y}}{\sqrt N}\Bigr)^{k}
    \lesssim N^{1-\frac{k}{2}}
    \le N^{-\frac12},
\end{align*}
whereas $\E[H_0(x)^2]=N$. Averaging the local fields over $\pi$ should therefore suppress the disorder contribution, while the signal contribution survives, since it points along $\xopt$ at every configuration. This motivates the \emph{average discrete gradient} $u_H\doteq\E_\pi[\dgrad H]$.

By Stein's lemma for the uniform measure on $\Cube$, which we establish in \cref{subsec:boolean-hypercube}, $u_H=\E_\pi[xH(x)]$ is linear in the Gaussian tensor $G$. In \cref{sec:averaged-discrete-differentials}, we use the symmetries of the hypercube and the rotation invariance of the Gaussian process $H_0$ to show that $u_H$ is a noisy observation of $\xopt$,
\begin{align}
    u_H=\lambda\sigma_N^2\,\xopt+\zeta,
    \qquad
    \zeta\sim\cN(0,\sigma_N^2I),
    \qquad
    \sigma_N^2=N\,\E_{x,x'\,\iidsim\pi}\Bigl[\Bigl(\frac{\dotp{x}{x'}}{N}\Bigr)^{k+1}\Bigr]\asymp_kN^{-\frac{k-1}{2}}.
    \label{eq:sketch-average-gradient}
\end{align}
Each coordinate of $u_H$ thus has signal-to-noise ratio $\lambda\sigma_N\asymp_k\lambda N^{-\frac{k-1}{4}}$, which is a factor of $N^{\frac{k-1}{4}}$ larger than that of the local field at a single uniformly random configuration. Consequently, with high probability, $m(\sgn(u_H))\asymp_k\min\{1,\lambda N^{-\frac{k-1}{4}}\}$, so that $\lambda\,m(\sgn(u_H))^{k-2}\asymp(\lambda N^{-\frac{k-2}{4}})^{k-1}$, and $\sgn(u_H)$ is a sufficiently good warm start for escaping the equator exactly when $\lambda\gtrsim\lambda_{\mathrm{ALG}}$.

This suggests initializing Glauber dynamics at a constant temperature from $\sgn(u_H)$. However, since $u_H$ aggregates the Hamiltonian over the entire configuration space, it is therefore a global statistic. This raises the question of whether it can nevertheless be approximated from the trajectory of Glauber dynamics, as required by the second desideratum.

\subsection{Approximating the Average Gradient via High Temperature Glauber Dynamics}
\label{subsec:sketch-ergodic-averaging}

Perhaps surprisingly, we show that $u_H$ can be approximated by the ergodic average $\Xbar_T\doteq\frac{1}{T}\int_0^T\tilX_t\,dt$ of high temperature Glauber dynamics $(\tilX_t)_{t\ge0}$ for $H$, run at a small inverse temperature $\beta_0>0$. Indeed, since this chain is irreducible with stationary distribution $\mu_{H,\beta_0}\propto\pi e^{\beta_0H}$, the ergodic theorem~\citep[cf.][Theorem~C.1]{levin2017markov} implies that $\Xbar_T$ converges to the mean of $\mu_{H,\beta_0}$ as $T\to\infty$. At high temperature, $\mu_{H,\beta_0}$ is a slight tilt of the uniform measure, and expanding to first order in $\beta_0$, we infer from $\E_\pi[x]=0$ and Stein's lemma that
\begin{align*}
    \lim_{T\to\infty}\Xbar_T
    =\frac{\E_\pi\bigl[x\,e^{\beta_0H(x)}\bigr]}{\E_\pi\bigl[e^{\beta_0H(x)}\bigr]}
    =\beta_0\E_\pi\bigl[xH(x)\bigr]+O(\beta_0^2)
    =\beta_0u_H+O(\beta_0^2).
\end{align*}
In other words, the mean of the Gibbs measure at high temperature is the linear response of the uniform measure to the field $H$, which Stein's lemma identifies with $u_H$.

Since $\beta_0u_H$ vanishes as $\beta_0\to0$, making the above heuristic rigorous requires quantitative versions of both the long time limit and the high temperature expansion, which we derive in \cref{sec:ergodic-averaging}. Notably, our proofs circumvent the analysis of the stationary distribution or the mixing time of high temperature Glauber dynamics, and instead proceed by comparing the Markov chain directly with the simple random walk on $\Cube$ (which corresponds to Glauber dynamics at $\beta=0$). Since each time step of high temperature Glauber dynamics replaces a uniformly random spin $x_i$ with a fresh spin of mean $\tanh(\beta_0\partial_iH(x))\approx\beta_0\partial_iH(x)$, the ergodic average $\Xbar_T$ closely tracks that of the rescaled local fields $\beta_0\dgrad H(\tilX_t)$. In turn, the latter concentrates around $\beta_0\E_\pi[\dgrad H]=\beta_0u_H$, since high temperature Glauber dynamics is a small perturbation of the simple random walk, whose stationary distribution is $\pi$. To quantify this concentration non-asymptotically, we establish a discrete version of the Bakry--\'Emery criterion~\citep{bakry2013analysis} on the hypercube using elementary Fourier analysis, and use it in conjunction with the Green operator of the simple random walk, a classical tool for proving central limit theorems for ergodic averages of Markov chains~\citep{komorowski2012fluctuations,eberle2025markov}. For tensor PCA, the resulting bounds establish that $\beta_0=\lambda^{-1}N^{-\frac{k}{2}-o(1)}$ and $T_{\mathrm{av}}=N^{2k+o(1)}$ suffice for $\sgn(\Xbar_{T_{\mathrm{av}}})$ to agree with $\sgn(u_H)$ on all but a vanishing fraction of coordinates with high probability, so that this warm start meets both desiderata. This motivates the construction of \cref{alg:averaged-annealed-glauber}, which first runs high temperature Glauber dynamics and uses the sign rounding of its ergodic average to initialize Glauber dynamics at a constant temperature.

Our non-asymptotic analysis of the ergodic average is established for generic Hamiltonians satisfying mild regularity conditions, and could be of independent interest. We believe that our proof techniques extend to diffusions on arbitrary manifolds with positive Ricci curvature, and more generally, to Markov processes with positive Bakry--\'Emery curvature. In fact, our forthcoming companion work~\citep{das2026averaged} proves such a guarantee for spherical Langevin dynamics with smooth Hamiltonians.

Computing the warm start and climbing towards the planted signal require different temperatures, which is why \cref{alg:averaged-annealed-glauber} is annealed. At high temperature, the signal drift in \cref{eq:sketch-overlap-ode} is negligible, so the first phase cannot climb towards the planted signal, even though its ergodic average encodes $u_H$. In contrast, the second phase, run at a constant temperature from $\sgn(\Xbar_{T_{\mathrm{av}}})$, can exploit the signal drift to escape the equator. Proving that it indeed does so requires us to finally confront the disorder drift in \cref{eq:sketch-overlap-ode}. However, controlling this term now becomes all the more delicate since this warm start itself depends on the disorder.

\subsection{Controlling the Disorder via Reverse Planting}
\label{subsec:sketch-reverse-planting}

To illustrate our ideas, we make two simplifications throughout this subsection. First, we ignore the estimation error of the ergodic average in the first phase, and assume that the second phase of \cref{alg:averaged-annealed-glauber} is initialized exactly at $\sgn(u_H)$. Second, we assume that $\xopt$ is uniformly random on $\Cube$. Our proofs in \cref{sec:strong-recovery-odd} dispense with both these simplifications, and apply simultaneously to all initializations in a small Hamming ball around $\sgn(u_H)$, and to any fixed $\xopt\in\Cube$ via a carefully constructed gauge transformation.

By Gaussian concentration, the disorder drift in \cref{eq:sketch-overlap-ode} is of order one uniformly over the hypercube, irrespective of how the trajectory depends on $H_0$. Consequently, once the correlation $m_t$ exceeds a large constant multiple of the \emph{critical correlation} $r_N\doteq\lambda^{-1/(k-1)}$, the signal drift satisfies $\beta k\lambda m_t^{k-1}\gg1$, and thus dominates both the entropic and disorder drifts, even in the worst case. As we show in \cref{sec:strong-recovery-odd}, the correlation then converges to one exponentially fast, such that Glauber dynamics achieves strong recovery within $O(\log N)$ additional time, and retains it over polynomial time horizons. Recall from \cref{subsec:sketch-averaging-disorder} that the warm start has correlation $m_0\doteq m(\sgn(u_H))$ of order $\min\{1,\lambda N^{-\frac{k-1}{4}}\}$, which itself becomes critical (i.e., $m_0\ge Cr_N$ for a large constant $C$) whenever ${\lambda\gtrsim N^{\frac{k-2}{4}+\frac{1}{4k}}}$. In this regime, strong recovery thus follows immediately, and it remains to treat the window $\lambda_{\mathrm{ALG}}\lesssim\lambda\lesssim\lambda_{\mathrm{ALG}}N^{\frac{1}{4k}}$, in which the warm start satisfies the escape condition $\lambda m_0^{k-2}\gtrsim1$ from \cref{subsec:sketch-curvature-entropy}, but is subcritical, i.e., $m_0\ll r_N$. In this window, our key technical challenge is to prove that the correlation nevertheless reaches $Cr_N$ within $O(1)$ time, before the disorder drift can overpower that of the signal.

However, neither worst-case nor typical bounds on the disorder drift suffice to justify this. Below the critical correlation, the signal drift can be much smaller than the worst-case magnitude of the disorder drift, which could thus overwhelm it. Moreover, although the disorder drift is $O(N^{-1/2})$ at configurations chosen independently of $H_0$, both the trajectory and its initialization depend on $H_0$. One can imagine a particularly dangerous scenario in which the warm start lies close to one of the exponentially many local maxima of $H_0$~\citep{deoliveira1997landscape,auffinger2013random}, around which the disorder drift could dominate the signal drift and trap the dynamics at polynomial timescales.

We rule out such adversarial scenarios by showing that the dependence of the warm start on the disorder is benign. Recall from \cref{eq:sketch-average-gradient} that $u_H=\lambda\sigma_N^2\xopt+\zeta$, where, by Stein's lemma, $\zeta=\E_\pi[\dgrad H_0]=\E_\pi[xH_0(x)]$ is the vector of first order Fourier coefficients of $H_0$. In other words, the warm start depends on $H_0$ only through its first Fourier level. Since $H_0$ is an isotropic Gaussian process, Gaussian conditioning shows that $\zeta$, and hence the warm start, is independent of the remaining Fourier levels $H_0(x)-\dotp{\zeta}{x}$, which comprise the bulk of the disorder.
This suggests comparing the second phase of \cref{alg:averaged-annealed-glauber} with the same procedure run on the pure disorder $H_0$. Precisely, we consider Glauber dynamics initialized at the sign rounding of the average discrete gradient of its Hamiltonian, and call its law the \emph{planted law} when the Hamiltonian is $H$, and the \emph{null law} when the Hamiltonian is $H_0$. In other words, the planted law is the true law of the second phase, and the null law is simply the planted law with $\lambda=0$. The dynamics in the second phase is thus initialized at $\sgn(u_H)$ under the planted law, and at $\sgn(\zeta)$ under the null law. By \cref{eq:sketch-overlap-drift}, the correlation under the null law evolves according to
\begin{align}
    \frac{dm_t}{dt}
    \approx\frac{1}{N}\sum_{i=1}^N\xopt_i\bigl[\tanh\bigl(\beta\partial_iH_0(X_t)\bigr)-X_{t,i}\bigr].
    \label{eq:sketch-null-drift}
\end{align}
We call the right-hand side of \cref{eq:sketch-null-drift} the \emph{null drift}, since it is precisely the net contribution of the entropic and disorder drifts. Since $\xopt$ is uniformly random and the null dynamics never sees it, $\xopt$ is independent of the null trajectory. Conditional on the trajectory, the null drift is therefore a Rademacher average with bounded weights, and Hoeffding's inequality shows that it remains of order $N^{-1/2}$ along the entire null trajectory.

Thus, to control the disorder drift in the true dynamics, it remains to transfer this estimate to the planted law via a continuous-time change of measure argument, which we call the \emph{reverse planting trick}. As the name suggests, it runs opposite to the classical planting trick~\citep{achlioptas2008algorithmic}, in which the null law is the true law of interest, and estimates established under the planted law are transferred to it via contiguity. However, executing this argument is quite challenging since the planted and null laws are far from contiguous. Indeed, their Hamiltonians differ by the signal $N\lambda m(x)^k$, which biases both the warm start and every subsequent update of the planted dynamics. The cost of transferring an estimate therefore grows with the signal accumulated along the trajectory. Fortunately, since the worst-case argument above takes over once the correlation reaches $Cr_N$, we only need to transfer estimates up to this time. To this end, we prove an \emph{adaptive contiguity} bound, which controls the cost of transferring an event up to any stopping time in terms of the signal accumulated before it.

We apply this adaptive contiguity bound up to the \emph{failure time}, i.e., the first time before the correlation reaches $Cr_N$ at which the null drift overpowers the signal drift. Until the failure time, the signal drift dominates, so the correlation grows at an accelerating rate, as in \cref{subsec:sketch-curvature-entropy}, and reaches $Cr_N$ within $O(1)$ time unless a failure occurs first. With a delicate stopped martingale analysis, we show that with high probability, no failure occurs, so that the dynamics reaches $Cr_N$ within $O(1)$ time and then achieves strong recovery. This completes our proof sketch of \cref{thm:glauber-odd-k-informal}.

\subsection{Strong Recovery for Any Tensor Order}
\label{subsec:sketch-any-order}

We now turn our attention to \cref{alg:annealed-glauber-ergodic-second-moments}, and first focus on the case when $k$ is even. In this case, $H(-x)=H(x)$, so by symmetry, the average discrete gradient $u_H$ vanishes identically, and carries no information about the planted signal. However, our previous intuition about suppressing the disorder via averaging still applies to the discrete Hessian $\dgrad^2H\doteq(\partial_i\partial_jH)_{i,j\in[N]}$, which decomposes as $\dgrad^2H(x)\approx\frac{k(k-1)\lambda}{N}m(x)^{k-2}\xopt(\xopt)^\top+\dgrad^2H_0(x)$, and therefore forms a noisy observation of $\xopt(\xopt)^\top$, since $m(x)^{k-2}\ge0$ for even $k$. Specifically, in \cref{sec:averaged-discrete-differentials}, we show that the \emph{average discrete Hessian} has a spiked Wigner type structure with a planted spike in the direction of $\xopt$,
\begin{align}
    A_H\doteq\E_\pi\bigl[\dgrad^2H\bigr]
    =\frac{\lambda\gamma_N}{N}\bigl(\xopt(\xopt)^\top-I\bigr)+W,
    \qquad
    \gamma_N\asymp_kN^{-\frac{k-2}{2}},
    \label{eq:sketch-average-hessian}
\end{align}
where $W$ is a symmetric matrix with zero diagonal and i.i.d.\ $\cN(0,\gamma_N/N)$ entries above the diagonal. The spike in \cref{eq:sketch-average-hessian} has eigenvalue of order $\lambda\gamma_N$, whereas $\norm{W}_{\mathrm{op}}\asymp\sqrt{\gamma_N}$ with high probability. Hence, the spike dominates the noise exactly when $\lambda\sqrt{\gamma_N}\asymp\lambda N^{-\frac{k-2}{4}}\gtrsim1$, i.e., when $\lambda\gtrsim\lambda_{\mathrm{ALG}}$, in which case, by the Davis--Kahan theorem, the top eigenvector of $A_H$ (as well as its sign rounding) achieves constant correlation with $\xopt$.\footnote{As $\xopt$ and $-\xopt$ are indistinguishable for even $k$, we assume for simplicity that the top eigenvector is aligned with $\xopt$.} The sign-rounded top eigenvector is then a warm start with correlation far above the critical correlation $r_N$. Consequently, by the worst-case argument in \cref{subsec:sketch-reverse-planting}, Glauber dynamics achieves strong recovery within $O(\log N)$ time.

As before, despite being a global statistic, $A_H$ can be approximated from the trajectory of high temperature Glauber dynamics, now via the ergodic second moment $\Sigmabar_T\doteq\frac{1}{T}\int_0^T\tilX_t\tilX_t^\top\,dt$. Indeed, by the ergodic theorem, $\Sigmabar_T$ converges to the second moment of $\mu_{H,\beta_0}$ as $T\to\infty$, and expanding to first order in $\beta_0$ as in \cref{subsec:sketch-ergodic-averaging}, we infer the following from $\E_\pi[xx^\top]=I$ and Stein's lemma.
\begin{align*}
    \lim_{T\to\infty}\Sigmabar_T
    =\frac{\E_\pi\bigl[xx^\top e^{\beta_0H(x)}\bigr]}{\E_\pi\bigl[e^{\beta_0H(x)}\bigr]}
    =I+\beta_0\E_\pi\bigl[(xx^\top-I)H(x)\bigr]+O(\beta_0^2)
    =I+\beta_0A_H+O(\beta_0^2).
\end{align*}
Since shifting by the identity does not change eigenvectors, the top eigenvector of $\Sigmabar_T$ thus approximates that of $A_H$, as we show in \cref{sec:ergodic-averaging}.

However, by symmetry, $A_H$ vanishes identically when $k$ is odd, just as $u_H$ does when $k$ is even. To obtain an algorithm that is agnostic to the parity of $k$, \cref{alg:annealed-glauber-ergodic-second-moments} instead uses the ergodic second moment of $F(x)\doteq(x,1)\in\R^{N+1}$, whose blocks comprise both $\Sigmabar_T$ and the ergodic average $\Xbar_T$. Combining both high temperature expansions, we thus expect that for large $T$ and small $\beta_0$,
\begin{align*}
    \bar M_T\doteq\frac{1}{T}\int_0^TF(\tilX_t)^{\otimes2}\,dt
    =\begin{pmatrix}
        \Sigmabar_T & \Xbar_T\\
        \Xbar_T^\top & 1
    \end{pmatrix}
    \approx I_{N+1}+\beta_0\cA_H,
    \qquad
    \cA_H\doteq\begin{pmatrix}
        A_H & u_H\\
        u_H^\top & 0
    \end{pmatrix}.
\end{align*}
When $k$ is odd, $A_H=0$, so the top eigenvector of $\cA_H$, when oriented such that its last coordinate is nonnegative, is $\frac{1}{\sqrt2}(u_H/\norm{u_H},1)$. The sign rounding of its first $N$ coordinates is then $\sgn(u_H)$, i.e., the warm start of \cref{alg:averaged-annealed-glauber}, to which the analysis in \cref{subsec:sketch-reverse-planting} applies. When $k$ is even, $u_H=0$, so $\cA_H$ is block diagonal, and its top eigenvector is obtained by appending a zero to that of $A_H$, which recovers the warm start for even $k$. Thus, the same top eigenvector computation yields a sufficiently good warm start for either parity, and \cref{alg:annealed-glauber-ergodic-second-moments} is agnostic to the parity of $k$. This completes our proof sketch of \cref{thm:glauber-every-k-informal}.

%% file: preliminaries.tex
\section{Notation and Preliminaries}
\label{sec:preliminaries}
For $d\geq2$, we write $M_d(\R)$ for the space of real $d\times d$ matrices and $\Sym_d\subseteq M_d(\R)$ for the subspace of real symmetric matrices. For $M\in M_d(\R)$, we write $\diag(M)$ for the diagonal matrix with the same diagonal entries as $M$, and define $\operatorname{off}(M)=M-\diag(M)$ as the off-diagonal matrix of $M$. For a vector $a\in\R^d$, we also use $\diag(a)$ to denote the diagonal matrix with entries $a_1,\ldots,a_d$. For a symmetric matrix $A\in\Sym_d$, we denote its eigenvalues in nonincreasing order by $\lambda_1(A)\geq\lambda_2(A)\geq\cdots\geq\lambda_d(A)$.

For $\ell\geq1$, we use $(\R^d)^{\otimes\ell}$ to denote the space of order-$\ell$ tensors, viewed as multilinear forms on $\R^d$, and $(\R^d)^{\odot\ell}$ to denote the subspace of symmetric tensors.

Unless specified otherwise, $\norm{\cdot}$ denotes the Euclidean norm, and $\bS^{d-1}$ denotes the unit sphere in $\R^d$. For any matrix $A\in M_d(\R)$, its operator norm is defined by $\norm{A}_{\mathrm{op}}=\sup_{x\in\bS^{d-1}}\norm{Ax}$. This definition extends naturally to tensors, wherein, for $T\in(\R^d)^{\otimes\ell}$, we define $\norm{T}_{\mathrm{op}}=\sup_{x_1,\ldots,x_\ell\in\bS^{d-1}}|T(x_1,\ldots,x_\ell)|$. For a vector $a\in\R^d$, we define its sign rounding $\sgn(a)\in\{\pm1\}^d$ coordinatewise by $\sgn(a)_i=+1$ if $a_i\ge0$ and $\sgn(a)_i=-1$ otherwise.

We use $c,C>0$ to denote universal constants and $c_k,C_k>0$ to denote constants depending only on $k$, whose values may change from line to line unless fixed explicitly. For nonnegative quantities $a$ and $b$, we write $a\lesssim b$ or $a=O(b)$ if $a\le Cb$, and $a\gtrsim b$ or $a=\Omega(b)$ if $a\ge cb$. We write $a\asymp b$ or $a=\Theta(b)$ if both hold, i.e., if $cb\le a\le Cb$. We write $a\ll b$ or $a=o(b)$ if $a/b\to0$ as $N\to\infty$, and $a\gg b$ if $b/a\to0$. Wherever appropriate, we use subscripts to indicate the parameters on which the constants depend, e.g., $a\lesssim_kb$ means $a\le C_kb$.

\subsection{Analysis on the Boolean Hypercube}
\label{subsec:boolean-hypercube}
We present a brief introduction to analysis on the Boolean hypercube, and refer the readers to \citet{ryan2014aobf} for a complete treatment. Let $\Cube = \{\pm 1\}^N$ denote the signed Boolean hypercube, and let $\pi$ denote the uniform measure on $\Cube$. For $x,y\in\Cube$, we define their Hamming distance by $d_H(x,y)=|\{i\in[N]:x_i\ne y_i\}|$. Equivalently, $d_H(x,y)=\nicefrac{1}{4}\norm{x-y}^2=\nicefrac{1}{2}(N-\dotp{x}{y})$. We use $\cB_H(x,r)$ to denote the normalized Hamming ball of radius $r$ centered at $x\in\Cube$, i.e.,
\begin{align*}
    \cB_H(x,r)\doteq\{y\in\Cube:N^{-1}d_H(x,y)\le r\}.
\end{align*}

For each $S\subseteq[N]$, we define the Fourier character $\chi_S(x)=\prod_{i\in S}x_i$, where $\chi_\varnothing=1$. These characters form an orthonormal basis of $L^2(\pi)$ under the inner product $\dotp{f}{g}_{L^2(\pi)}=\E_\pi[f(x)g(x)]$. Consequently, every $f:\Cube\to\R$ admits the Fourier--Walsh expansion
\begin{align*}
    f(x)=\sum_{S\subseteq[N]}\widehat f(S)\chi_S(x),
    \qquad
    \widehat f(S)=\E_\pi[f(x)\chi_S(x)].
\end{align*}

For $x\in\Cube$ and $i\in[N]$, let $x^{\sim i}$ denote $x$ with its $i$th coordinate flipped, i.e., $x^{\sim i}_j=(-1)^{\ind{j=i}}x_j$ and for $f:\Cube\to\R$, we define its discrete partial derivative as follows:
\begin{align*}
    \partial_i f(x)
    =\frac{1}{2}\left[f(x_{\neg i},x_i=+1)-f(x_{\neg i},x_i=-1)\right], 
\end{align*}
where $x_{\neg i}$ denotes all coordinates other than $i$. For vector, matrix and tensor valued functions, the partial derivatives are naturally defined elementwise. Moreover, for any $f : \R^N \to \R$, the discrete gradient $\dgrad f(x) \in \R^N$ is defined as $\dgrad f(x)_i \doteq \partial_i f(x)$. More generally, the discrete derivative tensor of order $\ell \ge 1$ is defined as follows:
\begin{align*}
    [\dgrad^\ell f(x)]_{i_1,\ldots,i_\ell}
    =\partial_{i_1}\cdots\partial_{i_\ell}f(x).
\end{align*}
In particular, $\dgrad^2 f$ denotes the discrete Hessian. The definitions naturally extend to tensor-valued functions as per the standard conventions used for Euclidean differential operators (e.g., for $f : \Cube \to R^m$, $\dgrad f \in \R^{N \times m}$ with $(\dgrad f)_{ij} = \partial_i f_j$, and so on) 
 
The discrete partial derivatives act on Fourier characters in the following fashion:
\begin{align}
\label{eq:partial-fourier-character}
    \partial_i\chi_S(x)
    =\begin{cases}
        \chi_{S\setminus\{i\}}(x), & i\in S,\\
        0, & i\notin S.
    \end{cases}
\end{align}

By linearity and the Fourier-Walsh expansion, we conclude that the discrete partial derivatives always commute and $\partial_i^2 f=0$. Hence, for every $\ell\ge1$, the tensor $\dgrad^\ell f$ is symmetric and its entries vanish whenever two indices coincide.

\begin{lemma}[Discrete product rule]
\label{lem:discrete-product-rule}
For any $f,g:\Cube\to\R$, $i\in[N]$, and $x\in\Cube$,
\begin{align*}
    \partial_i(fg)(x)
    =f(x)\partial_i g(x)+g(x)\partial_i f(x)
    -2x_i\partial_i f(x)\partial_i g(x).
\end{align*}
\end{lemma}

\begin{proof}
Apply the coordinate flip identity to $f$, $g$, and $fg$ to obtain
\begin{align*}
    2x_i\partial_i(fg)(x)
    &=f(x)g(x)-f(x^{\sim i})g(x^{\sim i})\\
    &=f(x)g(x)
    -\bigl(f(x)-2x_i\partial_i f(x)\bigr)
     \bigl(g(x)-2x_i\partial_i g(x)\bigr)\\
    &=2x_i\bigl(f(x)\partial_i g(x)+g(x)\partial_i f(x)\bigr)
      -4\partial_i f(x)\partial_i g(x).
\end{align*}

Dividing by $2x_i$ and using $x_i^2=1$, we obtain the desired identity.
\end{proof}

The following lemma is a Boolean hypercube analogue of the classical Stein's lemma for Gaussian random variables \citep{chatterjee2014survey} and will be used often in our proofs.
\begin{lemma}[Stein's Lemma for the hypercube]
\label{lem:hypercube-stein}
For any $f:\Cube\to\R$,
\begin{align*}
    \E_\pi[\dgrad f(x)]&=\E_\pi[xf(x)],\\
    \E_\pi[\dgrad^2 f(x)]
    &=\frac{1}{2}\E_\pi\left[x\dgrad f(x)^\top+\dgrad f(x)x^\top\right]
    =\E_\pi\left[(xx^\top-I)f(x)\right].
\end{align*}
\end{lemma}

\begin{proof}
Since $\pi$ is a product measure we obtain the following by conditioning on $x_{\neg i}$ for some $i \in [N]$
\begin{align*}
    \E_\pi[x_i f(x)]
    =\E_\pi\left[\nicefrac{1}{2}f(x_{\neg i},x_i=+1)
                  -\nicefrac{1}{2}f(x_{\neg i},x_i=-1)\right]
    =\E_\pi[\partial_i f(x)].
\end{align*}

which proves the first identity elementwise. Applying it to $\partial_j f$, we obtain $\E_\pi[\partial_i\partial_j f]=\E_\pi[x_i\partial_j f]$. Since discrete partial derivatives commute, we infer the following by exchanging $i$ and $j$ and averaging:
\begin{align*}
    \E_\pi[\partial_i\partial_j f(x)]
    =\frac{1}{2}\E_\pi[x_i\partial_j f(x)+x_j\partial_i f(x)].
\end{align*}
For distinct $i,j \in [N]$, we have $\partial_j(x_i f)=x_i\partial_j f$. Applying the first identity once more, we obtain
\begin{align*}
    \E_\pi[\partial_i\partial_j f(x)]
    =\E_\pi[\partial_j(x_i f)(x)]
    =\E_\pi[x_i x_j f(x)].
\end{align*}

For $i=j$, both $\partial_i^2 f$ and $(x_i^2-1)f$ vanish, and $\E_\pi[x_i\partial_i f]=0$ because $\partial_i f$ does not depend on $x_i$. Combining these identities entrywise proves the second identity.
\end{proof}

We use $O(N,\Z)$ to denote the matrix group of signed permutations, i.e., the orthogonal matrices with integer entries:
\begin{align*}
    O(N,\Z)=\left\{DP:\ D=\diag(e),\ e\in\Cube,\
    P\text{ is a permutation matrix}\right\}.
\end{align*}
We observe that $O(N,\Z)$ is a subgroup of $O(N)$ which acts bijectively on $\Cube$. In fact, $O(N,\Z)$ is the isometry group of $(\Cube,d_H)$. Consequently, $\pi$ is invariant under the action of $O(N,\Z)$. The following lemma characterizes the structure of matrices that commute with $O(N,\Z)$.
\begin{lemma}[Commutants of $O(N,\Z)$]
\label{lem:commutants-bn}
Any matrix that commutes with every element of $O(N,\Z)$ must be a scalar multiple of the identity.
\end{lemma}

\begin{proof}
Let $A\in\R^{N\times N}$ satisfy $AM=MA$ for every $M\in O(N,\Z)$. Consider distinct $i,j\in[N]$, and let $D$ be a diagonal sign matrix with $D_{ii}=+1$ and $D_{jj}=-1$. Since $D\in O(N,\Z)$, we have $AD=DA$, and hence
\begin{align*}
    (AD)_{ij}=A_{ij}D_{jj}=-A_{ij},
    \qquad
    (DA)_{ij}=D_{ii}A_{ij}=A_{ij}.
\end{align*}
It follows that $A_{ij}=0$ for $i\ne j$, so $A=\diag(a)$ for some $a\in\R^N$. Now let $P$ be the permutation matrix that exchanges coordinates $i$ and $j$. Since $P\in O(N,\Z)$, we have $A=PAP^\top$, and therefore $a_i=a_j$. Thus $a_1=\cdots=a_N$ and $A=cI$ for some $c\in\R$.
\end{proof}
Lastly, we shall make use of the following result on the higher moments of $\pi$
\begin{lemma}[Moments of Rademacher sums]
\label{lem:rademacher-moments}
Let $N\ge2$, $x\sim\pi$, and $S_N=\sum_{i=1}^N x_i$. For every fixed integer $\ell\ge1$, we have
\begin{align*}
    \E_\pi[S_N^\ell]
    =\begin{cases}
        0, & \text{if $\ell$ is odd},\\
        (\ell-1)!!\,N^{\nicefrac{\ell}{2}}+O_\ell(N^{\nicefrac{\ell}{2}-1}), & \text{if $\ell$ is even}.
    \end{cases}
\end{align*}
Moreover, for any distinct $i,j\in[N]$,
\begin{align*}
    \E_\pi[x_ix_jS_N^\ell]
    =\begin{cases}
        0, & \text{if $\ell$ is odd},\\
        \ell(\ell-1)!!\,N^{\nicefrac{\ell}{2}-1}+O_\ell(N^{\nicefrac{\ell}{2}-2}), & \text{if $\ell$ is even}.
    \end{cases}
\end{align*}
\end{lemma}

\begin{proof}
Let $B_N\sim\bin(N,\nicefrac{1}{2})$. Since $S_N\disteq2(B_N-\nicefrac{N}{2})$, the first claim follows from the standard centered moment formulas for binomial random variables \citep[Theorem~3]{skorski2025handy}. For the second claim, we first observe that $\E_\pi[x_ix_jS_N^\ell]=0$ for odd $\ell$ by symmetry. For even $\ell$, the identity
\begin{align*}
    S_N^2=N+2\sum_{i<j}x_ix_j
\end{align*}
and exchangeability imply the following
\begin{align*}
    \E_\pi[S_N^{\ell+2}]
    =N\E_\pi[S_N^\ell]+N(N-1)\E_\pi[x_ix_jS_N^\ell].
\end{align*}
Rearranging and applying the first claim, we obtain
\begin{align*}
    \E_\pi[x_ix_jS_N^\ell]
    &=\frac{\E_\pi[S_N^{\ell+2}]-N\E_\pi[S_N^\ell]}{N(N-1)}=\frac{1}{N(N-1)}\left[
        \bigl((\ell+1)!!-(\ell-1)!!\bigr)N^{\nicefrac{\ell}{2}+1}
        +O_\ell(N^{\nicefrac{\ell}{2}})
    \right]\\
    &=\ell(\ell-1)!!\,N^{\nicefrac{\ell}{2}-1}+O_\ell(N^{\nicefrac{\ell}{2}-2}).
\end{align*}
\end{proof}

\subsection{Discrete Sobolev Norms for Spin Glass Hamiltonians}
\label{subsec:discrete-sobolev-norms}
We now define a discrete Sobolev Norm for functions on $\Cube$. This extends the notion of the Sobolev $\cG$-norm for smooth functions on the sphere, introduced by Ben Arous, Gheissari, and Jagannath \citep{ben2020bounding,arous2020algorithmic}.

\begin{definition}[Discrete Sobolev Norm]
\label{def:discrete-sobolev-norm}
Fix an integer $n\geq0$. For any function $f:\Cube\to\R$, its Sobolev $\cD^n$ norm is defined as follows:
\begin{align}
\label{eq:discrete-sobolev-norm}
    \norm{f}_{\cD^n} &\doteq \sum_{\ell=0}^n N^{\nicefrac{\ell}{2}}\max_{x\in\Cube}\norm{\dgrad^\ell f(x)}_{\mathrm{op}},
\end{align}
where $\norm{\dgrad^0 f(x)}=|f(x)|$ by convention.
\end{definition}
The following lemma controls the discrete Sobolev norms of the pure $k$-spin Hamiltonian. Its proof is presented in \cref{app:kspin-sobolev-proof}.
\begin{lemma}[Sobolev norm concentration for the Ising pure $k$-spin Hamiltonian]
\label{lem:kspin-sobolev-norm-concentration}
Let $H_0$ be the Ising pure $k$-spin Hamiltonian defined in \cref{eq:tensor-pca-overlap-noise}. For any integer $n\ge0$ and every $t\ge0$, we have
\begin{align*}
    \Prb\left[
        \norm{H_0}_{\cD^n}
        \ge C_{k,n}N\left(1+\sqrt{\frac{t}{N}}\right)
    \right]
    \le \exp\left(-c_{k,n}(N+t)\right).
\end{align*}
\end{lemma}
We next control the contribution of the signal term by bounding the discrete Sobolev norms of the overlap and its powers. The proof is presented in \cref{app:signal-sobolev-proof}.
\begin{lemma}[Sobolev norm bounds for the signal term]
\label{lem:overlap-sobolev-bounds}
Let $m$ denote the correlation function defined in \cref{eq:tensor-pca-overlap-noise}. For any integer $n\ge0$, we have
\begin{align*}
    \norm{m}_{\cD^n}\le2,
    \qquad
    \norm{m^k}_{\cD^n}\le C_{k,n}.
\end{align*}
\end{lemma}

\begin{remark}
\label{rem:tensor-pca-sobolev-bound}
By \cref{lem:kspin-sobolev-norm-concentration,lem:overlap-sobolev-bounds}, , the tensor PCA Hamiltonian satisfies $\norm{H}_{\cD^n}\le C_{k,n}N(1+\lambda)$ for any fixed integer $n\ge0$ with probability at least $1-\exp(-c_{k,n}N)$. 
\end{remark}

\subsection{Continuous-Time Glauber Dynamics}
\label{subsec:cts-glauber}
We briefly introduce continuous-time single-site dynamics, a class of stochastic processes on $\Cube$ that includes Glauber dynamics. Given a family of spin update kernels $q( \cdot \mid i,t,x)_{i \in [N], t\geq 0, x \in \Cube}$ on $\{\pm 1\}$, the trajectory $(X_t)_{t \geq 0}$ of single-single dynamics evolves as follows: Let $(T_n)_{n \ge 1}$ denote the arrival times of a rate $N$ Poisson process \citep[Section~1.1]{bjork2021point}. At time $T_n$, the process independently samples a coordinate $I_n\sim\Uniform([N])$ and a spin $Z_n\sim q(\cdot\mid I_n,T_n,X_{T_n^-})$ from the conditional site update distribution, which is a probability measure on $\{\pm 1\}$. The dynamics then updates by setting $X_{T_n,I_n} \doteq Z_n$ and leaving the remaining coordinates unaffected. The continuous-time trajectory $(X_t)_{t \geq 0}$ is constant between arrival times. 

The updates of single-site dynamics form a marked Poisson point process \citep[Sections~7.2,~7.5]{bjork2021point}, where $(I_n,Z_n)$ marks the arrival $T_n$ and $q(z\mid i,t,X_{t^-})$ is the intensity of mark $(i,z)$. We use the phrase \emph{filtration of single-site dynamics up to time $T$} to denote the filtration generated by $X_0$ and $(T_n,I_n,Z_n)_{n:T_n\le t}$ for $0\le t\le T$, and the phrase \emph{law of single-site dynamics up to time $T$} to denote the joint law of $X_0$ and $(T_n,I_n,Z_n)_{n:T_n\le T}$. Observe that the trajectory $(X_t)_{t \geq 0}$ is a deterministic function of the initialization, arrival times and marks.
Furthermore, since the trajectory is piecewise constant and the Poisson interarrival times are i.i.d. exponentially distributed \citep[cf.][Section~1.1]{bjork2021point}, one can simulate the process exactly up to time $T$ in $O(N_T)$ steps by successively sampling the arrival times and marks, where $N_T\sim\Poisson(NT)$ denotes the number of arrivals up to time $T$.

Given a Hamiltonian $V:\Cube\to\R$ and an inverse temperature $\beta\ge0$, Glauber dynamics (\cref{alg:continuous-time-glauber}) is a specific instance of single-site dynamics associated with the site update kernel $q(z\mid i,t,x) = \nicefrac{1}{2}(1+z\tanh(\beta \partial_i V(x)))$.
The resulting trajectory $(X_t)_{t \geq 0}$ is a reversible Markov chain whose stationary distribution is the Gibbs measure $\mu_{V,\beta}(x)\propto\exp(\beta V(x))$ \citep[Section~3.3]{levin2017markov}.
\begin{algorithm}[!t]
\caption{Continuous-time Glauber dynamics}
\label[algorithm]{alg:continuous-time-glauber}
\textbf{Input:} Hamiltonian $V:\Cube\to\R$, inverse temperature $\beta\ge0$, initial configuration $X_0\in\Cube$, and simulation time $T>0$.
\begin{algorithmic}[1]
    \Statex \textbf{Simulating Poisson process arrival times.}
    \State Set $n=0$ and $T_0=0$.
    \While{$T_n\le T$}
        \State Set $T_{n+1}=T_n+E_n$, $E_n\iidsim\operatorname{Exp}(N)$, and increment $n\gets n+1$.
    \EndWhile
    \State Set $N_T=\sup\{n\ge0:T_n\le T\}$, the number of arrivals up to time $T$.
    \Statex \textbf{Simulating the trajectory.}
    \For{$n=1,\ldots,N_T$}
        \State Set $X_t=X_{T_{n-1}}$ for $T_{n-1}\le t<T_n$.
        \State Sample $I_n\sim\Uniform([N])$ and then $Z_n\sim q(\cdot\mid I_n,T_n,X_{T_n^-})$, where
        \begin{align*}
            q(z\mid i,t,x)
            =\frac{\exp\bigl(\beta V(x_{\neg i},x_i=z)\bigr)}
            {\sum_{b\in\{\pm1\}}\exp\bigl(\beta V(x_{\neg i},x_i=b)\bigr)},
            \qquad z\in\{\pm1\}.
        \end{align*}
        \State Set $X_{T_n,I_n}\doteq Z_n$ and $X_{T_n,j}\doteq X_{T_n^-,j}$ for $j\ne I_n$.
    \EndFor
    \State Set $X_t=X_{T_{N_T}}$ for $T_{N_T}\le t\le T$.
\end{algorithmic}
\textbf{Output:} Trajectory $(X_t)_{0\le t\le T}$.
\end{algorithm}
The associated Markov semigroup $P_t^{(V,\beta)}$ describes the evolution of scalar observables along the trajectory. Specifically, for any $f : \Cube \to \R$,
\begin{align*}
    P_t^{(V,\beta)}f(x)=\E[f(X_t)\mid X_0=x].
\end{align*}
The generator $\cL_{V,\beta}$, which characterizes the infinitesimal action of the semigroup, is defined as follows for any $f : \Cube \to \R$.
\begin{align}
\label{eq:glauber-generator}
    \begin{aligned}
        \cL_{V,\beta}f(x)
        &=\lim_{t\to0^+}\frac{P_t^{(V,\beta)}f(x)-f(x)}{t}
        =\sum_{i=1}^N c_i^{(V,\beta)}(x)\bigl[f(x^{\sim i})-f(x)\bigr],\\
        c_i^{(V,\beta)}(x)
        &=\frac12\bigl[1-x_i\tanh(\beta\partial_iV(x))\bigr].
    \end{aligned}
\end{align}
where $c_i^{(V,\beta)}(x)$ denotes the flip rate of coordinate $i$ in configuration $x$. The Markov semigroup and its associated generator also satisfies the Kolmogorov equations \citep[Section~20.1, Eq.~(20.7)]{levin2017markov}
\begin{align}
\label{eq:glauber-kolmogorov}
    \frac{d}{dt}P_t^{(V,\beta)}f
    =\cL_{V,\beta}P_t^{(V,\beta)}f
    =P_t^{(V,\beta)}\cL_{V,\beta}f. \qquad \forall \ f \in \Cube \to \R
\end{align}
Intuitively, for any observable $f:\Cube\to\R$, $\cL_{V,\beta}f(X_t)$ represents instantaneous expected rate of change of $f(X_t)$ along the trajectory. This interpretation is made rigorous via Dynkin's formula~\citep[cf.][Chapter~4, Proposition~1.7]{ethier1986markov}.

\begin{lemma}[Dynkin's Formula for Continuous-time Glauber dynamics]
\label{lem:single-site-dynkin}
Let $(X_t)_{t\ge0}$ denote the trajectory of Glauber dynamics for the Hamiltonian $V$ at inverse temperature $\beta$. For any $f:\Cube\to\R$, the process
\begin{align*}
    M_t^f=f(X_t)-f(X_0)-\int_0^t\cL_{V,\beta}f(X_s)\,ds
\end{align*}
is a square integrable martingale with $M_0^f=0$, known as the Dynkin martingale. Its predictable quadratic variation is
\begin{align}
\label{eq:single-site-predictable-bracket}
    \langle M^f\rangle_t
    =2\int_0^t\Gamma_{V,\beta}(f,f)(X_s)\,ds,
\end{align}
where $\Gamma_{V,\beta}$ denotes the carr\'e du champ operator associated with $\cL_{V,\beta}$, defined by
\begin{align*}
    \Gamma_{V,\beta}(f,g)(x)
    \doteq\frac12\bigl[\cL_{V,\beta}(fg)-f\cL_{V,\beta}g-g\cL_{V,\beta}f\bigr](x)
    =2\sum_{i=1}^N c_i^{(V,\beta)}(x)\partial_i f(x)\partial_i g(x).
\end{align*}
\end{lemma}
The following uniform concentration guarantee for Dynkin martingales shall be used often in our proofs. 

\begin{lemma}[Concentration of Dynkin martingales]
\label{lem:dynkin-martingale-concentration}
Let $(X_t)_{t\ge0}$ denote the trajectory of Glauber dynamics for the Hamiltonian $V$ at inverse temperature $\beta$. Suppose that $f:\Cube\to\R$ satisfies $\norm{f}_{\cD^1}\le B$ and $\max_{x\in\Cube}\norm{\dgrad f(x)}_\infty\le L$ for some $B,L>0$, and let $M_t^f$ denote its Dynkin martingale. For every $T,r>0$, we have
\begin{align*}
    \Prb\left[\sup_{0\le t\le T}|M_t^f|\ge r\right]
    \le2\exp\left(-\frac{cNr^2}{B^2T+NLr}\right),
\end{align*}
where $c>0$ is a universal constant.
\end{lemma}

\begin{proof}
By \cref{def:discrete-sobolev-norm}, $\max_{x\in\Cube}\norm{\dgrad f(x)}\le \nicefrac{B}{\sqrt N}$. Moreover, since $c_i^{(V,\beta)}(x)\le1$, we infer from \cref{lem:single-site-dynkin} that
\begin{align*}
    \langle M^f\rangle_T
    =4\int_0^T\sum_{i=1}^N c_i^{(V,\beta)}(X_s)\bigl(\partial_i f(X_s)\bigr)^2\,ds
    \le\frac{4B^2T}{N}.
\end{align*}
Moreover, each jump of $M^f$ is a coordinate flip of $f$, so
\begin{align*}
    |\Delta M_t^f|
    \le\max_{x\in\Cube,\ i\in[N]}|f(x^{\sim i})-f(x)|
    =2\max_{x\in\Cube,\ i\in[N]}|\partial_i f(x)|
    \le2L.
\end{align*}
The proof is concluded by applying a standard Bernstein bound for continuous-time jump martingales \citep[cf.][Corollary~11(b)]{howard2020time}.
\end{proof}

\begin{remark*}
By \cref{def:discrete-sobolev-norm}, $\max_{x\in\Cube}\norm{\dgrad f(x)}_\infty\le N^{-1/2}\norm{f}_{\cD^1}$. Hence \cref{lem:dynkin-martingale-concentration} may always be applied with $L=B/\sqrt N$, at the cost of a weaker concentration guarantee when a smaller bound on the discrete partial derivatives is available.
\end{remark*}

The Markov semigroup $P_t^{(V,\beta)}$ and the generator $\cL_{V,\beta}$ act entrywise on vector and matrix-valued observables. Consequently, the corresponding vector (resp. matrix) valued Dynkin martingale is obtained by applying \cref{lem:single-site-dynkin} to each entry. To this end, we obtain the following uniform concentration guarantee for operator norms of rectangular matrix valued Dynkin martingales via a direct application of~\cref{lem:dynkin-martingale-concentration}.

\begin{lemma}[Concentration of matrix-valued Dynkin martingales]
\label{lem:matrix-dynkin-uniform-concentration}
Let $m,n\ge1$, and let $F:\Cube\to\R^{m\times n}$ be any function satisfying $\norm{u^\top F(\cdot)v}_{\cD^1}\le B\sqrt N$ for every $u\in\bS^{m-1}$, $v\in\bS^{n-1}$ and some $B>0$. Suppose also that $\max_{x\in\Cube,\ i\in[N]}\norm{\partial_i F(x)}_{\mathrm{op}}\le L$ for some $L>0$. Let $(X_t)_{t\ge0}$ denote the trajectory of Glauber dynamics for the Hamiltonian $V$ at inverse temperature $\beta$, and let $M_t^F$ denote the associated Dynkin martingale. For every $T>0$ and $\delta\in(0,1)$, the following holds with probability at least $1-\delta$:
\begin{align*}
    \sup_{0\le t\le T}\norm{M_t^F}_{\mathrm{op}}
    \le C\left[
        B\sqrt{T\bigl(m+n+\log(\nicefrac{1}{\delta})\bigr)}
        +L\bigl(m+n+\log(\nicefrac{1}{\delta})\bigr)
    \right],
\end{align*}
where $C>0$ is a universal constant. In particular, taking $n=1$ yields the same guarantee for vector-valued observables, with the Euclidean norm in place of the operator norm.
\end{lemma}

\begin{proof}
Fix $u\in\bS^{m-1}$ and $v\in\bS^{n-1}$, and set $f_{u,v}(x)=u^\top F(x)v$. By linearity of the generator, $u^\top M_t^Fv=M_t^{f_{u,v}}$. Since $\norm{f_{u,v}}_{\cD^1}\le B\sqrt N$ and $\max_{x\in\Cube}\norm{\dgrad f_{u,v}(x)}_\infty\le\max_{x\in\Cube,\ i\in[N]}\norm{\partial_i F(x)}_{\mathrm{op}}\le L$, we obtain from \cref{lem:dynkin-martingale-concentration} that, for every $s>0$,
\begin{align*}
    \Prb\left[
        \sup_{0\le t\le T}|u^\top M_t^Fv|
        \ge C\bigl(B\sqrt{Ts}+Ls\bigr)
    \right]
    \le2e^{-s}.
\end{align*}
Let $\cN_m$ and $\cN_n$ be fixed $\nicefrac{1}{4}$-nets of $\bS^{m-1}$ and $\bS^{n-1}$ with $|\cN_m|\le9^m$ and $|\cN_n|\le9^n$. For every $A\in\R^{m\times n}$, a standard $\eps$-net comparison implies
\begin{align*}
    \norm{A}_{\mathrm{op}}
    \le2\max_{u\in\cN_m,\,v\in\cN_n}|u^\top Av|.
\end{align*}
The proof is concluded by taking $s=(m+n)\log9+\log(\nicefrac{2}{\delta})$, applying a union bound over $\cN_m\times\cN_n$, and adjusting constants appropriately.
\end{proof}

\begin{remark}
\label{rem:symmetric-matrix-dynkin-concentration}
If $m=n$ and $F$ takes values in $\Sym_n$, it suffices to replace the Sobolev norm bound by $\norm{w^\top F(\cdot)w}_{\cD^1}\le B\sqrt N$ for every $w\in\bS^{n-1}$. 
In either case, the Sobolev norm bound implies $\norm{\partial_i F(x)}_{\mathrm{op}}\le B$ for every $i \in [N]$ and $x \in \Cube$, so one may always take $L=B$.
\end{remark}

Finally, the following change of measure theorem for single-site dynamics is a jump process analogue of the classical Girsanov theorem for diffusion processes \citep[Section~6.4]{stroock2007multidimensional}. We refer the reader to \citet[cf. Lemma~2.2 and Theorem~2.4]{sokol2015exponential} for a detailed proof.

\begin{lemma}[Change of Measure for Single-Site Dynamics]
\label{lem:single-site-change-measure}
Fix $T>0$, and let $\Q^T$ denote the law of single-site dynamics up to time $T$, with spin update kernel $q$ and trajectory $(X_t)_{0\le t\le T}$. Let $a_i(z,t)$ denote tilts that are predictable with respect to the filtration of single-site dynamics up to time $T$, satisfying $0<a_i(z,t)\le C$ for a deterministic $C<\infty$ and all $i\in[N]$, $z\in\{\pm1\}$, and $t\le T$. For $0\le t\le T$, define
\begin{align*}
    L_t^a
    =\exp\left(\sum_{n:T_n\le t}\log a_{I_n}(Z_n,T_n)
    -\int_0^t\sum_{i=1}^N\sum_{z\in\{\pm1\}}
    q(z\mid i,s,X_{s^-})\bigl(a_i(z,s)-1\bigr)\,ds\right).
\end{align*}
Then $(L_t^a)_{0\le t\le T}$ is a mean-one martingale under $\Q^T$. Define the tilted law $\Q^{T,a}$ by $\frac{d\Q^{T,a}}{d\Q^T}=L_T^a$. Under $\Q^{T,a}$, updates with mark $(i,z)$ have intensity $a_i(z,t)q(z\mid i,t,X_{t^-})$ for $t\le T$, and $X_0$ has the same distribution as under $\Q^T$.
\end{lemma}

%% file: averaged_differentials.tex
\section{Averaged Discrete Differentials of the Tensor PCA Hamiltonian}
\label{sec:averaged-discrete-differentials}

In this section, we analyze the statistical properties of the average discrete gradient and Hessian of the tensor PCA Hamiltonian. In \cref{subsec:average-discrete-gradient}, we show that, for odd $k$, the average discrete gradient is a noisy observation of the planted vector. Analogously, in \cref{subsec:average-discrete-hessian}, we show that, for even $k$, the average discrete Hessian has a spiked Wigner type structure, with a rank-one spike in the planted direction.

\subsection{Analysis of the Average Discrete Gradient}
\label{subsec:average-discrete-gradient}

We now establish a signal-noise decomposition of the average discrete gradient of the tensor PCA Hamiltonian.

\begin{theorem}[Signal-Noise Decomposition of the Average Discrete Gradient]
\label{avg-discrete-grad-snr}
Let $H$ denote the tensor PCA Hamiltonian defined in \cref{eq:tensor-pca-hamiltonian}, and define the average gradient $u_H\doteq\E_\pi[\dgrad H(x)]$. Then, the following holds:
\begin{enumerate}
    \item If $k$ is even, then $u_H=0$ holds deterministically.
    \item If $k$ is odd, then $u_H=N\lambda\alpha_N\plant+\zeta$, where $\zeta\doteq\E_\pi[xH_0(x)]\sim\cN(0,\sigma_N^2I)$, $\alpha_N\asymp_k N^{-\nicefrac{k}{2}}$, and $\sigma_N^2=\alpha_N\sqrt{N}\asymp_k N^{-\tfrac{k-1}{2}}$.
\end{enumerate}
\end{theorem}

\begin{proof}
By Stein's Lemma (\cref{lem:hypercube-stein}) and \cref{eq:tensor-pca-hamiltonian},
\begin{align}
\label{eq:average-gradient-decomposition}
    u_H=\E_\pi[xH(x)]
    =N\lambda\E_\pi[xm(x)^k]+\E_\pi[xH_0(x)].
\end{align}
If $k$ is even, then $H(-x)=H(x)$, and invariance of $\pi$ under $x\mapsto-x$ implies that $u_H=0$.

Henceforth, assume $k$ is odd. We first analyze the signal term in \cref{eq:average-gradient-decomposition}. Let $D=\diag(\xopt)$. Since $D \in O(N,\Z)$, for $x\sim\pi$, the vector $y=Dx$ is also distributed according to $\pi$. Define $\widetilde m(y)=N^{-1}\sum_{i=1}^N y_i$, the overlap with the all-ones vector. We observe that
\begin{align*}
    m(x)
    =\frac{1}{N}\sum_{i=1}^N \xopt_i x_i
    =\frac{1}{N}\sum_{i=1}^N(Dx)_i
    =\widetilde m(y).
\end{align*}
Moreover, $D^2=I$, so $x=Dy$. Consequently,
\begin{align}
\label{eq:average-gradient-signal-transform}
    \E_{x\sim\pi}[xm(x)^k]
    =\E_{y\sim\pi}[Dy\,\widetilde m(y)^k]
    =D\E_{y\sim\pi}[y\,\widetilde m(y)^k].
\end{align}

Now let $P$ be any permutation matrix. Since $Py\sim\pi$ whenever $y\sim\pi$, and $\widetilde m(Py)=\widetilde m(y)$, we have
\begin{align*}
    \E_{y\sim\pi}[y\,\widetilde m(y)^k]
    =\E_{y\sim\pi}[Py\,\widetilde m(Py)^k]
    =P\E_{y\sim\pi}[y\,\widetilde m(y)^k].
\end{align*}
Since $\E_{y\sim\pi}[y\,\widetilde m(y)^k]$ is fixed by every permutation matrix, it is of the form $c\mathbf{1}$ for some $c\in\R$. We conclude from \cref{eq:average-gradient-signal-transform} that
\begin{align*}
    \E_{x\sim\pi}[xm(x)^k]
    =cD\mathbf{1}
    =c\xopt
    =\alpha_N\plant,
\end{align*}
where $\alpha_N=c\sqrt N$. Taking the inner product with the unit vector $\plant$, we compute $\alpha_N$ as follows:
\begin{align}
\label{eq:average-gradient-alpha-moment}
    \alpha_N
    &=\E_{x\sim\pi}[\dotp{x}{\plant}m(x)^k]
    =\sqrt N\,\E_{x\sim\pi}[m(x)^{k+1}]\nonumber\\
    &=\sqrt N\,\E_{y\sim\pi}[\widetilde m(y)^{k+1}]
    =N^{-k-\nicefrac{1}{2}}\E_{y\sim\pi}\left[\left(\sum_{i=1}^N y_i\right)^{k+1}\right].
\end{align}
Since $k+1$ is even, we infer from \cref{lem:rademacher-moments} that
\begin{align*}
    \alpha_N
    =N^{-k-\nicefrac{1}{2}}\left[k!!\,N^{\tfrac{k+1}{2}}+O_k(N^{\tfrac{k-1}{2}})\right]
    \asymp_k N^{-\nicefrac{k}{2}}.
\end{align*}

We now analyze the noise term in \cref{eq:average-gradient-decomposition}. By Stein's Lemma (\cref{lem:hypercube-stein}) and \cref{eq:tensor-pca-overlap-noise},
\begin{align*}
    \zeta
    =\E_\pi[\dgrad H_0(x)]
    =\E_\pi[xH_0(x)]
    =N^{-\tfrac{k-1}{2}}\E_{x\sim\pi}\left[x\dotp{G}{x^{\otimes k}}\right]
    =\Lambda(G),
\end{align*}
where $\Lambda:(\R^N)^{\otimes k}\to\R^N$ is a linear operator. Since $G$ is a centered Gaussian tensor, $\zeta\sim\cN(0,\Sigma)$ for some covariance matrix $\Sigma$.

We now prove that the law of $\zeta$ is invariant under $O(N,\Z)$. Fix $U\in O(N,\Z)\subseteq O(N)$. By rotational invariance of the standard Gaussian, $U^{\otimes k}G\disteq G$. Moreover, $U^\top x\sim\pi$ whenever $x\sim\pi$. Making the appropriate change of variables, we obtain
\begin{align*}
    \Lambda(U^{\otimes k}G)
    &=N^{-\tfrac{k-1}{2}}\E_{x\sim\pi}
      \left[x\dotp{U^{\otimes k}G}{x^{\otimes k}}\right]\\
    &=N^{-\tfrac{k-1}{2}}\E_{x\sim\pi}
      \left[x\dotp{G}{(U^\top x)^{\otimes k}}\right]\\
    &=N^{-\tfrac{k-1}{2}}U\E_{y\sim\pi}
      \left[y\dotp{G}{y^{\otimes k}}\right]
    =U\Lambda(G).
\end{align*}
Together with $U^{\otimes k}G\disteq G$, this implies that
\begin{align*}
    \zeta=\Lambda(G)
    \disteq\Lambda(U^{\otimes k}G)
    =U\Lambda(G)
    =U\zeta.
\end{align*}
Consequently, for every $U\in O(N,\Z)$,
\begin{align*}
    \Sigma
    =\E_G[\zeta\zeta^\top]
    =\E_G[(U\zeta)(U\zeta)^\top]
    =U\E_G[\zeta\zeta^\top]U^\top
    =U\Sigma U^\top.
\end{align*}
Thus, $\Sigma$ commutes with every element of $O(N,\Z)$, and \cref{lem:commutants-bn} implies that $\Sigma=\sigma_N^2I$ for some $\sigma_N^2\ge0$. Finally, to compute $\sigma_N^2$, we take the trace and use \cref{eq:kspin-covariance} to obtain the following.
\begin{align*}
    \sigma_N^2
    &=\frac{1}{N}\Tr(\Sigma)
    =\frac{1}{N}\E_G[\norm{\zeta}^2]\\
    &=\frac{1}{N}\E_G\left[
      \dotp{\E_{x\sim\pi}[xH_0(x)]}{\E_{x'\sim\pi}[x'H_0(x')]}\right]\\
    &=\frac{1}{N}\E_{x,x'\,\iidsim\pi}
      \left[\dotp{x}{x'}\,\E_G[H_0(x)H_0(x')]\right]\\
    &=N^{-k}\E_{x,x'\,\iidsim\pi}\left[\dotp{x}{x'}^{k+1}\right].
\end{align*}
Consider $x,x'\iidsim\pi$ and let $z_i=x_ix'_i$. Then $z\sim\pi$ and $\dotp{x}{x'}=\sum_{i=1}^N z_i$. Comparing with \cref{eq:average-gradient-alpha-moment}, we conclude that
\begin{align}
\label{eq:average-gradient-variance-moment}
    \sigma_N^2
    =N^{-k}\E_{z\sim\pi}\left[\left(\sum_{i=1}^N z_i\right)^{k+1}\right]
    =\alpha_N\sqrt N
    \asymp_k N^{-\tfrac{k-1}{2}}.
\end{align}
\end{proof}

\begin{remark}
\label{rem:average-gradient-snr}
We make the following observations from \cref{avg-discrete-grad-snr}. For even $k$, spin flip symmetry forces $u_H=0$, so the average discrete gradient carries no information about the planted signal. For odd $k$, it is a noisy observation of the planted vector $\plantsign$, with signal-to-noise ratio $\lambda\sigma_N\asymp_k\lambda N^{-\tfrac{k-1}{4}}$. In particular, at the conjectured algorithmic threshold $\lambda\gtrsim_k N^{\tfrac{k-2}{4}}$, the SNR is at least $\Omega(N^{-\nicefrac{1}{4}})$. This is formalized in the following corollary.
\end{remark}

\begin{corollary}[Correlation of the Sign-Rounded Average Discrete Gradient]
\label{cor:avg-discrete-grad-correlation}
Let $k\ge3$ be odd and $\lambda\ge0$, and define $\widehat u_H\doteq\sgn(u_H)$. Then,
\begin{align*}
    \Prb\left[
        c_k\min\left\{1,\lambda N^{-\tfrac{k-1}{4}}\right\}
        \le m(\widehat u_H)
        \le C_k\min\left\{1,\lambda N^{-\tfrac{k-1}{4}}\right\}
    \right]
    \ge1-2\exp\left(-c_k\min\left\{N,\lambda^2N^{-\tfrac{k-3}{2}}\right\}\right).
\end{align*}
In particular, if $\lambda\ge C_kN^{\tfrac{k-2}{4}}$, then, for all sufficiently large $N$,
\begin{align*}
    \Prb\left[m(\widehat u_H)\ge N^{-\nicefrac{1}{4}}\right]
    \ge1-\exp(-c_k\sqrt N).
\end{align*}
\end{corollary}

\begin{proof}
By \cref{avg-discrete-grad-snr} and the identity $\sigma_N^2=\alpha_N\sqrt N$, we have $u_H=\lambda\sigma_N^2\xopt+\zeta$, where $\zeta\sim\cN(0,\sigma_N^2I)$. Let $t_N=\lambda\sigma_N$ and $Z_i=\tfrac{\xopt_i\zeta_i}{\sigma_N}$. Then $Z_1,\ldots,Z_N\iidsim\cN(0,1)$, and the signed coordinates $Y_i=\xopt_i(\widehat u_H)_i$ satisfy, almost surely,
\begin{align*}
    Y_i
    =\xopt_i\sgn\left(\lambda\sigma_N^2\xopt_i+\zeta_i\right)
    =\sgn\left(\lambda\sigma_N^2+\xopt_i\zeta_i\right)
    =\sgn(t_N+Z_i).
\end{align*}
Consequently, the $Y_i$ are i.i.d. biased Rademacher variables and $m(\widehat u_H)=N^{-1}\sum_{i=1}^N Y_i$.

We first bound their common mean. By symmetry of the standard Gaussian density,
\begin{align*}
    \mu_N\doteq\E[Y_i]
    &=\Prb[Z_i\ge-t_N]-\Prb[Z_i<-t_N]
    =2\Prb[0\le Z_i\le t_N]\\
    &=\frac{2}{\sqrt{2\pi}}\int_0^{t_N}e^{-\tfrac{s^2}{2}}\,ds
    \ge\frac{2}{\sqrt{2\pi e}}\min\{t_N,1\},
\end{align*}
where we restrict the integral to $[0,\min\{t_N,1\}]$ and use $e^{-\tfrac{s^2}{2}}\ge e^{-\nicefrac{1}{2}}$ on $[0,1]$. Similarly, bounding the Gaussian density by $\tfrac{1}{\sqrt{2\pi}}$ and using $\mu_N\le1$, we obtain $\mu_N\le\min\{t_N,1\}$. Moreover, Since the $Y_i$ are i.i.d. biased Rademacher variables with mean $\mu_N$, we obtain the following via a Chernoff bound \citep{van2014probability}:
\begin{align*}
    \Prb\left[\left|m(\widehat u_H)-\mu_N\right|\ge\frac{\mu_N}{2}\right]
    \le2\exp\left(-\frac{N\mu_N^2}{8}\right)
    \le2\exp\left(-cN\min\{t_N^2,1\}\right).
\end{align*}
The first claim follows by applying $c_k\min\{t_N,1\}\le\mu_N\le\min\{t_N,1\}$ with $t_N=\lambda\sigma_N\asymp_k\lambda N^{-\tfrac{k-1}{4}}$ and adjusting constants appropriately.

Finally, if $\lambda\ge C_kN^{\tfrac{k-2}{4}}$, then $\lambda N^{-\tfrac{k-1}{4}}\ge C_kN^{-\nicefrac{1}{4}}$ and $\lambda^2N^{-\tfrac{k-3}{2}}\ge C_k^2\sqrt N$. The second claim then follows for sufficiently large $N$ by adjusting constants appropriately.
\end{proof}

\subsection{Analysis of the Average Discrete Hessian}
\label{subsec:average-discrete-hessian}
Our analysis of the average discrete Hessian requires the following concentration inequality for operator norms of symmetric subgaussian matrices.

\begin{lemma}[{\cite[Corollary~4.4.8]{vershynin2018high}}]
\label{lem:subgaussian-symmetric-matrix-op-norm}
Let $W\in\Sym_N$ have zero diagonal and independent, mean-zero $\sigma$-subgaussian entries above the diagonal, for some $\sigma>0$. Then, for every $t\ge0$,
\begin{align*}
    \Prb\left[\norm{W}_{\mathrm{op}}\ge C\sigma(\sqrt N+t)\right]
    \le\exp(-ct^2),
\end{align*}
\end{lemma}

We now characterize the structure of the average discrete Hessian. In particular, we observe that for even $k$, it exhibits a spiked Wigner type structure.

\begin{theorem}[Signal-Noise Decomposition of the Average Discrete Hessian]
\label{avg-discrete-hess-snr}
Let $H$ denote the tensor PCA Hamiltonian defined in \cref{eq:tensor-pca-hamiltonian}, and define the average Hessian $A_H\doteq\E_\pi[\dgrad^2H(x)]$. Then, the following holds:
\begin{enumerate}
    \item If $k$ is odd, then $A_H=0$ holds deterministically.
    \item If $k$ is even, then
    \begin{align*}
        A_H=\lambda\gamma_N\left(\plant(\plant)^\top-\frac{I}{N}\right)+W_N,
    \end{align*}
    where $\gamma_N\asymp_k N^{-\tfrac{k-2}{2}}$ and $W_N$ satisfies
    \begin{align*}
        \Prb\left[\norm{W_N}_{\mathrm{op}}\ge C_kN^{-\tfrac{k-2}{4}}\right]
        \le\exp(-cN).
    \end{align*}
\end{enumerate}
\end{theorem}

\begin{proof}
By Stein's Lemma (\cref{lem:hypercube-stein}) and \cref{eq:tensor-pca-hamiltonian,eq:tensor-pca-overlap-noise}, we obtain
\begin{align}
\label{eq:average-hessian-decomposition}
    A_H
    &=\E_\pi[(xx^\top-I)H(x)]
    =A_{\mathrm{sig}}+W_N,\\
    A_{\mathrm{sig}}
    &=N\lambda\E_\pi[(xx^\top-I)m(x)^k],\nonumber\\
    W_N
    &=\E_\pi[(xx^\top-I)H_0(x)]
    =N^{-\tfrac{k-1}{2}}\E_\pi\left[(xx^\top-I)\dotp{G}{x^{\otimes k}}\right].\nonumber
\end{align}
If $k$ is odd, then $H(-x)=-H(x)$, whereas $xx^\top-I$ is invariant under $x\mapsto-x$. Invariance of $\pi$ under this map therefore implies that $A_H=0$. Henceforth, assume $k$ is even. We first analyze the signal term in \cref{eq:average-hessian-decomposition}. Define
\begin{align}
\label{eq:average-hessian-gamma-moment}
    \gamma_N
    &=N^{-(k-2)}\E_{z\sim\pi}\left[z_1z_2\left(\sum_{\ell=1}^N z_\ell\right)^k\right]
    =k(k-1)!!\,N^{-\tfrac{k-2}{2}}\left[1+O_k(N^{-1})\right],
\end{align}
where the second equality follows from \cref{lem:rademacher-moments}. Since both $\pi$ and $\sum_{\ell=1}^N z_\ell$ are invariant under coordinate permutations, we obtain the same expectation upon replacing $z_1z_2$ by $z_iz_j$ for any distinct $i,j\in[N]$.

Let $D=\diag(\xopt)$. For $x\sim\pi$, the vector $y=Dx$ also has distribution $\pi$, and $x_i=\xopt_i y_i$. Hence, for distinct $i,j\in[N]$,
\begin{align*}
    (A_{\mathrm{sig}})_{ij}
    &=\lambda N^{1-k}\E_{x\sim\pi}\left[x_ix_j\left(\sum_{\ell=1}^N\xopt_\ell x_\ell\right)^k\right]\\
    &=\lambda N^{1-k}\xopt_i\xopt_j\E_{y\sim\pi}\left[y_iy_j\left(\sum_{\ell=1}^N y_\ell\right)^k\right]
    =\frac{\lambda\gamma_N}{N}\xopt_i\xopt_j.
\end{align*}
Since $(A_{\mathrm{sig}})_{ii}=0$ and $\xopt=\sqrt N\,\plant$, we conclude that
\begin{align*}
    A_{\mathrm{sig}}
    =\frac{\lambda\gamma_N}{N}\left(\xopt(\xopt)^\top-I\right)
    =\lambda\gamma_N\left(\plant(\plant)^\top-\frac{I}{N}\right).
\end{align*}
Moreover, \cref{eq:average-hessian-gamma-moment} implies that $\lambda\gamma_N\asymp_k\lambda N^{-\tfrac{k-2}{2}}$.

We now analyze the noise term in \cref{eq:average-hessian-decomposition}. Since $W_N$ is a linear function of the centered Gaussian tensor $G$, its entries are jointly centered Gaussian. Moreover, $W_N$ is symmetric with zero diagonal. For $i,j,a,b\in[N]$ with $i\ne j$ and $a\ne b$, we use \cref{eq:kspin-covariance} to obtain
\begin{align*}
    \E_G[(W_N)_{ij}(W_N)_{ab}]
    &=\E_{x,x'\,\iidsim\pi}\left[x_ix_jx'_ax'_b\,\E_G[H_0(x)H_0(x')]\right]\\
    &=N^{-(k-1)}\E_{x,x'\,\iidsim\pi}\left[x_ix_jx'_ax'_b\left(\sum_{\ell=1}^N x_\ell x'_\ell\right)^k\right].
\end{align*}
Let $z_\ell=x_\ell x'_\ell$ for each $\ell\in[N]$. Conditional on any $x$, the vector $z$ is uniform on $\Cube$, so $x$ and $z$ are independent with law $\pi$, and $x'_\ell=x_\ell z_\ell$. Consequently,
\begin{align}
    \E_G[(W_N)_{ij}(W_N)_{ab}]
    &=N^{-(k-1)}\E_{x\sim\pi}[x_ix_jx_ax_b]\,
      \E_{z\sim\pi}\left[z_az_b\left(\sum_{\ell=1}^N z_\ell\right)^k\right]\nonumber\\
    &=\frac{\gamma_N}{N}\ind{\{i,j\}=\{a,b\}}.
\label{eq:average-hessian-noise-covariance}
\end{align}
Since the entries of $W_N$ are jointly centered Gaussian, their law is uniquely determined by their covariance $\E_G[(W_N)_{ij}(W_N)_{ab}]$. Hence, \cref{eq:average-hessian-noise-covariance} implies that the entries of $W_N$ above the diagonal are independent and distributed as $\cN(0,\nicefrac{\gamma_N}{N})$. Moreover, since $W_N$ has zero diagonal, we conclude the following by applying \cref{lem:subgaussian-symmetric-matrix-op-norm} with $\sigma^2=\nicefrac{\gamma_N}{N}$ and $t=\sqrt N$, we obtain
\begin{align*}
    \Prb\left[\norm{W_N}_{\mathrm{op}}\ge2C\sqrt{\gamma_N}\right]
    \le\exp(-cN).
\end{align*}
Observing that $\gamma_N\asymp_k N^{-\tfrac{k-2}{2}}$, concludes the proof after adjusting constants appropriately.
\end{proof}

\begin{remark}
\label{rem:average-hessian-snr}
For even $k$, we observe from \cref{avg-discrete-hess-snr} that the rank $1$ spike in the average discrete Hessian is precisely $\plant$. Indeed, $A_{\mathrm{sig}}$ has top eigenvector $\plant$ and eigengap $\lambda_1(A_{\mathrm{sig}})-\lambda_2(A_{\mathrm{sig}})=\lambda\gamma_N$. Relative to the noise scale $\sqrt{\gamma_N}$, the signal-to-noise ratio is therefore $\lambda\sqrt{\gamma_N}\asymp_k\lambda N^{-\tfrac{k-2}{4}}$. Hence, the Davis--Kahan theorem \citep[Theorem~4.5.5]{vershynin2018high} suggests that a top unit eigenvector of $A_H$ achieves constant absolute correlation with $\plant$ with probability at least $1-\exp(-cN)$ whenever $\lambda$ is above the conjectured algorithmic threshold, namely when $\lambda\gtrsim N^{\tfrac{k-2}{4}}$ 
\end{remark}

%% file: ergodic_averaging.tex
\section{Ergodic Averages of High Temperature Glauber Dynamics}
\label{sec:ergodic-averaging}

Our goal in this section is to establish the following recovery guarantee for ergodic averages of high temperature Glauber dynamics. 

\begin{theorem}[Recovery by Ergodic Averaging for Tensor PCA]
\label{thm:tpca-averaged-differential-recovery}
Let $\lambda\ge C_kN^{\tfrac{k-2}{4}}$, and let $(\tilX_t)_{t\ge0}$ denote Glauber dynamics for the tensor PCA Hamiltonian $H$ at inverse temperature $\beta_0>0$, initialized at $\tilX_0\sim\pi$ independently of the Gaussian tensor. Define
\begin{align*}
    \Xbar_T\doteq\frac{1}{T}\int_0^T\tilX_t\,dt,
    \qquad
    \Sigmabar_T\doteq\frac{1}{T}\int_0^T\tilX_t\tilX_t^\top\,dt,
\end{align*}
and write $u_H=\E_\pi[\dgrad H]$ and $A_H=\E_\pi[\dgrad^2H]$. For every fixed $\vartheta\in(0,1)$ and all sufficiently large $N$, the following hold whenever $0<\beta_0\lesssim_k\lambda^{-1}N^{-\nicefrac{k}{2}-\vartheta}$ and $T_{\mathrm{av}}\gtrsim_k\beta_0^{-2}\lambda^{-2}N^{k+2\vartheta}$, with sufficiently small and large implicit constants, respectively:
\begin{enumerate}
    \item
    If $k$ is odd, then, with probability at least $1-\exp(-c_kN^{(1-\vartheta)/2})$,
    \begin{align*}
        \frac{1}{N}d_H\bigl(\sgn(\Xbar_{T_{\mathrm{av}}}),\sgn(u_H)\bigr)
        \le\min\{N^{-\vartheta},\lambda^2\sigma_N^2N^{-\vartheta/2}\},
    \end{align*}
    where $\sigma_N\asymp_k N^{-(k-1)/4}$ is as defined in \cref{avg-discrete-grad-snr}.

    \item If $k$ is even and $\widehat v$ is a top unit eigenvector of $\Sigmabar_{T_{\mathrm{av}}}$, then, with probability at least $1-\exp(-c_kN)$,
    \begin{align*}
        \bigl|m(\sgn(\widehat v))\bigr|
        \ge 1-C_k\left(\frac{N^{\tfrac{k-2}{2}}}{\lambda^2}+N^{-2\vartheta}\right).
    \end{align*}
\end{enumerate}
\end{theorem}

\Cref{thm:tpca-averaged-differential-recovery}, which is proved in \cref{subsec:tpca-recovery-proof}, is obtained by specializing our results on ergodic averages of Glauber dynamics for smooth Hamiltonians to tensor PCA. These results and their proof techniques, which we believe may be of independent interest, are presented in \cref{subsec:recovering-average-discrete-gradient,subsec:recovering-average-discrete-hessian}. In \cref{subsec:srw-resolvents} we establish some preliminaries regarding continuous-time Glauber dynamics.

\subsection{Functional Analysis of the Simple Random Walk}
\label{subsec:srw-resolvents}

Let $(P_t)_{t\ge0}$ denote the Markov semigroup of the continuous-time simple random walk (SRW) on $\Cube$, whose site update kernel is uniform on $\{\pm1\}$. The associated generator $\cL$ and carr\'e du champ operator $\Gamma$ satisfy the following for any $f,g:\Cube\to\R$,
\begin{align}
\label{eq:srw-generator}
    \cL f(x)
    &=\frac12\sum_{i=1}^N\bigl[f(x^{\sim i})-f(x)\bigr]
    =-\dotp{x}{\dgrad f(x)},\\
\label{eq:srw-carre-du-champ}
    \Gamma(f,g)(x)
    &=\frac14\sum_{i=1}^N\bigl[f(x^{\sim i})-f(x)\bigr]\bigl[g(x^{\sim i})-g(x)\bigr]
    =\dotp{\dgrad f(x)}{\dgrad g(x)}.
\end{align}

Observe that SRW corresponds to Glauber dynamics at infinite temperature, i.e., $\beta=0$. Hence, for any Hamiltonian $V:\Cube\to\R$ and inverse temperature $\beta\ge0$, we infer the following identity from \cref{eq:glauber-generator,eq:srw-generator}, which we call the \emph{drift-diffusion decomposition} for the generator of Glauber dynamics. 
\begin{align}
\label{eq:glauber-generator-srw-decomposition}
    \cL_{V,\beta}f(x)
    =\cL f(x)+\dotp{\tanh(\beta\dgrad V(x))}{\dgrad f(x)}.
\end{align}
Here, $\tanh$ implicitly acts elementwise, and the term $\dotp{\tanh(\beta\dgrad V(x))}{\dgrad f(x)}$ is called the \emph{drift term} of the generator. 

An essential ingredient in our proofs is the following discrete gradient commutation identity for the SRW semigroup, which can be interpreted as a hypercube analogue of the celebrated Bakry--\'Emery estimate for Markov diffusion semigroups on manifolds satisfying a curvature-dimension condition~\citep[Section~3.2.3]{bakry2013analysis}.

\begin{lemma}[Bakry--\'Emery Type Discrete Gradient Estimate for the Hypercube]
\label{lem:hypercube-bakry-emery}
The following holds for every $f:\Cube\to\R^m$ and $t\ge0$,
\begin{align*}
    \dgrad(P_tf)=e^{-t}P_t(\dgrad f),
\end{align*}
Moreover, if $\sup_{x \in \Cube }\norm{\dgrad f(x)}_{\mathrm{op}}\le L$, then
\begin{align*}
    \sup_{x \in \Cube}\norm{\dgrad(P_tf)(x)}_{\mathrm{op}}\le e^{-t}L
    \qquad\forall \ t \ge 0
\end{align*}
\end{lemma}

\begin{proof}
Fix $u\in\R^m$ and set $g(x)=\dotp{f(x)}{u}$. By the Fourier--Walsh expansion,
\begin{align*}
    g(x)=\sum_{S\subseteq[N]}\widehat g(S)\chi_S(x).
\end{align*}
Recall from \cref{eq:partial-fourier-character} that $\partial_i\chi_S=\chi_{S\setminus\{i\}}$ if $i\in S$, and $\partial_i\chi_S=0$ otherwise. By \cref{eq:srw-generator},
\begin{align}
\label{eq:srw-generator-character}
    \cL\chi_S(x)
    =\frac12\sum_{i=1}^N\bigl[\chi_S(x^{\sim i})-\chi_S(x)\bigr]
    =\frac12\sum_{i\in S}(-2\chi_S(x))
    =-|S|\chi_S(x).
\end{align}
Consequently, the Kolmogorov equation~\cref{eq:glauber-kolmogorov} at $\beta=0$ implies
\begin{align*}
    \frac{d}{dt}P_t\chi_S=P_t(\cL\chi_S)=-|S|P_t\chi_S.
\end{align*}
Since $P_0=\mathrm{Id}$, we obtain the following by integrating the above identity
\begin{align}
\label{eq:srw-semigroup-character}
    P_t\chi_S=e^{-|S|t}\chi_S.
\end{align}
Consequently, for $i\notin S$, both $\partial_i(P_t\chi_S)$ and $P_t(\partial_i\chi_S)$ vanish. Moreover, for $i\in S$, we have the following:
\begin{align*}
    \partial_i(P_t\chi_S)
    &=e^{-|S|t}\chi_{S\setminus\{i\}},\\
    P_t(\partial_i\chi_S)
    &=e^{-(|S|-1)t}\chi_{S\setminus\{i\}}
    =e^t\partial_i(P_t\chi_S).
\end{align*}
Thus, $\partial_i(P_t\chi_S)=e^{-t}P_t(\partial_i\chi_S)$ for every $i\in[N]$ and $S\subseteq[N]$. By the Fourier--Walsh expansion and linearity of $P_t$ and $\dgrad$, we conclude that $\dgrad(P_tg)=e^{-t}P_t(\dgrad g)$. Applying this identity to $g=\dotp{f}{u}$, we obtain
\begin{align*}
    \bigl[\dgrad(P_tf)\bigr]u
    =e^{-t}P_t\bigl[(\dgrad f)u\bigr].
\end{align*}
Since $u\in\R^m$ was arbitrary, this proves the first identity. Finally, by Jensen's inequality for Markov semigroups,
\begin{align*}
    \norm{\dgrad(P_tf)(x)}_{\mathrm{op}}
    &=e^{-t}\norm{P_t(\dgrad f)(x)}_{\mathrm{op}}\\
    &\le e^{-t}P_t\bigl(\norm{\dgrad f}_{\mathrm{op}}\bigr)(x)
    \le e^{-t}L,
\end{align*}
which proves the second claim.
\end{proof}
We now introduce the Green operator (or restricted inverse of $-\cL$) for the simple random walk. Let $L^2(\pi)\supseteq L^2_0(\pi)=\{f\in L^2(\pi):\E_\pi[f]=0\}$ denote the subspace of zero mean functions on $\Cube$. The Green operator $\cR:L^2_0(\pi)\to L^2_0(\pi)$ is then defined as follows:
\begin{align}
\label{eq:srw-green-operator}
    \cR f=\int_0^\infty P_tf\,dt
    \qquad \forall \ f\in L^2_0(\pi).
\end{align}
For centered vector and matrix valued observables, $\cR$ acts elementwise. The following lemma shows that $\cR$ is a well-defined bounded linear operator and establishes several properties that will be vital to our analysis.
\begin{lemma}[Properties of the Green operator]
\label{lem:srw-green-operator}
The Green operator $\cR$ is a well-defined bounded linear operator on $L^2_0(\pi)$. Moreover, for any $f \in L^2_0(\pi)$, $\cR f$ is a solution to the Poisson equation
\begin{align}
\label{eq:srw-poisson-equation}
    \cL(\cR f)=-f
\end{align}
Furthermore, $\norm{\cR f}_{\cD^1}\le2\norm{f}_{\cD^1}$ for every $f\in L^2_0(\pi)$.
\end{lemma}

\begin{proof}
Observe that the set of non-constant Fourier characters $\{\chi_S:\varnothing\ne S\subseteq[N]\}$ forms an orthonormal basis of $L^2_0(\pi)$. Fix an arbitrary $f\in L^2_0(\pi)$. By the Kolmogorov equation, $P_t$ acts on these characters as per \cref{eq:srw-semigroup-character} from the proof of \cref{lem:hypercube-bakry-emery}. Hence, we obtain the following via linearity of $\cR$ and a Fourier-Walsh expansion
\begin{align*}
    \cR f
    &=\sum_{\varnothing\ne S\subseteq[N]}\widehat f(S)\chi_S\int_0^\infty e^{-|S|t}\,dt
    =\sum_{\varnothing\ne S\subseteq[N]}\frac{\widehat f(S)}{|S|}\chi_S.
\end{align*}
Since the sum is finite and each integral converges, $\cR$ is a well-defined linear operator on $L^2_0(\pi)$. By Parseval's identity~\citep[Section~1.4]{ryan2014aobf}, we obtain
\begin{align*}
    \norm{\cR f}_{L^2(\pi)}^2
    =\sum_{\varnothing\ne S\subseteq[N]}\frac{|\widehat f(S)|^2}{|S|^2}
    \le\sum_{\varnothing\ne S\subseteq[N]}|\widehat f(S)|^2
    =\norm{f}_{L^2(\pi)}^2.
\end{align*}
Hence, $\cR$ is a bounded linear operator on $L^2_0(\pi)$. To prove \cref{eq:srw-poisson-equation}, we use linearity of $\cL$ and \cref{eq:srw-generator-character} from the proof of \cref{lem:hypercube-bakry-emery} to obtain
\begin{align*}
    \cL(\cR f)
    &=\sum_{\varnothing\ne S\subseteq[N]}\frac{\widehat f(S)}{|S|}\cL\chi_S =-\sum_{\varnothing\ne S\subseteq[N]}\widehat f(S)\chi_S
    =-f,
\end{align*}
where the last equality uses $\widehat f(\varnothing)=\E_\pi[f]=0$.
We now prove the discrete Sobolev norm bound. Since discrete derivatives are finite differences, we may apply them inside the integral defining $\cR f$. Hence, by \cref{lem:hypercube-bakry-emery},
\begin{align*}
    \max_{x\in\Cube}\norm{\dgrad(\cR f)(x)}
    &\le\int_0^\infty\max_{x\in\Cube}\norm{\dgrad(P_tf)(x)}\,dt\\
    &\le\max_{x\in\Cube}\norm{\dgrad f(x)}\int_0^\infty e^{-t}\,dt
    =\max_{x\in\Cube}\norm{\dgrad f(x)}.
\end{align*}

Since $\pi$ is the stationary measure of $P_t$ and SRW is ergodic, we infer the following from the Kolmogorov equation~\cref{eq:glauber-kolmogorov} and \cref{eq:srw-generator} for every $x\in\Cube$ and $t\ge0$,
\begin{align*}
    P_tf(x)-\E_\pi[f]
    &=-\int_t^\infty\frac{d}{ds}P_sf(x)\,ds\\
    &=-\int_t^\infty\cL(P_sf)(x)\,ds
    =\int_t^\infty\dotp{x}{\dgrad(P_sf)(x)}\,ds.
\end{align*}
Since $\E_\pi[f]=0$ and $\norm{x}=\sqrt N$, we obtain the following from \cref{lem:hypercube-bakry-emery}.
\begin{align*}
    \norm{P_tf}_\infty
    &\le\sqrt N\int_t^\infty\max_{x\in\Cube}\norm{\dgrad(P_sf)(x)}\,ds\\
    &\le\sqrt N\max_{x\in\Cube}\norm{\dgrad f(x)}\int_t^\infty e^{-s}\,ds
    =\sqrt N e^{-t}\max_{x\in\Cube}\norm{\dgrad f(x)}.
\end{align*}
Integrating in $t$ and using \cref{eq:srw-green-operator}, we obtain
\begin{align*}
    \norm{\cR f}_\infty
    &\le\int_0^\infty\norm{P_tf}_\infty\,dt
    \le\sqrt N\max_{x\in\Cube}\norm{\dgrad f(x)}\int_0^\infty e^{-t}\,dt\\
    &=\sqrt N\max_{x\in\Cube}\norm{\dgrad f(x)}.
\end{align*}
Combining these estimates with \cref{def:discrete-sobolev-norm}, we conclude that
\begin{align*}
    \norm{\cR f}_{\cD^1}
    &=\norm{\cR f}_\infty+\sqrt N\max_{x\in\Cube}\norm{\dgrad(\cR f)(x)}\\
    &\le2\sqrt N\max_{x\in\Cube}\norm{\dgrad f(x)}
    \le2\norm{f}_{\cD^1}.
\end{align*}
\end{proof}

\begin{remark}
\label{rem:srw-green-operator-provenance}
By \cref{eq:srw-poisson-equation}, we may interpret $\cR$ as the inverse of $-\cL$ restricted to $L^2_0(\pi)$. Solutions to the Poisson equation are a canonical tool for proving central limit theorems for ergodic averages of Markov processes~\citep{komorowski2012fluctuations,eberle2025markov} (see also~\citep[Section~4.1]{chatterjee2005concentration} for connections to Stein's method of exchangeable pairs). More recently, such methods have been used to compare stationary distributions of contractive Markov chains under Dobrushin's conditions~\citep{bresler2019stein,reinert2019approximating}. In subsequent sections, we shall use the Green operator to analyze ergodic averages along high-temperature Glauber dynamics for Hamiltonians with bounded $\cD^2$ norm.
\end{remark}

\subsection{Recovering the Average Discrete Gradient}
\label{subsec:recovering-average-discrete-gradient}
We now show that the ergodic average of high-temperature Glauber dynamics approximates the average discrete gradient under $\pi$. For a potential $V:\Cube\to\R$, let $(Z_t)_{t\ge0}$ denote the trajectory of Glauber dynamics at inverse temperature $\beta_0\ge0$, initialized at $Z_0\sim\pi$. For any averaging time $T>0$, define its ergodic average by
\begin{align}
\label{eq:glauber-ergodic-average}
    \Zavg_T\doteq\frac{1}{T}\int_0^T Z_t\,dt.
\end{align}

In \cref{subsec:sketch-ergodic-averaging}, we argued heuristically that, for $V=H$, the ergodic average $\Zavg_T$ converges to the mean of the Gibbs measure $\mu_{V,\beta_0}$ as $T\to\infty$, and that this mean equals $\beta_0\E_\pi[\dgrad V]+O(\beta_0^2)$. We formalize this heuristic for general potentials in \cref{thm:glauber-recovers-average-gradient} below. The error bound in this theorem comprises terms that vanish as $T\to\infty$, which quantify the long time limit, and terms of higher order in $\beta_0$, which quantify the high temperature expansion.

\begin{theorem}[Recovery of the Average Discrete Gradient]
\label{thm:glauber-recovers-average-gradient}
Suppose that $\norm{V}_{\cD^2}\le BN$ for some $B>0$. For every $T>0$ and $\delta\in(0,1)$, the following holds with probability at least $1-\delta$:
\begin{align*}
    \norm{\Zavg_T-\beta_0\E_\pi[\dgrad V]}
    &\lesssim (1+\beta_0B)\left[
        \sqrt{\frac{N+\log(\nicefrac{1}{\delta})}{T}}
        +\frac{N+\log(\nicefrac{1}{\delta})}{T}
    \right]+\beta_0^2B^2\sqrt N+\beta_0^3B^3N^{3/2},
\end{align*}
where the implicit constant is universal.
\end{theorem}

\begin{proof}
The proof proceeds in two steps. The first step uses the Poisson equation from \cref{lem:srw-green-operator} and the drift-diffusion identity \cref{eq:glauber-generator-srw-decomposition} for the generator of Glauber dynamics to prove that the ergodic average of $\dgrad V(Z_t)$ concentrates around $\E_\pi[\dgrad V]$. The second step linearizes the drift term in \cref{eq:glauber-generator-srw-decomposition} to show that the ergodic averages of $Z_t$ and $\beta_0 \dgrad V(Z_t)$ are close to each other. 

Define $F(x)=\dgrad V(x)-\E_\pi[\dgrad V]$ and $\Phi=\cR F$. For any $w\in\bS^{N-1}$, the function $f_w=\dotp{F}{w}$ has zero mean under $\pi$. Since $\norm{V}_{\cD^2}\le BN$, we have $\norm{f_w}_{\cD^1}\le3B\sqrt N$. Consequently, by linearity of $\cR$ and \cref{lem:srw-green-operator}, we obtain $\norm{\dotp{\Phi}{w}}_{\cD^1}=\norm{\cR f_w}_{\cD^1}\le6B\sqrt N$. Taking the supremum over $w\in\bS^{N-1}$, we obtain
\begin{align}
\label{eq:gradient-recovery-green-bounds}
    \max_{x\in\Cube}\norm{\Phi(x)}\le6B\sqrt N,
    \qquad
    \max_{x\in\Cube}\norm{\dgrad\Phi(x)}_{\mathrm{op}}\le6B.
\end{align}

The Poisson equation~\cref{eq:srw-poisson-equation} from \cref{lem:srw-green-operator} implies $\cL\Phi=-F$. Combining this with the drift-diffusion decomposition~\cref{eq:glauber-generator-srw-decomposition}, we obtain
\begin{align}
\label{eq:gradient-recovery-generator}
    \cL_{V,\beta_0}\Phi(x)
    =-F(x)+\dgrad\Phi(x)^\top\tanh(\beta_0\dgrad V(x)).
\end{align}
Let $M_t^\Phi$ denote the Dynkin martingale associated with $\Phi$ along $(Z_t)_{t\ge0}$. By \cref{lem:single-site-dynkin},
\begin{align*}
    M_T^\Phi
    =\Phi(Z_T)-\Phi(Z_0)-\int_0^T\cL_{V,\beta_0}\Phi(Z_t)\,dt.
\end{align*}
By rearranging terms and using \cref{eq:gradient-recovery-generator}, we obtain
\begin{align}
\label{eq:gradient-recovery-dynkin-decomposition}
    \frac{1}{T}\int_0^T\dgrad V(Z_t)\,dt-\E_\pi[\dgrad V]
    &=\frac{\Phi(Z_0)-\Phi(Z_T)}{T}
      +\frac{M_T^\Phi}{T}
      +\frac{1}{T}\int_0^T
      \dgrad\Phi(Z_t)^\top\tanh(\beta_0\dgrad V(Z_t))\,dt.
\end{align}
We control the three terms separately. By \cref{eq:gradient-recovery-green-bounds}, the endpoint term satisfies
\begin{align}
\label{eq:gradient-recovery-endpoint-bound}
    \norm{\frac{\Phi(Z_0)-\Phi(Z_T)}{T}}
    \le\frac{12B\sqrt N}{T}.
\end{align}
For the drift term, since $\norm{V}_{\cD^2}\le BN$, we obtain the following from $|\tanh z|\le|z|$ and \cref{eq:gradient-recovery-green-bounds}:
\begin{align}
\label{eq:gradient-recovery-remainder-bound}
    \norm{\dgrad\Phi(x)^\top\tanh(\beta_0\dgrad V(x))}
    &\le\norm{\dgrad\Phi(x)}_{\mathrm{op}}\norm{\tanh(\beta_0\dgrad V(x))}\nonumber\\
    &\le\beta_0\norm{\dgrad\Phi(x)}_{\mathrm{op}}\norm{\dgrad V(x)}
    \le6\beta_0B^2\sqrt N.
\end{align}

It remains to control the martingale term. Applying \cref{lem:matrix-dynkin-uniform-concentration} with $m=N$ and $n=1$, we obtain the following with probability at least $1-\delta/2$:
\begin{align}
\label{eq:gradient-recovery-martingale-bound}
    \frac{\norm{M_T^\Phi}}{T}
    \lesssim B\left[
        \sqrt{\frac{N+\log(\nicefrac{1}{\delta})}{T}}
        +\frac{N+\log(\nicefrac{1}{\delta})}{T}
    \right].
\end{align}
Substituting \cref{eq:gradient-recovery-endpoint-bound,eq:gradient-recovery-remainder-bound,eq:gradient-recovery-martingale-bound} into \cref{eq:gradient-recovery-dynkin-decomposition}, we obtain the following with probability at least $1-\delta/2$:
\begin{align}
\label{eq:gradient-recovery-gradient-average-bound}
    \norm{\frac{1}{T}\int_0^T\dgrad V(Z_t)\,dt-\E_\pi[\dgrad V]}
    \lesssim B\left[
        \sqrt{\frac{N+\log(\nicefrac{1}{\delta})}{T}}
        +\frac{N+\log(\nicefrac{1}{\delta})}{T}
    \right]+\beta_0B^2\sqrt N.
\end{align}

It remains to compare the respective ergodic averages of $Z_t$ and $\dgrad V(Z_t)$. To this end, applying the generator decomposition~\cref{eq:glauber-generator-srw-decomposition} to eac $x_i$ for $i \in [N]$, we obtain
\begin{align*}
    \cL_{V,\beta_0}x
    =-x+\tanh(\beta_0\dgrad V(x)).
\end{align*}
Let $M_t^x$ denote the associated Dynkin martingale. By \cref{lem:single-site-dynkin},
\begin{align*}
    M_T^x
    =Z_T-Z_0-\int_0^T\bigl[\tanh(\beta_0\dgrad V(Z_t))-Z_t\bigr]\,dt.
\end{align*}
By rearranging terms and using \cref{eq:glauber-ergodic-average}, we obtain
\begin{align}
\label{eq:gradient-recovery-spin-decomposition}
    \Zavg_T-\frac{\beta_0}{T}\int_0^T\dgrad V(Z_t)\,dt
    =\frac{Z_0-Z_T}{T}+\frac{M_T^x}{T}
    +\frac{1}{T}\int_0^T
      \bigl[\tanh(\beta_0\dgrad V(Z_t))-\beta_0\dgrad V(Z_t)\bigr]\,dt.
\end{align}
Since $Z_0, Z_T \in \Cube$, the endpoint term is at most $2\sqrt N/T$. Moreover, since $\norm{V}_{\cD^2}\le BN$, the elementary inequality $|\tanh z-z|\le|z|^3$ implies
\begin{align*}
    \norm{\tanh(\beta_0\dgrad V(x))-\beta_0\dgrad V(x)}
    \le\beta_0^3\norm{\dgrad V(x)}^3
    \le\beta_0^3B^3N^{3/2}.
\end{align*}
Finally, $\norm{\dotp{x}{w}}_{\cD^1}\le2\sqrt N$ for every $w\in\bS^{N-1}$. Thus, applying \cref{lem:matrix-dynkin-uniform-concentration} with $m=N$ and $n=1$, we obtain the following with probability at least $1-\delta/2$:
\begin{align*}
    \frac{\norm{M_T^x}}{T}
    \lesssim
        \sqrt{\frac{N+\log(\nicefrac{1}{\delta})}{T}}
        +\frac{N+\log(\nicefrac{1}{\delta})}{T}.
\end{align*}
Combining these estimates with \cref{eq:gradient-recovery-spin-decomposition}, we conclude that, with probability at least $1-\delta/2$,
\begin{align}
\label{eq:gradient-recovery-spin-average-bound}
    \norm{\Zavg_T-\frac{\beta_0}{T}\int_0^T\dgrad V(Z_t)\,dt}
    \lesssim
        \sqrt{\frac{N+\log(\nicefrac{1}{\delta})}{T}}
        +\frac{N+\log(\nicefrac{1}{\delta})}{T}
        +\beta_0^3B^3N^{3/2}.
\end{align}
The proof is concluded by combining \cref{eq:gradient-recovery-gradient-average-bound,eq:gradient-recovery-spin-average-bound} via the triangle inequality and a union bound.
\end{proof}

\subsection{Recovering the Average Discrete Hessian}
\label{subsec:recovering-average-discrete-hessian}
Borrowing the notation from \cref{subsec:recovering-average-discrete-gradient}, we use $Z_t$ to denote Glauber dynamics for a potential $V:\Cube\to\R$ at inverse temperature $\beta_0\ge0$, initialized at $Z_0\sim\pi$. For any averaging time $T>0$, define its uncentered ergodic second moment by
\begin{align*}
    \Savg_T\doteq\frac{1}{T}\int_0^T Z_t Z_t^\top\,dt.
\end{align*}
Define the matrix-valued observable $Q_V:\Cube\to\Sym_N$ by
\begin{align}
\label{eq:hessian-recovery-observable}
    Q_V(x)\doteq\frac12\operatorname{off}\left(
        x\dgrad V(x)^\top+\dgrad V(x)x^\top
    \right).
\end{align}
The following Lemma, which controls the discrete Sobolev norm of $Q_V$, is proven in \cref{app:hessian-observable-proof}.
\begin{lemma}[Properties of the Hessian observable]
\label{lem:hessian-observable-properties}
For any $V:\Cube\to\R$, the observable $Q_V$ satisfies $\E_\pi[Q_V]=\E_\pi[\dgrad^2V]$. Moreover, if $\norm{V}_{\cD^2}\le BN$ for some $B>0$, then $\norm{w^\top Q_V(\cdot)w}_{\cD^1}\lesssim BN$ for every $w\in\bS^{N-1}$.
\end{lemma}

The next theorem establishes a concentration guarantee for the ergodic second moment for potentials with bounded Sobolev norm.
\begin{theorem}[Recovery of the Average Discrete Hessian]
\label{thm:glauber-recovers-average-hessian}
Suppose that $\norm{V}_{\cD^2}\le BN$ for some $B>0$. For every $T>0$ and $\delta\in(0,1)$, the following holds with probability at least $1-\delta$:
\begin{align*}
    \norm{\Savg_T-I-\beta_0\E_\pi[\dgrad^2V]}_{\mathrm{op}}
    &\lesssim (1+\beta_0B)\sqrt{N}\left[
        \sqrt{\frac{N+\log(\nicefrac{1}{\delta})}{T}}
        +\frac{N+\log(\nicefrac{1}{\delta})}{T}
    \right]+\beta_0^2B^2N+\beta_0^3B^3N^2.
\end{align*}
\end{theorem}

\begin{proof}
As in the proof of \cref{thm:glauber-recovers-average-gradient}, we first use the Poisson equation from \cref{lem:srw-green-operator} and the drift-diffusion identity~\cref{eq:glauber-generator-srw-decomposition} to show that the ergodic average of $Q_V(Z_t)$ concentrates around $\E_\pi[\dgrad^2V]$. We then linearize the drift of $xx^\top-I$ to compare its ergodic average with that of $\beta_0Q_V(Z_t)$.

Define $F(x)=Q_V(x)-\E_\pi[\dgrad^2V]$ and $\Phi=\cR F$. For every $w\in\bS^{N-1}$, the function $f_w=w^\top Fw$ has zero mean under $\pi$ by \cref{lem:hessian-observable-properties} and satisfies $\norm{f_w}_{\cD^1}\lesssim BN$. Consequently, by linearity of $\cR$ and \cref{lem:srw-green-operator}, we obtain $\norm{w^\top\Phi(\cdot)w}_{\cD^1}=\norm{\cR f_w}_{\cD^1}\lesssim BN$. Since $\Phi$ takes values in $\Sym_N$, we conclude that
\begin{align}
\label{eq:hessian-recovery-green-bounds}
    \max_{x\in\Cube}\norm{\Phi(x)}_{\mathrm{op}}\lesssim BN,
    \qquad
    \max_{x\in\Cube}\sup_{w\in\bS^{N-1}}\norm{\dgrad(w^\top\Phi w)(x)}\lesssim B\sqrt N.
\end{align}

The Poisson equation~\cref{eq:srw-poisson-equation} from \cref{lem:srw-green-operator} implies $\cL\Phi=-F$. Combining this with the drift-diffusion decomposition~\cref{eq:glauber-generator-srw-decomposition}, we obtain
\begin{align}
\label{eq:hessian-recovery-generator}
    \cL_{V,\beta_0}\Phi(x)
    =-F(x)+\sum_{i=1}^N\tanh(\beta_0\partial_iV(x))\partial_i\Phi(x).
\end{align}
Let $M_t^\Phi$ denote the Dynkin martingale associated with $\Phi$ along $(Z_t)_{t\ge0}$. By \cref{lem:single-site-dynkin},
\begin{align*}
    M_T^\Phi
    =\Phi(Z_T)-\Phi(Z_0)-\int_0^T\cL_{V,\beta_0}\Phi(Z_t)\,dt.
\end{align*}
By rearranging terms and using \cref{eq:hessian-recovery-generator}, we obtain
\begin{align}
\label{eq:hessian-recovery-dynkin-decomposition}
    \frac{1}{T}\int_0^TQ_V(Z_t)\,dt-\E_\pi[\dgrad^2V]
    =\frac{\Phi(Z_0)-\Phi(Z_T)}{T}+\frac{M_T^\Phi}{T}
      +\frac{1}{T}\int_0^T\sum_{i=1}^N
      \tanh(\beta_0\partial_iV(Z_t))\partial_i\Phi(Z_t)\,dt.
\end{align}
We control the three terms separately. By \cref{eq:hessian-recovery-green-bounds}, the endpoint term satisfies
\begin{align}
\label{eq:hessian-recovery-endpoint-bound}
    \norm{\frac{\Phi(Z_0)-\Phi(Z_T)}{T}}_{\mathrm{op}}
    \lesssim\frac{BN}{T}.
\end{align}
For the drift term, since $\norm{V}_{\cD^2}\le BN$, we obtain the following from Cauchy--Schwarz, $|\tanh z|\le|z|$, and \cref{eq:hessian-recovery-green-bounds}:
\begin{align}
\label{eq:hessian-recovery-drift-bound}
    \norm{\sum_{i=1}^N\tanh(\beta_0\partial_iV(x))\partial_i\Phi(x)}_{\mathrm{op}}
    &=\sup_{w\in\bS^{N-1}}\left|\dotp{\dgrad(w^\top\Phi w)(x)}{\tanh(\beta_0\dgrad V(x))}\right|\nonumber\\
    &\le\sup_{w\in\bS^{N-1}}\norm{\dgrad(w^\top\Phi w)(x)}\norm{\tanh(\beta_0\dgrad V(x))}
    \lesssim\beta_0B^2N.
\end{align}

It remains to control the martingale term. Applying \cref{lem:matrix-dynkin-uniform-concentration} with $m=n=N$, we obtain the following with probability at least $1-\delta/2$:
\begin{align}
\label{eq:hessian-recovery-martingale-bound}
    \frac{\norm{M_T^\Phi}_{\mathrm{op}}}{T}
    \lesssim B\sqrt N\left[
        \sqrt{\frac{N+\log(\nicefrac{1}{\delta})}{T}}
        +\frac{N+\log(\nicefrac{1}{\delta})}{T}
    \right].
\end{align}
Substituting \cref{eq:hessian-recovery-endpoint-bound,eq:hessian-recovery-drift-bound,eq:hessian-recovery-martingale-bound} into \cref{eq:hessian-recovery-dynkin-decomposition}, we obtain the following with probability at least $1-\delta/2$:
\begin{align}
\label{eq:hessian-recovery-observable-average-bound}
    \norm{\frac{1}{T}\int_0^TQ_V(Z_t)\,dt-\E_\pi[\dgrad^2V]}_{\mathrm{op}}
    \lesssim B\sqrt N\left[
        \sqrt{\frac{N+\log(\nicefrac{1}{\delta})}{T}}
        +\frac{N+\log(\nicefrac{1}{\delta})}{T}
    \right]+\beta_0B^2N.
\end{align}

It remains to compare $\Savg_T-I$ with the ergodic average of $\beta_0Q_V(Z_t)$. To this end, define $K(x)=xx^\top-I$ and $r(x)=\tanh(\beta_0\dgrad V(x))-\beta_0\dgrad V(x)$.
Applying the generator decomposition~\cref{eq:glauber-generator-srw-decomposition} to $x_ix_j$ for distinct $i,j\in[N]$, we obtain
\begin{align}
\label{eq:hessian-recovery-pair-generator}
    \cL_{V,\beta_0}(x_ix_j)
    =-2x_ix_j+x_j\tanh(\beta_0\partial_iV(x))+x_i\tanh(\beta_0\partial_jV(x)).
\end{align}
Since $\cL_{V,\beta_0}(x_i^2)=0$, we infer from \cref{eq:hessian-recovery-observable,eq:hessian-recovery-pair-generator} that
\begin{align}
\label{eq:hessian-recovery-quadratic-generator}
    \cL_{V,\beta_0}K(x)=-2K(x)+2\beta_0Q_V(x)+2E(x),
    \qquad
    E(x)=\frac12\operatorname{off}\left(xr(x)^\top+r(x)x^\top\right).
\end{align}
Let $M_t^K$ denote the associated Dynkin martingale. By \cref{lem:single-site-dynkin},
\begin{align*}
    M_T^K=K(Z_T)-K(Z_0)-\int_0^T\cL_{V,\beta_0}K(Z_t)\,dt.
\end{align*}
By rearranging terms and using \cref{eq:hessian-recovery-quadratic-generator}, we obtain
\begin{align}
\label{eq:hessian-recovery-second-moment-decomposition}
    \Savg_T-I-\frac{\beta_0}{T}\int_0^TQ_V(Z_t)\,dt
    =\frac{K(Z_0)-K(Z_T)}{2T}+\frac{M_T^K}{2T}+\frac{1}{T}\int_0^TE(Z_t)\,dt.
\end{align}
Since $\norm{K(x)}_{\mathrm{op}}\le N$, the endpoint term is at most $N/T$. Moreover, since $\norm{V}_{\cD^2}\le BN$, the elementary inequality $|\tanh z-z|\le|z|^3$ implies $\norm{r(x)}\le\beta_0^3\norm{\dgrad V(x)}^3\le\beta_0^3B^3N^{3/2}$. Using $\norm{\operatorname{off}(M)}_{\mathrm{op}}\le2\norm{M}_{\mathrm{op}}$ and $x \in \Cube$, we conclude that
\begin{align}
\label{eq:hessian-recovery-second-moment-remainder-bound}
    \norm{E(x)}_{\mathrm{op}}
    \le2\sqrt N\norm{r(x)}\le2\beta_0^3B^3N^2.
\end{align}

Finally, for every $w\in\bS^{N-1}$,
\begin{align*}
    \partial_i(w^\top K(x)w)=2w_i\bigl(\dotp{w}{x}-w_ix_i\bigr),
    \qquad
    \norm{\dgrad(w^\top Kw)(x)}\le2(\sqrt N+1).
\end{align*}
Together with $\norm{K(x)}_{\mathrm{op}}\le N$, this implies $\norm{w^\top K(\cdot)w}_{\cD^1}\lesssim N$. Thus, applying \cref{lem:matrix-dynkin-uniform-concentration} with $m=n=N$, we obtain the following with probability at least $1-\delta/2$:
\begin{align}
\label{eq:hessian-recovery-quadratic-martingale-bound}
    \frac{\norm{M_T^K}_{\mathrm{op}}}{2T}
    \lesssim\sqrt N\left[
        \sqrt{\frac{N+\log(\nicefrac{1}{\delta})}{T}}
        +\frac{N+\log(\nicefrac{1}{\delta})}{T}
    \right].
\end{align}
Combining the endpoint bound and \cref{eq:hessian-recovery-second-moment-remainder-bound,eq:hessian-recovery-quadratic-martingale-bound} with \cref{eq:hessian-recovery-second-moment-decomposition}, we conclude that, with probability at least $1-\delta/2$,
\begin{align}
\label{eq:hessian-recovery-second-moment-bound}
    \norm{\Savg_T-I-\frac{\beta_0}{T}\int_0^TQ_V(Z_t)\,dt}_{\mathrm{op}}
    \lesssim\sqrt N\left[
        \sqrt{\frac{N+\log(\nicefrac{1}{\delta})}{T}}
        +\frac{N+\log(\nicefrac{1}{\delta})}{T}
    \right]+\beta_0^3B^3N^2.
\end{align}
The proof is concluded by combining \cref{eq:hessian-recovery-observable-average-bound,eq:hessian-recovery-second-moment-bound} via the triangle inequality and a union bound.
\end{proof}

\begin{remark}
\label{rem:ergodic-averaging-diffusion-extension}
The proof of \cref{thm:glauber-recovers-average-gradient,thm:glauber-recovers-average-hessian} is quite flexible and can be changed slightly to yield similar results for a large family of stochastic processes (beyond Glauber Dynamics).
For instance, the Green operator and Dynkin martingale arguments can be adapted to Markov diffusion processes on compact Riemannian manifolds, using pointwise carr\'e du champ estimates obtained via the Bakry--\'Emery curvature-dimension condition~\citep[Section~1.16]{bakry2013analysis}. In our forthcoming companion paper~\citep{das2026averaged} on spherical tensor PCA, we apply these techniques to establish analogous recovery guarantees for spherical Langevin dynamics.
\end{remark}

\subsection{Proof of \texorpdfstring{\cref{thm:tpca-averaged-differential-recovery}}{Theorem~\ref{thm:tpca-averaged-differential-recovery}}}
\label{subsec:tpca-recovery-proof}

\begin{proof}[Proof of \cref{thm:tpca-averaged-differential-recovery}]
\leavevmode\par
Let $\cE=\{\norm{H}_{\cD^2}\le C_k\lambda N\}$. By \cref{lem:kspin-sobolev-norm-concentration,lem:overlap-sobolev-bounds}, we have $\Prb[\cE]\ge1-\exp(-c_kN)$. For every realization of $H$ in $\cE$, we apply \cref{thm:glauber-recovers-average-gradient,thm:glauber-recovers-average-hessian} with $B\asymp_k\lambda$ and $\delta=\exp(-N)$. Taking a union bound, we obtain the following jointly with $\cE$, with probability at least $1-\exp(-c_kN)$:
\begin{align}
\label{eq:tpca-averaging-normalized-errors}
    \norm{\frac{\Xbar_{T_{\mathrm{av}}}}{\beta_0}-u_H}
    \le\eps_N,
    \qquad
    \norm{\frac{\Sigmabar_{T_{\mathrm{av}}}-I}{\beta_0}-A_H}_{\mathrm{op}}
    \le\sqrt N\,\eps_N.
\end{align}
Our assumptions on on $\beta_0$ and $T_{\mathrm{av}}$ imply
\begin{align}
\label{eq:tpca-averaging-error-scale}
    \eps_N
    &\lesssim_k
        (\beta_0^{-1}+\lambda)\left(
            \sqrt{\frac{N}{T_{\mathrm{av}}}}+\frac{N}{T_{\mathrm{av}}}
        \right)
        +\beta_0\lambda^2\sqrt N+\beta_0^2\lambda^3N^{3/2}
        \le c_k\lambda N^{-\tfrac{k-1}{2}-\vartheta}.
\end{align}
Here $c_k>0$ can be made arbitrarily small by choosing the leading constant in the upper bound on $\beta_0$ sufficiently small and that in the lower bound on $T_{\mathrm{av}}$ sufficiently large.

For odd $k$, \cref{avg-discrete-grad-snr} implies $u_H=\lambda\sigma_N^2\xopt+\zeta$, where $\zeta\sim\cN(0,\sigma_N^2I)$ and $\sigma_N^2\asymp_k N^{-\tfrac{k-1}{2}}$. Fix any $i \in [N]$. If $\sgn((u_H)_i)\ne\sgn(\Xbar_{T_{\mathrm{av}},i})$, then
\begin{align}
\label{eq:tpca-sign-mismatch-error}
    \left|\frac{\Xbar_{T_{\mathrm{av}},i}}{\beta_0}-(u_H)_i\right|
    &=\left|\frac{\Xbar_{T_{\mathrm{av}},i}}{\beta_0}\right|+|(u_H)_i|
    \ge |(u_H)_i|.
\end{align}

Set $\rho_N\doteq\min\{N^{-\vartheta},\lambda^2\sigma_N^2N^{-\vartheta/2}\}$ and $h_N\doteq\tfrac18\sigma_N\min\{\lambda\sigma_NN^{-\vartheta},\lambda^2\sigma_N^2N^{-\vartheta/2}\}$. Using \cref{eq:tpca-sign-mismatch-error}, we obtain
\begin{align}
\label{eq:tpca-sign-rounding-hamming-decomposition}
    \frac{1}{N}d_H\bigl(\sgn(\Xbar_{T_{\mathrm{av}}}),\sgn(u_H)\bigr)
    &=\frac{1}{N}\sum_{i=1}^N
        \ind{\sgn(\Xbar_{T_{\mathrm{av}},i})\ne\sgn((u_H)_i)}\nonumber\\
    &=\frac{1}{N}\sum_{i=1}^N
        \ind{\sgn(\Xbar_{T_{\mathrm{av}},i})\ne\sgn((u_H)_i)}
        \ind{|(u_H)_i|\le h_N}\nonumber\\
    &\quad+\frac{1}{N}\sum_{i=1}^N
        \ind{\sgn(\Xbar_{T_{\mathrm{av}},i})\ne\sgn((u_H)_i)}
        \ind{|(u_H)_i|>h_N}\nonumber\\
    &\le\frac{1}{N}\sum_{i=1}^N\ind{|(u_H)_i|\le h_N}
        +\frac{1}{Nh_N^2}
        \norm{\frac{\Xbar_{T_{\mathrm{av}}}}{\beta_0}-u_H}^2.
\end{align}

We control the first term in \cref{eq:tpca-sign-rounding-hamming-decomposition}. By standard bounds on the Gaussian density,
\begin{align*}
    \Prb\left[|(u_H)_i|\le h_N\right]
    &\le\frac{\lambda\sigma_NN^{-\vartheta}}{4\sqrt{2\pi}}
        \exp\left(-\frac{\lambda^2\sigma_N^2}{8}\right)
    \le\frac14N^{-\vartheta},\\
    \Prb\left[|(u_H)_i|\le h_N\right]
    &\le\frac{2h_N}{\sqrt{2\pi}\sigma_N}
    \le\frac14\lambda^2\sigma_N^2N^{-\vartheta/2}.
\end{align*}
Thus, $\Prb[|(u_H)_i|\le h_N]\le\rho_N/4$. Independence of the co-ordinates and a Chernoff bound implies
\begin{align}
\label{eq:tpca-small-gradient-coordinate-count}
    \Prb\left[
        \frac{1}{N}\sum_{i=1}^N\ind{|(u_H)_i|\le h_N}
        >\frac{\rho_N}{2}
    \right]
    \le\exp(-cN\rho_N).
\end{align}

For the second term in \cref{eq:tpca-sign-rounding-hamming-decomposition}, we use \cref{eq:tpca-averaging-normalized-errors,eq:tpca-averaging-error-scale} and $\lambda^2\sigma_N^2\gtrsim_kN^{-1/2}$ to conclude that the following holds with probability at least $1-\exp(-c_kN)$,
\begin{align}
\label{eq:tpca-sign-rounding-error-bound}
    \frac{1}{Nh_N^2}
    \norm{\frac{\Xbar_{T_{\mathrm{av}}}}{\beta_0}-u_H}^2
    \lesssim_k\frac1N+\frac{N^{-1-\vartheta}}{\lambda^2\sigma_N^2}
    \lesssim_k N^{-1}+N^{-1/2-\vartheta}
    =o(\rho_N).
\end{align}
Since $N\rho_N\gtrsim_kN^{(1-\vartheta)/2}$, combining \cref{eq:tpca-small-gradient-coordinate-count,eq:tpca-sign-rounding-error-bound} with \cref{eq:tpca-sign-rounding-hamming-decomposition} and taking a union bound, we conclude that
\begin{align*}
    \Prb\left[\frac{1}{N}d_H\bigl(\sgn(\Xbar_{T_{\mathrm{av}}}),\sgn(u_H)\bigr)\le\rho_N\right]
    \ge1-\exp(-c_kN^{(1-\vartheta)/2}),
\end{align*}
which proves the claim for odd $k$.

For even $k$, write $A_{\mathrm{sig}}=\lambda\gamma_N(\plant(\plant)^\top-I/N)$. By \cref{avg-discrete-hess-snr,eq:tpca-averaging-normalized-errors,eq:tpca-averaging-error-scale}, with probability at least $1-\exp(-c_kN)$,
\begin{align}
\label{eq:tpca-hessian-perturbation-bound}
    \frac{\norm{\Sigmabar_{T_{\mathrm{av}}}-(I+\beta_0A_{\mathrm{sig}})}_{\mathrm{op}}}{\beta_0\lambda\gamma_N}
    &\le\frac{\sqrt N\,\eps_N+\norm{W_N}_{\mathrm{op}}}{\lambda\gamma_N}
    \lesssim_k N^{-\vartheta}+\frac{N^{\tfrac{k-2}{4}}}{\lambda}.
\end{align}
The matrix $I+\beta_0A_{\mathrm{sig}}$ has top eigenvector $\plant$ and eigengap $\beta_0\lambda\gamma_N$. Applying the Davis--Kahan theorem~\citep[Theorem~4.5.5]{vershynin2018high} to $\Sigmabar_{T_{\mathrm{av}}}$ and $I+\beta_0A_{\mathrm{sig}}$, and using \cref{eq:tpca-hessian-perturbation-bound}, we obtain
\begin{align*}
    \min_{s\in\{\pm1\}}\norm{\widehat v-s\plant}
    &\lesssim\frac{\norm{\Sigmabar_{T_{\mathrm{av}}}-(I+\beta_0A_{\mathrm{sig}})}_{\mathrm{op}}}{\beta_0\lambda\gamma_N}
    \lesssim_k N^{-\vartheta}+\frac{N^{\tfrac{k-2}{4}}}{\lambda}.
\end{align*}
Fix $s\in\{\pm1\}$. If $\sgn(\widehat v_i)\ne s\xopt_i$, then
\begin{align*}
    |\widehat v_i-s\plant_i|
    &=|\widehat v_i|+|\plant_i|
    \ge\frac{1}{\sqrt N}.
\end{align*}
Squaring and summing over $i\in[N]$, we obtain
\begin{align*}
    \frac{1}{N}d_H\bigl(\sgn(\widehat v),s\xopt\bigr)
    &=\frac{1}{N}\sum_{i=1}^N\ind{\sgn(\widehat v_i)\ne s\xopt_i}\\
    &\le\sum_{i=1}^N|\widehat v_i-s\plant_i|^2
        \ind{\sgn(\widehat v_i)\ne s\xopt_i}\\
    &\le\norm{\widehat v-s\plant}^2.
\end{align*}
Consequently,
\begin{align*}
    1-\bigl|m(\sgn(\widehat v))\bigr|
    &=\frac{2}{N}\min_{s\in\{\pm1\}}d_H\bigl(\sgn(\widehat v),s\xopt\bigr)\\
    &\le2\min_{s\in\{\pm1\}}\norm{\widehat v-s\plant}^2
    \lesssim_k \frac{N^{\tfrac{k-2}{2}}}{\lambda^2}+N^{-2\vartheta}.
\end{align*}
A union bound proves the required claim for even $k$.
\end{proof}

%% file: strong_recovery_odd.tex
\section{Strong Recovery at the Algorithmic Threshold for Odd \texorpdfstring{$k$}{k}}
\label{sec:strong-recovery-odd}

Throughout this section, we assume that $k$ is odd. We use $(X_t)_{t\ge0}$ to denote the trajectory of Glauber dynamics for the tensor PCA Hamiltonian $H$ at a constant inverse temperature $\beta > 0$, and assume that the signal strength satisfies $\lambda\ge C_{\beta,k}N^\gamma$, where $\gamma\ge\tfrac{k-2}{4}$ and $C_{\beta,k}$ is a sufficiently large constant.

Our goal in this section is to prove the following theorem, which is a formal version of \cref{thm:glauber-odd-k-informal}. 

\begin{theorem}[Persistent Strong Recovery above the Algorithmic Threshold]
\label{thm:persistent-strong-recovery-algorithmic-threshold}
Let $k\ge3$ be odd, let $\beta>0$, and assume $\lambda\ge C_{\beta,k}N^{\tfrac{k-2}{4}}$, where $C_{\beta,k}$ is a sufficiently large constant. Fix $\vartheta\in(0,1)$ and $\eta\in(0,\nicefrac{1}{2})$, and let $(X_t)_{t\ge0}$ denote the trajectory of the second phase of \cref{alg:averaged-annealed-glauber}, initialized at $\tilX_0\sim\pi$ with inverse temperatures $\beta_0$ and $\beta$ for the two respective phases, where $\beta_0\asymp_k\lambda^{-1}N^{-\nicefrac{k}{2}-\vartheta}$ and $T_{\mathrm{av}}\asymp_kN^{2k+4\vartheta}$. Set $T_{\mathrm{str}}\doteq\log N$. Then, for all sufficiently large $N$ and every deterministic $T$ satisfying $T_{\mathrm{str}}\le T\le\poly(N)$,
\begin{align*}
    \Prb\left[\inf_{t\in[T_{\mathrm{str}},T]}m(X_t)\ge1-N^{-\eta}\right]
    \ge1-\exp\left(-c_{\beta,k}N^{\min\{(1-\vartheta)/2,\,1-2\eta\}}\right).
\end{align*}
\end{theorem}

The dynamics exhibits different behaviors, depending on how warm the start is.
From a sufficiently warm start, the dynamics is ``critical'': the drift induced by the signal completely dominates the noise, as we will show in \cref{subsec:boosting-critical-to-strong-recovery}.
However, for $\lambda$ just above the algorithmic threshold, the initialization provided by phase 1 of our algorithm is warm, but not warm enough for the easy dynamics to kick in. 
This ``subcritical'' regime is the topic of \cref{subsec:critical-recovery-warm-start,subsec:average-gradient-gauge,subsec:reverse-planting,subsec:critical-recovery-reverse-planting}, where we show that the initialization depends on the noise in a mild way, the signal induces a steady drift which eventually overwhelms the influence of the chaos term.
The complete proof is presented in \cref{subsec:proof-persistent-strong-recovery-algorithmic-threshold}.

\subsection{Boosting Critical Recovery to Persistent Strong Recovery}
\label{subsec:boosting-critical-to-strong-recovery}

Recall from \cref{subsec:sketch-reverse-planting} that the signal drift dominates both the entropic and disorder drifts, even in the worst case, whenever the correlation exceeds a large constant multiple of the \emph{critical correlation} $\lambda^{-\tfrac{1}{k-1}}$. Indeed, the signal drift is of order $\lambda m(X_t)^{k-1}$, whereas the disorder drift is of order one uniformly over the hypercube. We formalize this worst-case argument in \cref{thm:strong-recovery-after-critical-correlation} below, which shows that, once the process has attained critical correlation, running Glauber dynamics for an additional $\Theta(\log N)$ time suffices to achieve $1-\poly(\nicefrac{1}{N})$ correlation with high probability. Beyond that, the correlation stays above this level throughout any $\poly(N)$ time horizon.

\begin{theorem}[Persistent Strong Recovery after Critical Correlation]
\label{thm:strong-recovery-after-critical-correlation}
Fix $\eta,s\in(0,\nicefrac{1}{2})$, define $\bar r_N\doteq N^{-s}$, and assume that $\lambda\bar r_N^{k-1}\ge C_{\beta,k}$. Define the stopping time
\begin{align*}
    \tau_{C\bar r_N}\doteq\inf\{t\ge0:m(X_t)\ge C\bar r_N\},
\end{align*}
where $C\geq 1$ is a constant, and set $T_{\mathrm{str}}\doteq1+2\eta\log N$. Then, the following holds for all sufficiently large $N$, every deterministic $T$ satisfying $T_{\mathrm{str}}\le T\le\poly(N)$, and any initial state $X_0$, possibly depending on $H$,
\begin{align*}
    \Prb\left[\inf_{t\in[T_{\mathrm{str}},T]}m(X_{\tau_{C\bar r_N}+t})\ge1-N^{-\eta}\right]
    \ge1-\exp(-c_kN)-2\exp\left(-cN\min\{\bar r_N^2,N^{-2\eta}\}\right).
\end{align*}
\end{theorem}

\begin{proof}
    Write $q_N\doteq N^{-\eta}$ and $\eps_N\doteq\tfrac{1}{32}\min\{\bar r_N,q_N\}$. Throughout the proof, we assume that $N$ is sufficiently large and $C_{\beta, k}$ is a sufficiently large constant. By the strong Markov property~\citep[Chapter~4]{ethier1986markov} applied at the stopping time $\tau_{C\bar r_N}$, it suffices to consider an initial state $X_0$ with $m(X_0)\ge C\bar r_N$, possibly depending on $H$.

Since $\dgrad m=\xopt/N$ and $\tanh$ is odd, we infer from \cref{eq:srw-generator,eq:glauber-generator-srw-decomposition} that
\begin{align}
\label{eq:strong-recovery-overlap-drift}
    \cL_{H,\beta}m(x)
    &=-m(x)+\frac1N\sum_{i=1}^N\tanh\bigl(\beta\xopt_i\partial_iH(x)\bigr)
      =1-m(x)-2b(x),\\
    b(x)&\doteq\frac{1}{2N}\sum_{i=1}^N\Bigl[1-\tanh\bigl(\beta\xopt_i\partial_iH(x)\bigr)\Bigr].\nonumber
\end{align}
Let $M_t^m$ denote the Dynkin martingale associated with $m$. By \cref{lem:single-site-dynkin},
\begin{align}
\label{eq:strong-recovery-dynkin}
    M_t^m=m(X_t)-m(X_0)-\int_0^t\cL_{H,\beta}m(X_s)\,ds.
\end{align}
Consider the events
\begin{align*}
    \cE\doteq\left\{\max_{x\in\Cube}\norm{\dgrad H_0(x)}^2\le C_1N\right\},
    \qquad
    \calM_j\doteq\left\{\sup_{j\le t\le j+1}\bigl|M_t^m-M_j^m\bigr|\le\eps_N\right\},
\end{align*}
and let $\calM\doteq\bigcap_{j=0}^{\lceil T\rceil-1}\calM_j$. By \cref{lem:kspin-sobolev-norm-concentration}, we can choose $C_1$ such that $\Prb[\cE]\ge1-\exp(-c_kN)$. Moreover, since $\norm{m}_{\cD^1}\le2$ by \cref{lem:overlap-sobolev-bounds} and $\dgrad m=\xopt/N$, we infer from the Markov property, \cref{lem:dynkin-martingale-concentration}, and a union bound that
\begin{align}
\label{eq:strong-recovery-martingale-bound}
    \Prb[\calM^c]
    \le\sum_{j=0}^{\lceil T\rceil-1}\Prb[\calM_j^c]
    \le2\lceil T\rceil\exp\left(-cN\min\{\bar r_N^2,q_N^2\}\right).
\end{align}
Henceforth, we work on the event $\cE\cap\calM$.

We first express the deficit $1-m(X_t)$ in terms of $b(X_t)$. By \cref{eq:strong-recovery-overlap-drift,eq:strong-recovery-dynkin},
\begin{align*}
    1-m(X_t)+M_t^m
    =1-m(X_0)-\int_0^t\bigl[1-m(X_s)-2b(X_s)\bigr]\,ds.
\end{align*}
Since the integrand is bounded and measurable, the right-hand side is Lipschitz in $t$, and thus differentiable almost everywhere. Hence, for almost every $t\ge0$,
\begin{align*}
    \frac{d}{dt}\Bigl[e^t\bigl(1-m(X_t)+M_t^m\bigr)\Bigr]
    &=e^t\bigl(1-m(X_t)+M_t^m\bigr)
      +e^t\,\frac{d}{dt}\bigl(1-m(X_t)+M_t^m\bigr)\\
    &=e^t\bigl(1-m(X_t)+M_t^m\bigr)-e^t\bigl[1-m(X_t)-2b(X_t)\bigr]\\
    &=e^t\bigl[M_t^m+2b(X_t)\bigr].
\end{align*}
Integrating this identity from $0$ to $t$ and using $M_0^m=0$, we obtain
\begin{align*}
    e^t\bigl(1-m(X_t)+M_t^m\bigr)
    =1-m(X_0)+\int_0^te^s\bigl[M_s^m+2b(X_s)\bigr]\,ds.
\end{align*}
Multiplying both sides by $e^{-t}$ and rearranging, we conclude that
\begin{align*}
    1-m(X_t)
    =e^{-t}\bigl(1-m(X_0)\bigr)
     +\int_0^te^{-(t-s)}\,2b(X_s)\,ds
     +\int_0^te^{-(t-s)}M_s^m\,ds-M_t^m.
\end{align*}
Since we work on $\calM$, we infer from the triangle inequality that
\begin{align*}
    \bigl|M_t^m-M_s^m\bigr|
    &\le\bigl|M_t^m-M_{\lfloor t\rfloor}^m\bigr|
      +\sum_{l=\lfloor s\rfloor+1}^{\lfloor t\rfloor}\bigl|M_l^m-M_{l-1}^m\bigr|
      +\bigl|M_s^m-M_{\lfloor s\rfloor}^m\bigr|\\
    &\le\bigl(\lfloor t\rfloor-\lfloor s\rfloor+2\bigr)\eps_N
    \le(t-s+3)\eps_N,
    \qquad0\le s\le t\le T.
\end{align*}
Since $M_0^m=0$ and $\int_0^te^{-(t-s)}\,ds=1-e^{-t}$, the last two terms satisfy
\begin{align*}
    \left|\int_0^te^{-(t-s)}M_s^m\,ds-M_t^m\right|
    &=\left|\int_0^te^{-(t-s)}\bigl(M_s^m-M_t^m\bigr)\,ds-e^{-t}\bigl(M_t^m-M_0^m\bigr)\right|\\
    &\le\eps_N\int_0^te^{-(t-s)}(t-s+3)\,ds+\eps_Ne^{-t}(t+3)
    \le7\eps_N,
    \qquad0\le t\le T,
\end{align*}
where the last inequality uses $\int_0^\infty e^{-u}(u+3)\,du=4$ and $e^{-t}(t+3)\le3$. Consequently,
\begin{align}
\label{eq:strong-recovery-deficit-bound}
    1-m(X_t)
    \le e^{-t}\bigl(1-m(X_0)\bigr)+\int_0^te^{-(t-s)}\,2b(X_s)\,ds+7\eps_N,
    \qquad0\le t\le T.
\end{align}

We next bound $b(x)$ in terms of the correlation $m(x)$. Since $H=N\lambda m^k+H_0$, we have $\xopt_i\partial_iH(x)=N\lambda\xopt_i\partial_i(m^k)(x)+\xopt_i\partial_iH_0(x)$ for every $i\in[N]$. Define $m_{-i}(x)\doteq N^{-1}\sum_{j\ne i}\xopt_jx_j$. If $m(x)\ge4/N$, then $m_{-i}(x)\pm1/N\ge m(x)-2/N\ge m(x)/2$. Hence, by the definition of the discrete partial derivative and the mean value theorem,
\begin{align}
\label{eq:strong-recovery-signal-field}
    N\lambda\xopt_i\partial_i(m^k)(x)
    &=\frac{N\lambda}{2}\left[
        \left(m_{-i}(x)+\frac1N\right)^k
        -\left(m_{-i}(x)-\frac1N\right)^k
    \right]\nonumber\\
    &\ge k\lambda\left(\frac{m(x)}2\right)^{k-1}
    \ge c_k\lambda m(x)^{k-1}.
\end{align}
Let $\cI(x)\doteq\{i\in[N]:\xopt_i\partial_iH_0(x)\le-\tfrac{c_k}{2}\lambda m(x)^{k-1}\}$. Since we work on $\cE$,
\begin{align*}
    |\cI(x)|\,\frac{c_k^2\lambda^2m(x)^{2k-2}}{4}
    \le\sum_{i\in\cI(x)}\bigl(\partial_iH_0(x)\bigr)^2
    \le\norm{\dgrad H_0(x)}^2
    \le C_1N.
\end{align*}
Moreover, for every $i\notin\cI(x)$, we infer from \cref{eq:strong-recovery-signal-field} that $\xopt_i\partial_iH(x)>\tfrac{c_k}{2}\lambda m(x)^{k-1}$, and therefore
\begin{align*}
    \frac{1-\tanh\bigl(\beta\xopt_i\partial_iH(x)\bigr)}{2}
    =\frac{1}{1+\exp\bigl(2\beta\xopt_i\partial_iH(x)\bigr)}
    \le\exp\bigl(-c_k\beta\lambda m(x)^{k-1}\bigr).
\end{align*}
Bounding the same quantity by $1$ for $i\in\cI(x)$, we conclude that the following holds whenever $m(x)\ge4/N$:
\begin{align}
\label{eq:strong-recovery-update-error}
    b(x)
    &\le\frac{|\cI(x)|}{N}+\exp\bigl(-c_k\beta\lambda m(x)^{k-1}\bigr)\nonumber\\
    &\le\frac{C_k}{\lambda^2m(x)^{2k-2}}
        +\exp\bigl(-c_k\beta\lambda m(x)^{k-1}\bigr).
\end{align}
We observe that the right-hand side of \cref{eq:strong-recovery-update-error} is nonincreasing in $m(x)$.

We now show that the correlation stays above $C\bar r_N/2$ throughout $[0,T]$. Since $s<\nicefrac{1}{2}$, we have $\bar r_N\ge N^{-\nicefrac{1}{2}}$, and hence $C\bar r_N/2\ge4/N$. As the right-hand side of \cref{eq:strong-recovery-update-error} is nonincreasing in $m(x)$ and $\lambda \bar r_N^{k-1}\ge C_{\beta,k}$, we obtain the following whenever $m(x)\ge C\bar r_N/2$ after adjusting constants appropriately:
\begin{align}
\label{eq:strong-recovery-crude-rate}
    2b(x)
    \le\frac{C_k}{C^{2k-2}C_{\beta,k}^2}
      +2\exp\left(-c_k\beta C^{k-1}C_{\beta,k}\right)
    \le\frac{1}{16}.
\end{align}
Let $\rho\doteq\inf\{t\ge0:m(X_t)\le C\bar r_N/2\}$. For $s<\rho$, we have $m(X_s)>C\bar r_N/2$, and hence $2b(X_s)\le\nicefrac{1}{16}$ by \cref{eq:strong-recovery-crude-rate}. Consequently, \cref{eq:strong-recovery-deficit-bound} implies
\begin{align}
\label{eq:strong-recovery-crude-deficit}
    1-m(X_t)
    \le e^{-t}\bigl(1-m(X_0)\bigr)+\bigl(1-e^{-t}\bigr)\frac{1}{16}+7\eps_N,
    \qquad0\le t\le\rho\wedge T.
\end{align}
Since $1-m(X_0)\le1-C\bar r_N$ and $\nicefrac{1}{16}\le1-C\bar r_N$, it follows that
\begin{align*}
    1-m(X_t)\le1-C\bar r_N+7\eps_N<1-\frac{C\bar r_N}{2},
    \qquad0\le t\le\rho\wedge T,
\end{align*}
where the last inequality uses $7\eps_N\le7\bar r_N/32<C\bar r_N/2$. If $\rho\le T$, taking $t=\rho$ contradicts $m(X_\rho)\le C\bar r_N/2$. We conclude that $\rho>T$, so \cref{eq:strong-recovery-crude-deficit} holds for every $t\in[0,T]$.

We next show that the deficit falls below $q_N$ by time $T_{\mathrm{str}}$ and stays below $q_N$ up to time $T$. Since $\lambda^{-2}\le C_{\beta,k}^{-2}N^{-2\gamma}$ with $2\gamma\ge\tfrac{k-2}{2}\ge\nicefrac{1}{2}>\eta$, and $\exp(-c_k\beta\lambda)=o(N^{-\eta})$, we infer from \cref{eq:strong-recovery-update-error} that, whenever $m(x)\ge\nicefrac{1}{3}$,
\begin{align}
\label{eq:strong-recovery-sharp-rate}
    2b(x)\le\frac{C_k}{\lambda^2}+2\exp(-c_k\beta\lambda)\le\frac{q_N}{64}.
\end{align}
By \cref{eq:strong-recovery-crude-deficit},
\begin{align*}
    1-m(X_t)\le e^{-1}+\frac{1}{16}+7\eps_N<\frac23,
    \qquad1\le t\le T,
\end{align*}
where we used $\eps_N=o_N(1)$. Consequently, $m(X_t)>\nicefrac{1}{3}$ for every $t\in[1,T]$. Splitting the integral in \cref{eq:strong-recovery-deficit-bound} at $s=1$, we obtain
\begin{align}
\label{eq:strong-recovery-persistence}
    1-m(X_t)
    &\le e^{-t}\bigl(1-m(X_0)\bigr)
      +\int_0^1e^{-(t-s)}\,2b(X_s)\,ds\nonumber\\
    &\qquad+\int_1^te^{-(t-s)}\,2b(X_s)\,ds
      +7\eps_N\nonumber\\
    &\le e^{-t}
      +\frac{1}{16}\int_0^1e^{-(t-s)}\,ds
      +\frac{q_N}{64}\int_1^te^{-(t-s)}\,ds
      +\frac{7q_N}{32}\nonumber\\
    &\le q_N^2+\frac{q_N^2}{16}+\frac{q_N}{64}+\frac{7q_N}{32}
    <q_N,
    \qquad T_{\mathrm{str}}\le t\le T.
\end{align}
Here, the second inequality uses \cref{eq:strong-recovery-crude-rate,eq:strong-recovery-sharp-rate}.

By \cref{eq:strong-recovery-persistence}, we have $m(X_t)>1-q_N$ for every $t\in[T_{\mathrm{str}},T]$ on $\cE\cap\calM$. Taking a union bound and using \cref{eq:strong-recovery-martingale-bound}, we obtain
\begin{align*}
    \Prb\left[\inf_{t\in[T_{\mathrm{str}},T]}m(X_t)<1-q_N\right]
    \le\Prb[\cE^c]+\Prb[\calM^c]
    \le\exp(-c_kN)+2\lceil T\rceil\exp\left(-cN\min\{\bar r_N^2,q_N^2\}\right).
\end{align*}
Since $T\le\poly(N)$ and $N\min\{\bar r_N^2,q_N^2\}/\log N\to\infty$ for $s,\eta<\nicefrac{1}{2}$, the proof is concluded by absorbing the factor $\lceil T\rceil$ after adjusting constants appropriately.
\end{proof}

A direct application of \cref{thm:tpca-averaged-differential-recovery,thm:strong-recovery-after-critical-correlation} establishes persistent strong recovery for \cref{alg:averaged-annealed-glauber} in the regime $\lambda\gtrsim N^{\tfrac{k-2}{4}+\tfrac{1}{4k}}$, which is above the algorithmic threshold by a factor of $N^{\tfrac{1}{4k}}$. Specifically, we prove the following.

\begin{corollary}[Persistent Strong Recovery Slightly above the Algorithmic Threshold]
\label{cor:averaged-glauber-persistent-strong-recovery}

Suppose $\gamma\ge\tfrac{k-2}{4}+\tfrac{1}{4k}$. Fix $\vartheta\in(0,1)$ and $\eta\in(0,\nicefrac{1}{2})$. Let $(X_t)_{t\ge0}$ denote the trajectory of the second phase of \cref{alg:averaged-annealed-glauber}, initialized at $\tilX_0\sim\pi$ with inverse temperatures $\beta_0$ and $\beta$ for the two respective phases, where $\beta_0$ and $T_{\mathrm{av}}$ are chosen as in \cref{thm:tpca-averaged-differential-recovery} for this choice of $\vartheta$. Set $T_{\mathrm{str}}\doteq1+2\eta\log N$. Then, for all sufficiently large $N$ and every deterministic $T$ satisfying $T_{\mathrm{str}}\le T\le\poly(N)$,

\begin{align*}
    \Prb\left[\inf_{t\in[T_{\mathrm{str}},T]}m(X_t)\ge1-N^{-\eta}\right]
    \ge1-\exp\left(-c_kN^{\min\{(1-\vartheta)/2,\,1-2\eta\}}\right).
\end{align*}

\end{corollary}

\begin{proof}

Define $a\doteq\tfrac{k-1}{4k}$ and note that $\gamma-\tfrac{k-1}{4}\ge-a$, so that $\lambda\sigma_N\gtrsim_kC_{\beta,k}N^{-a}$. Combining \cref{thm:tpca-averaged-differential-recovery} with \cref{cor:avg-discrete-grad-correlation}, we obtain the following with probability at least $1-\exp(-c_kN^{(1-\vartheta)/2})$, for sufficiently large $N$ and after choosing $C_{\beta,k}$ sufficiently large:
\begin{align*}
    m(X_0)
    &\ge m(\sgn(u_H))-\frac{2}{N}d_H\bigl(X_0,\sgn(u_H)\bigr)\\
    &\ge c_k\min\{1,\lambda\sigma_N\}-2\min\{N^{-\vartheta},\lambda^2\sigma_N^2N^{-\vartheta/2}\}\\
    &\ge\frac{c_k}{2}\min\{1,\lambda\sigma_N\}
    \ge 2N^{-a}.
\end{align*}
For the third inequality, we use $\lambda^2\sigma_N^2N^{-\vartheta/2}\le\lambda\sigma_NN^{-\vartheta/2}$ when $\lambda\sigma_N\le1$, and $N^{-\vartheta}=o(1)$ when $\lambda\sigma_N>1$.

By construction, $(X_t)_{t\ge0}$ is Glauber dynamics for $H$ at inverse temperature $\beta$, started from $X_0=\sgn(\Xbar_{T_{\mathrm{av}}})$. As \cref{thm:strong-recovery-after-critical-correlation} allows the initial state to depend on $H$, we may apply it to $(X_t)_{t\ge0}$ with $s=a$ and $C=1$. Since $\gamma\ge(k-1)a$, the correlation level $\bar r_N=N^{-a}$ satisfies $\lambda\bar r_N^{k-1}\ge C_{\beta,k}$, and $\tau_{\bar r_N}=0$ on the above event. Finally, since $1-2a=\tfrac{k+1}{2k}>\nicefrac{1}{2}>\tfrac{1-\vartheta}{2}$, the proof is concluded via a union bound after adjusting constants appropriately.
\end{proof}

\subsection{Achieving Critical Recovery above the Algorithmic Threshold from a Sufficiently Warm Start}
\label{subsec:critical-recovery-warm-start}

We now focus our attention on strong recovery in the regime $N^{\tfrac{k-2}{4}}\lesssim\lambda\lesssim N^{\tfrac{k-2}{4}+\tfrac{1}{4k}}$. Specifically, we set $\lambda=C_{\beta,k}N^\gamma$ with $\tfrac{k-2}{4}\le\gamma<\tfrac{k-2}{4}+\tfrac{1}{4k}$, and define the critical correlation parameter as $r_N\doteq\lambda^{-\tfrac{1}{k-1}}\asymp N^{-\tfrac{\gamma}{k-1}}$. In this regime, \cref{cor:avg-discrete-grad-correlation} suggests that the warm start computed by the first phase of \cref{alg:averaged-annealed-glauber} attains only subcritical correlation, of order $\lambda N^{-\tfrac{k-1}{4}}\ll r_N$.

In view of \cref{thm:strong-recovery-after-critical-correlation}, it suffices to show that the second phase of \cref{alg:averaged-annealed-glauber} attains critical correlation in polynomial time. To this end, we prove a stronger guarantee in \cref{thm:critical-recovery-warm-start}, which establishes that, simultaneously for all initializations sufficiently close to the sign-rounded average gradient, Glauber dynamics attains critical correlation within $O(1)$ time.

For $e\in\Cube$, let $(X_t(e))_{t\ge0}$ denote the trajectory of Glauber dynamics for $H$ at inverse temperature $\beta$, initialized at $X_0(e)=e$. Since $k$ is odd, we recall from \cref{avg-discrete-grad-snr} that the average gradient $u_H=\E_\pi[\dgrad H]$ is distributed as $\cN(\lambda\sigma_N^2\xopt,\sigma_N^2I)$, where $\sigma_N^2\asymp_kN^{-\tfrac{k-1}{2}}$.

\begin{theorem}[Critical Recovery from a Sufficiently Warm Start]
\label{thm:critical-recovery-warm-start}
Suppose $\tfrac{k-2}{4}\le\gamma\le\tfrac{k-2}{4}+\tfrac{1}{4k}$. Fix $\vartheta\in(0,\nicefrac{1}{2})$, and let $u_H=\E_\pi[\dgrad H]$. For every $e\in\Cube$, define the stopping time
\begin{align*}
    T_*(e)\doteq\inf\bigl\{t\ge0:m(X_t(e))\ge Cr_N\bigr\}.
\end{align*}
Then, the following holds for all sufficiently large $N$:
\begin{align*}
    \Prb\left[T_*(e)\le C_1N^{-(k-1)\left(\gamma-\tfrac{k-2}{4}\right)}
    \ \ \forall\,e\in\cB_H\bigl(\sgn(u_H),\lambda^2\sigma_N^2N^{-\vartheta}\bigr)\right]
    \ge1-\exp\bigl(-cN\lambda^2\sigma_N^2\bigr),
\end{align*}
where $C_1,c>0$ are constants depending only on $\beta$, $k$ and $C$.
\end{theorem}

The key challenge in establishing \cref{thm:critical-recovery-warm-start} is that the initialization may depend nontrivially on the disorder $H_0$. For $\tfrac{k-2}{4}\le\gamma<\tfrac{k-2}{4}+\tfrac{1}{4k}$, this dependence can significantly complicate the analysis, since the correlation attained by the warm start is too weak to ensure that the contribution of the signal term $N\lambda m^k$ uniformly dominates that of $H_0$. Indeed, as we already discuss in \cref{subsec:sketch-reverse-planting}, one can imagine a potentially dangerous initialization with nontrivial but subcritical correlation that lies close to one of the exponentially many local maxima of $H_0$~\citep{benarous2019landscape,deoliveira1997landscape,ros2019complex}. In the local neighborhood of such an initialization, the drift induced by $H_0$ could dominate that of the signal, thereby trapping the dynamics at polynomial timescales. 

Our analysis rules out such a scenario for any initialization sufficiently close to $\sgn(u_H)$ by proving that the dependence of $u_H$ on the disorder is benign. Indeed, by Stein's lemma, $u_H$ depends on $H_0$ only through its first Fourier level. Since $H_0$ is an isotropic Gaussian process, a Gaussian conditioning argument shows that $u_H$ is independent of the higher Fourier levels of $H_0$, which comprise the bulk of the disorder. Our proof exploits this independence through an intricate dynamical change of measure argument, which proves that for any initialization sufficiently close to $\sgn(u_H)$, the drift contribution of the signal term dominates that of the disorder until the dynamics achieves critical recovery. In \cref{subsec:average-gradient-gauge}, we formalize the independence of $u_H$ from the higher Fourier levels of $H_0$ via Gaussian conditioning. \Cref{subsec:reverse-planting} develops the dynamical change of measure argument which we call the reverse planting trick. Finally, we prove \cref{thm:critical-recovery-warm-start} in \cref{subsec:critical-recovery-reverse-planting}.

\subsection{The Average Gradient Gauge Transformation}
\label{subsec:average-gradient-gauge}

The following lemma, which is a straightforward consequence of Gaussian conditioning, allows us to control the dependence of $u_H$ on the disorder.

\begin{lemma}[Decomposition of the $k$-Spin Hamiltonian]
\label{lem:kspin-gaussian-conditioning}
Let $k\ge3$ be odd, and let $\zeta\doteq\E_\pi[\dgrad H_0]$. Then, $H_0(x)=\dotp{\zeta}{x}+H_0^\perp(x)$, where $H_0^\perp:\Cube\to\R$ is a centered Gaussian process independent of $\zeta$ with
\begin{align*}
    \E\bigl[H_0^\perp(x)H_0^\perp(y)\bigr]
    =N\left(\frac{\dotp{x}{y}}{N}\right)^k-\sigma_N^2\dotp{x}{y}.
\end{align*}
\end{lemma}

\begin{proof}
By \cref{lem:hypercube-stein}, $\zeta=\E_{y\sim\pi}[yH_0(y)]$. Hence, for every $x\in\Cube$ and $i\in[N]$, we obtain the following from \cref{eq:kspin-covariance} and \cref{eq:average-gradient-variance-moment} in the proof of \cref{avg-discrete-grad-snr}:
\begin{align*}
    \E[H_0(x)\zeta_i]
    &=\E_{y\sim\pi}\bigl[y_i\E[H_0(x)H_0(y)]\bigr]\\
    &=N^{1-k}\E_{y\sim\pi}\left[y_i\left(\sum_{j=1}^Nx_jy_j\right)^k\right]\\
    &=x_iN^{1-k}\E_{z\sim\pi}\left[z_i\left(\sum_{j=1}^Nz_j\right)^k\right]
    =\sigma_N^2x_i,
\end{align*}
where we used the change of variables $z_j=x_jy_j$ and permutation invariance of $\pi$.

Since $k\ge3$ is odd, \cref{eq:average-gradient-variance-moment} implies that $\sigma_N^2>0$. As $(H_0(x))_{x\in\Cube}$ and $\zeta$ are jointly Gaussian and $\E[\zeta\zeta^\top]=\sigma_N^2I$, the Gaussian conditioning formula~\citep[Theorems~9.1 and~9.9]{janson1997gaussian} implies that $\E[H_0(x)\mid\zeta]=\dotp{\zeta}{x}$ and that $H_0^\perp(x)\doteq H_0(x)-\dotp{\zeta}{x}$ is a centered Gaussian process independent of $\zeta$. 
Moreover, we obtain 
\begin{align*}
    \E\bigl[H_0^\perp(x)H_0^\perp(y)\bigr]
    &=\E\bigl[(H_0(x)-\dotp{\zeta}{x})(H_0(y)-\dotp{\zeta}{y})\bigr]\\
    &=\E[H_0(x)H_0(y)]-\sum_{i=1}^Ny_i\E[H_0(x)\zeta_i]-\sum_{i=1}^Nx_i\E[H_0(y)\zeta_i]+\sum_{i,j=1}^Nx_iy_j\E[\zeta_i\zeta_j]\\
    &=N\left(\frac{\dotp{x}{y}}{N}\right)^k-2\sigma_N^2\sum_{i=1}^Nx_iy_i+\sigma_N^2\sum_{i=1}^Nx_iy_i\\
    &=N\left(\frac{\dotp{x}{y}}{N}\right)^k-\sigma_N^2\dotp{x}{y},
\end{align*}
which concludes the proof.
\end{proof}

By \cref{lem:hypercube-stein}, $\dotp{\zeta}{x}$ is the first Fourier level of $H_0$. Since $u_H=\lambda\sigma_N^2\xopt+\zeta$ by \cref{avg-discrete-grad-snr}, \cref{lem:kspin-gaussian-conditioning} shows that $u_H$ is independent of the higher Fourier levels $H_0^\perp$, which comprise the bulk of the disorder. However, the initialization considered in \cref{thm:critical-recovery-warm-start} still depends on the disorder through $\sgn(u_H)$. To simplify the analysis, we apply a randomized gauge transformation, which we call the \emph{average gradient gauge transformation}. This transformation maps $\sgn(u_H)$ to the all-ones vector $\mathbf{1}$, so that the transformed dynamics starts from a fixed state and the disorder dependence of the initialization is transferred to the gauge-transformed Hamiltonian. Our approach is motivated by the Nishimori gauge transformation in the spin glass literature~\citep{nishimori2001statistical}, except that we map the initial state of the dynamics, rather than the planted signal, to $\mathbf{1}$. Precisely, we define the following random variables for $i\in[N]$:
\begin{align}
\label{eq:gauge-variables}
    \xi_i\doteq\xopt_i\zeta_i,
    \qquad
    U_i\doteq\lambda\sigma_N^2+\xi_i,
    \qquad
    \phi_i\doteq\sgn(U_i),
    \qquad
    R_i\doteq\abs{U_i}.
\end{align}
Since $\xopt\in\Cube$ is fixed, $\xi\sim\cN(0,\sigma_N^2I)$. Furthermore, $u_H=\lambda\sigma_N^2\xopt+\zeta$, and hence $u_H=\xopt\odot U$ and, almost surely, $\sgn(u_H)=\xopt\odot\phi$, where $\odot$ denotes the entrywise product. Finally, we formally define the average-gradient gauge transformation and the associated gauge-transformed dynamics as follows:
\begin{align*}
    y=y(x)\doteq\xopt\odot\phi\odot x,
    \qquad
    Y_t(e)\doteq\xopt\odot\phi\odot X_t(\xopt\odot\phi\odot e).
\end{align*}
Observe that $y(\sgn(u_H))=\mathbf{1}$ and $Y_t(\mathbf{1})=\xopt\odot\phi\odot X_t(\sgn(u_H))$. The following lemma analyzes the effect of the gauge transformation on the tensor PCA Hamiltonian.

\begin{lemma}[Gauge-Transformed Hamiltonian]
\label{lem:tpca-hamiltonian-under-gauge}
Let $m_\phi(y)\doteq m(\xopt\odot\phi\odot y)$ and $H_\phi(y)\doteq H(\xopt\odot\phi\odot y)$ denote the gauge-transformed correlation and the gauge-transformed tensor PCA Hamiltonian, respectively. Then, $m_\phi(y)=\tfrac{1}{N}\dotp{\phi}{y}$ and $H_\phi=H_{\mathrm{null}}+H_{\mathrm{sig}}$, where
\begin{align*}
    H_{\mathrm{null}}(y)
    &\doteq\dotp{R}{y}+W(y),
    \qquad
    W(y)\doteq H_0^\perp(\xopt\odot\phi\odot y),\\
    H_{\mathrm{sig}}(y)
    &\doteq N\lambda m_\phi(y)^k-N\lambda\sigma_N^2m_\phi(y).
\end{align*}
Moreover, the following hold:
\begin{enumerate}
    \item The Gaussian process $W$ is independent of $(\phi,R)$ and has the same law as $H_0^\perp$.
    \item For any $e\in\Cube$, the trajectory $Y_t(e)$ has the same law as Glauber dynamics for the gauge-transformed Hamiltonian $H_\phi$ at inverse temperature $\beta$, initialized at $Y_0=e$.
\end{enumerate}
\end{lemma}

\begin{proof}
Since $\xopt\in\Cube$, by direct computation,
\begin{align*}
    m_\phi(y)
    =m(\xopt\odot\phi\odot y)
    =\frac{1}{N}\dotp{\xopt}{\xopt\odot\phi\odot y}
    =\frac{1}{N}\sum_{i=1}^N(\xopt_i)^2\phi_iy_i
    =\frac{1}{N}\dotp{\phi}{y}.
\end{align*}
Next, since $u_H=\lambda\sigma_N^2\xopt+\zeta$ and $\dotp{\xopt}{x}=Nm(x)$, we infer from \cref{eq:tensor-pca-hamiltonian,lem:kspin-gaussian-conditioning} that
\begin{align*}
    H(x)
    &=N\lambda m(x)^k+\dotp{\zeta}{x}+H_0^\perp(x)\\
    &=\dotp{u_H}{x}+N\lambda m(x)^k-N\lambda\sigma_N^2m(x)+H_0^\perp(x).
\end{align*}
Moreover, since $u_H=\xopt\odot U$,
\begin{align*}
    \dotp{u_H}{\xopt\odot\phi\odot y}
    =\sum_{i=1}^NU_i\phi_iy_i
    =\sum_{i=1}^NR_iy_i
    =\dotp{R}{y}.
\end{align*}
Consequently,
\begin{align*}
    H_\phi(y)
    &=H(\xopt\odot\phi\odot y)\\
    &=\dotp{R}{y}+N\lambda m_\phi(y)^k-N\lambda\sigma_N^2m_\phi(y)+H_0^\perp(\xopt\odot\phi\odot y)\\
    &=H_{\mathrm{null}}(y)+H_{\mathrm{sig}}(y).
\end{align*}

We next prove the first claim. Since $\xopt\in\Cube$ is fixed, $(\phi,R)$ is a measurable function of $\zeta$. Moreover, by \cref{lem:kspin-gaussian-conditioning}, $H_0^\perp$ is a centered Gaussian process independent of $\zeta$. Since the map $y\mapsto\xopt\odot\phi\odot y$ is an isometry of $\Cube$ and the covariance of $H_0^\perp$ depends only on $\dotp{x}{y}$, the law of $H_0^\perp$ is invariant under this transformation. Hence,
\begin{align*}
    \Law(W\mid\phi,R)
    =\Law\bigl(H_0^\perp(\xopt\odot\phi\odot\cdot)\mid\phi,R\bigr)
    =\Law(H_0^\perp\mid\phi,R)
    =\Law(H_0^\perp).
\end{align*}
Thus, $W$ is independent of $(\phi,R)$ and has the same law as $H_0^\perp$.

The second claim is a standard consequence of the gauge invariance of Glauber dynamics (see \citep[\S4.10]{nishimori2001statistical} for a proof in the special case of the Edwards--Anderson model). We present a proof below for completeness. Define the gauge transformation operator $\Lambda_\phi:\Cube\to\Cube$ as $\Lambda_\phi x\doteq\xopt\odot\phi\odot x$, and for any $f:\Cube\to\R$, define $\Lambda_\phi f\doteq f\circ\Lambda_\phi$. Observe that $\Lambda_\phi^2=I$, that $\Lambda_\phi$ commutes with coordinate flips, i.e., $\Lambda_\phi x^{\sim i}=(\Lambda_\phi x)^{\sim i}$, and that $H_\phi=\Lambda_\phi H$. Recall from \cref{eq:glauber-generator} that the generator of Glauber dynamics for $H$ at inverse temperature $\beta$ is
\begin{align*}
    \cL_{H,\beta}F(x)
    =\sum_{i=1}^Nc_i^{(H,\beta)}(x)\bigl[F(x^{\sim i})-F(x)\bigr],
    \qquad
    c_i^{(H,\beta)}(x)
    =\frac12\bigl[1-x_i\tanh(\beta\partial_iH(x))\bigr].
\end{align*}
Since $2x_i\partial_iH(x)=H(x)-H(x^{\sim i})$, the flip rate $c_i^{(H,\beta)}(x)$ depends on $x$ only through $H(x^{\sim i})-H(x)$. Moreover, for any $x\in\Cube$,
\begin{align*}
    H_\phi\bigl((\Lambda_\phi x)^{\sim i}\bigr)-H_\phi(\Lambda_\phi x)
    &=H_\phi(\Lambda_\phi x^{\sim i})-H_\phi(\Lambda_\phi x)\\
    &=H(\Lambda_\phi^2x^{\sim i})-H(\Lambda_\phi^2x)
    =H(x^{\sim i})-H(x),
\end{align*}
and hence, $c_i^{(H_\phi,\beta)}(\Lambda_\phi x)=c_i^{(H,\beta)}(x)$. It follows that for any $F:\Cube\to\R$,
\begin{align*}
    (\cL_{H,\beta}\Lambda_\phi F)(x)
    &=\sum_{i=1}^Nc_i^{(H,\beta)}(x)\bigl[F(\Lambda_\phi x^{\sim i})-F(\Lambda_\phi x)\bigr]\\
    &=\sum_{i=1}^Nc_i^{(H_\phi,\beta)}(\Lambda_\phi x)\bigl[F\bigl((\Lambda_\phi x)^{\sim i}\bigr)-F(\Lambda_\phi x)\bigr]
    =(\Lambda_\phi\cL_{H_\phi,\beta}F)(x).
\end{align*}
Hence, $\Lambda_\phi\cL_{H,\beta}\Lambda_\phi=\cL_{H_\phi,\beta}$. Finally, since $Y_t(e)=\Lambda_\phi X_t(\Lambda_\phi e)$, for any $F:\Cube\to\R$,
\begin{align*}
    \left.\frac{d}{dt}\E\bigl[F(Y_t(e))\bigr]\right|_{t=0}
    &=\left.\frac{d}{dt}\E\bigl[F\bigl(\Lambda_\phi X_t(\Lambda_\phi e)\bigr)\bigr]\right|_{t=0}\\
    &=(\cL_{H,\beta}\Lambda_\phi F)(\Lambda_\phi e)
    =(\Lambda_\phi\cL_{H,\beta}\Lambda_\phi F)(e)
    =(\cL_{H_\phi,\beta}F)(e),
\end{align*}
where the expectations are taken with the Hamiltonian $H$ and initial state $e$ held fixed.
Thus, $\cL_{H_\phi,\beta}$ is the generator of $Y_t(e)$, and hence $Y_t(e)$ is Glauber dynamics for the gauge-transformed Hamiltonian $H_\phi$ at inverse temperature $\beta$, initialized at $Y_0=e$.
\end{proof}

\Cref{lem:tpca-hamiltonian-under-gauge} shows that, modulo the lower order linear terms $\dotp{R}{y}$ and $-N\lambda\sigma_N^2m_\phi(y)$, the average gradient gauge transformation turns $H$ into a tensor PCA Hamiltonian with planted vector $\phi$ and disorder $W$, which is independent of $(\phi,R)$. The decomposition $H_\phi=H_{\mathrm{null}}+H_{\mathrm{sig}}$ is constructed so that the entire dependence of $H_\phi$ on $\phi$ resides in $H_{\mathrm{sig}}$. To this end, $H_{\mathrm{null}}$ contains the first Fourier level of $H$ (which takes the form $\dotp{R}{y}$ after the gauge transformation), including that of the signal.

\subsection{The Reverse Planting Trick and Adaptive Contiguity}
\label{subsec:reverse-planting}

We use $\cL_{\mathrm{pl}}\doteq\cL_{H_\phi,\beta}$ and $\cL_{\mathrm{null}}\doteq\cL_{H_{\mathrm{null}},\beta}$ to denote the generators of Glauber dynamics at inverse temperature $\beta$ for $H_\phi$ and $H_{\mathrm{null}}$, respectively. By Dynkin's formula, the quantity $(\cL_{\mathrm{null}}m_\phi)(Y_t(e))$ encapsulates the contribution of the disorder to the time evolution of the correlation $m_\phi$ along the Glauber trajectory. Hence, in order to show that the signal contribution dominates that of the disorder until the dynamics achieves critical recovery, our analysis needs to control $(\cL_{\mathrm{null}}m_\phi)(Y_t(e))$ along the gauge-transformed trajectory.
Explicitly, by \cref{eq:glauber-generator}, we have the following for every $y\in\Cube$.
\begin{align*}
    (\cL_{\mathrm{null}}m_\phi)(y)
    =\frac{1}{N}\sum_{i=1}^N\phi_i\bigl[\tanh\bigl(\beta\partial_iH_{\mathrm{null}}(y)\bigr)-y_i\bigr].
\end{align*}
Although the coordinates of $\phi$ are marginally i.i.d.\ biased Rademacher variables, the above quantity is hard to control since $\phi$ is correlated with $\partial_iH_{\mathrm{null}}(y)=R_i+\partial_iW(y)$ through $R$ as per \cref{eq:gauge-variables}. Moreover, the trajectory $Y_t(e)$ along which it is evaluated depends nontrivially on $\phi$ through $H_{\mathrm{sig}}$.

We circumvent these difficulties via a change of measure to a carefully constructed null law, under which the dynamics is independent of $\phi$. For a fixed $e \in \Cube$, we formally define both the law of the true dynamics (which we call the planted law) and this null law, representing them as laws of Glauber dynamics in the random environment $(\phi,R,W)$.

\begin{definition}[Planted Law]
\label{def:planted-law}
The planted law $\Prb$ is the law of Glauber dynamics in a random environment, which factorizes as $\Prb=\Prb_{\mathrm{env}}(\phi,R,W)\cdot\Prb_{\mathrm{dyn}}(\cdot\mid\phi,R,W)$, where $\Prb_{\mathrm{env}}$ and $\Prb_{\mathrm{dyn}}$ are defined as follows:
\begin{enumerate}
    \item $(\phi_i,R_i)=(\sgn(U_i),\abs{U_i})$, where $U_i\iidsim\cN(\lambda\sigma_N^2,\sigma_N^2)$ for $i\in[N]$.
    \item $W$ is independent of $(\phi,R)$ and has the same law as $H_0^\perp$.
    \item Conditional on the environment $(\phi,R,W)$, $\Prb_{\mathrm{dyn}}$ is the law of Glauber dynamics at inverse temperature $\beta$ for the Hamiltonian $H_\phi(y)=\dotp{R}{y}+W(y)+N\lambda m_\phi(y)^k-N\lambda\sigma_N^2m_\phi(y)$, initialized at $Y_0=e$. Its generator is $\cL_{\mathrm{pl}}$.
\end{enumerate}
\end{definition}

\begin{definition}[Null Law]
\label{def:null-law}
The null law $\Q$ is the law of Glauber dynamics in a random environment, which factorizes as $\Q=\Q_{\mathrm{env}}(\phi,R,W)\cdot\Q_{\mathrm{dyn}}(\cdot\mid\phi,R,W)$, where $\Q_{\mathrm{env}}$ and $\Q_{\mathrm{dyn}}$ are defined as follows:
\begin{enumerate}
    \item $\phi_i\iidsim\Unif(\{\pm1\})$ and $R_i=\abs{U_i}$, where $U_i\iidsim\cN(0,\sigma_N^2)$ for $i\in[N]$, independently of $\phi$.
    \item $W$ is independent of $(\phi,R)$ and has the same law as $H_0^\perp$.
    \item Conditional on the environment $(\phi,R,W)$, $\Q_{\mathrm{dyn}}$ is the law of Glauber dynamics at inverse temperature $\beta$ for the Hamiltonian $H_{\mathrm{null}}(y)=\dotp{R}{y}+W(y)$, initialized at $Y_0=e$. Its generator is $\cL_{\mathrm{null}}$.
\end{enumerate}
\end{definition}

Under both laws, we write $Y_t=Y_t(e)$ for the associated trajectory. By \cref{lem:tpca-hamiltonian-under-gauge}, $\Prb$ is the law of the true dynamics in the gauge-transformed coordinates. The two laws are best understood by inverting the gauge transformation, which, for simplicity, we do only for $e=\mathbf{1}$. Under $\Prb$, $\xopt\odot\phi\odot Y_t(\mathbf{1})$ is Glauber dynamics for $H$ initialized at its sign-rounded average gradient $\sgn(u_H)=\xopt\odot\phi$. Similarly, under $\Q$, $\xopt\odot\phi\odot Y_t(\mathbf{1})$ is Glauber dynamics for the Hamiltonian
\begin{align*}
    \widetilde H(x)
    \doteq H_{\mathrm{null}}(\xopt\odot\phi\odot x)
    =\dotp{\widetilde\zeta}{x}+W(\xopt\odot\phi\odot x),
    \qquad
    \widetilde\zeta\doteq\xopt\odot\phi\odot R,
\end{align*}
initialized at $\xopt\odot\phi=\sgn(\widetilde\zeta)$. Moreover, $\widetilde\zeta\sim\cN(0,\sigma_N^2I)$ and $W(\xopt\odot\phi\odot\cdot)$ are independent under $\Q$, and the latter is identically distributed as $H_0^\perp$ by isotropy. Applying \cref{lem:kspin-gaussian-conditioning}, we observe that $\Law(\widetilde H,\widetilde\zeta)=\Law(H_0,\zeta)$, where $\zeta=\E_\pi[\dgrad H_0]$, and consequently, $\widetilde\zeta$ is the average gradient of $\widetilde H$.

Thus, we conclude that both the planted and null laws describe Glauber dynamics initialized at the sign-rounded average gradient of their respective underlying Hamiltonians, which correspond to the true tensor PCA Hamiltonian under $\Prb$ and the pure disorder Hamiltonian under $\Q$. Under this setup, $\Q$ corresponds precisely to $\Prb$ with $\lambda=0$, as the signs of $U_i\iidsim\cN(0,\sigma_N^2)$ are then uniform and independent of their magnitudes, and $H_{\mathrm{sig}}$ vanishes. In particular, the null dynamics carries no information about $\xopt$.

The independence of $\phi$ and the Glauber trajectory under the null law yields the following uniform concentration guarantee for $\cL_{\mathrm{null}}m_\phi$.

\begin{lemma}[Uniform Concentration along the Null Trajectory]
\label{lem:null-trajectory-uniform-concentration}
For every $T>0$ and $e\in\Cube$,
\begin{align*}
    \E_{\Q}\left[\exp\left(\frac{N}{16}\sup_{t\in[0,T]}\bigl|(\cL_{\mathrm{null}}m_\phi)(Y_t(e))\bigr|^2\right)\right]
    \le\sqrt2(1+NT).
\end{align*}
\end{lemma}

\begin{proof}
Throughout, we work under $\Q$ and write $Y_t=Y_t(e)$. Let $(\cF_t)_{0\le t\le T}$ denote the filtration of Glauber dynamics up to time $T$ defined in \cref{subsec:cts-glauber}. By \cref{def:null-law}, conditional on $R,W,\cF_T$, the coordinates of $\phi$ remain i.i.d.\ Rademacher variables, and the trajectory $(Y_t)_{0\le t\le T}$ is $\cF_T$-measurable. Moreover, by \cref{eq:glauber-generator},
\begin{align}
\label{eq:null-drift-rademacher-sum}
    (\cL_{\mathrm{null}}m_\phi)(y)
    =\frac{1}{N}\sum_{i=1}^N\phi_i\eta_i(y),
    \qquad
    \eta_i(y)\doteq\tanh\bigl(\beta\partial_iH_{\mathrm{null}}(y)\bigr)-y_i.
\end{align}
Fix an $\cF_T$-measurable random state $y \in \Cube$. Conditional on $R,W,\cF_T$, the coefficients $\eta_i(y)$ are fixed with $\abs{\eta_i(y)}\le2$ and $\phi_i\iidsim\Unif(\{\pm1\})$. Hence, \cref{eq:null-drift-rademacher-sum} is a Rademacher sum with variance proxy $4/N$. Hence, by~\citet[Theorem~2.6]{wainwright2019high}, we obtain the following:
\begin{align}
\label{eq:null-drift-conditional-mgf}
    \E_{\Q}\left[\exp\left(\frac{N}{16}(\cL_{\mathrm{null}}m_\phi)(y)^2\right)\,\middle|\,R,W,\cF_T\right]
    \le\sqrt2.
\end{align}

Recall from \cref{subsec:cts-glauber} that $N_T\sim\Poisson(NT)$ denotes the number of arrivals up to time $T$. Since the trajectory is piecewise constant, it visits at most $N_T+1$ states during $[0,T]$, each of which is $\cF_T$-measurable. Bounding the supremum of the exponential by the sum over these states, we infer from \cref{eq:null-drift-conditional-mgf} that
\begin{align*}
    \E_{\Q}\left[\exp\left(\frac{N}{16}\sup_{t\in[0,T]}\bigl|(\cL_{\mathrm{null}}m_\phi)(Y_t)\bigr|^2\right)\,\middle|\,R,W,\cF_T\right]
    \le\sqrt2(N_T+1).
\end{align*}
The proof is concluded by taking expectations and using $\E[N_T]=NT$.
\end{proof}

By Markov's inequality, \cref{lem:null-trajectory-uniform-concentration} implies sharp subgaussian-style concentration for $\cL_{\mathrm{null}}m_\phi$ along the entire Glauber trajectory under the null law. We wish to use this estimate to control the influence of the disorder along the true Glauber trajectory, which requires transferring it from the null law to the planted law. We call this approach the \emph{reverse planting trick}, as it runs in the direction opposite to canonical applications of the planting trick~\citep{achlioptas2008algorithmic}, where the null model is the true distribution of interest, and statements proved under the planted model are transferred to it via contiguity.

However, the planted and null laws are far from contiguous, as their Hamiltonians differ by the signal term $H_{\mathrm{sig}}$, whose leading order term is $N\lambda m_\phi^k$. The discrepancy between the two laws is thus governed by the strength of the signal. Since we only need to transfer the uniform concentration of $\cL_{\mathrm{null}}m_\phi$ to the planted law until the dynamics achieves critical recovery, we instead prove the following fine-grained contiguity estimate, which adapts to the signal strength along the trajectory, and applies to any stopping time.

\begin{theorem}[Adaptive Contiguity Estimate]
\label{thm:adaptive-contiguity}
Fix $e\in\Cube$ and $T>0$, and let $(\cG_t)_{0\le t\le T}$ denote the joint filtration of the environment $(\phi,R,W)$ and Glauber dynamics up to time $T$. Let $\tau$ be a $(\cG_t)$-stopping time satisfying $\tau\le T$ almost surely, and assume $\norm{\dgrad H_{\mathrm{sig}}(Y_t(e))}_\infty\le D$ for all $t<\tau$. Then, for every $E\in\cG_\tau$,
\begin{align*}
    \Prb(E)
    \le\exp\left(\frac{1}{2}N\lambda^2\sigma_N^2\right)\E_{\Q}\left[\exp\left(\beta^2\cosh^2(\beta D)\int_0^\tau\norm{\dgrad H_{\mathrm{sig}}(Y_t(e))}^2\,dt\right)\ind{E}\right]^{\nicefrac{1}{2}}.
\end{align*}
\end{theorem}

\begin{proof}
Throughout, we write $Y_t=Y_t(e)$. For $t\le T$, let $\Prb^t$ and $\Q^t$ denote the joint laws of the environment $(\phi,R,W)$ and of Glauber dynamics up to time $t$ under the planted and null laws, respectively. By \cref{def:planted-law,def:null-law}, these laws factorize as $\Prb^t=\Prb_{\mathrm{env}}\cdot\Prb_{\mathrm{dyn}}^t$ and $\Q^t=\Q_{\mathrm{env}}\cdot\Q_{\mathrm{dyn}}^t$, where the superscript $t$ denotes the restriction to Glauber dynamics up to time $t$. Hence, the likelihood ratio $Z(t)\doteq\frac{d\Prb^t}{d\Q^t}$ factorizes as follows:
\begin{align*}
    Z(t)=Z_{\mathrm{env}}Z_{\mathrm{dyn}}(t),
    \qquad
    Z_{\mathrm{env}}\doteq\frac{d\Prb_{\mathrm{env}}}{d\Q_{\mathrm{env}}},
    \qquad
    Z_{\mathrm{dyn}}(t)\doteq\frac{d\Prb_{\mathrm{dyn}}^t(\cdot\mid\phi,R,W)}{d\Q_{\mathrm{dyn}}^t(\cdot\mid\phi,R,W)}.
\end{align*}

We begin by analyzing $Z_{\mathrm{env}}$. By \cref{def:planted-law,def:null-law}, $W$ is independent of $(\phi,R)$ and has the same law under $\Prb_{\mathrm{env}}$ and $\Q_{\mathrm{env}}$. Moreover, $(\phi,R)\disteq(\sgn(U),\abs{U})$, where $U\sim\cN(\lambda\sigma_N^2\mathbf{1},\sigma_N^2I)$ under $\Prb_{\mathrm{env}}$ and $U\sim\cN(0,\sigma_N^2I)$ under $\Q_{\mathrm{env}}$. Hence, $Z_{\mathrm{env}}$ is the likelihood ratio of these Gaussian laws evaluated at $\phi\odot R$, and we obtain the following via a standard Gaussian MGF computation.
\begin{align}
\label{eq:environment-likelihood-second-moment}
    \begin{aligned}
        \E_{\Q}\bigl[Z_{\mathrm{env}}^2\bigr]
        &=\E_{U\sim\cN(0,\sigma_N^2I)}\left[\exp\left(-\frac{1}{\sigma_N^2}\left(\norm{U-\lambda\sigma_N^2\mathbf{1}}^2-\norm{U}^2\right)\right)\right]\\
        &=\exp(-N\lambda^2\sigma_N^2)\,\E_{U\sim\cN(0,\sigma_N^2I)}\left[\exp\left(2\lambda\sum_{i=1}^NU_i\right)\right]
        =\exp(N\lambda^2\sigma_N^2).
    \end{aligned}
\end{align}

We now analyze $Z_{\mathrm{dyn}}$. For $h\in\R$ and $z\in\{\pm1\}$, let $\nu_h(z)\doteq\frac12\bigl(1+z\tanh(\beta h)\bigr)$. By \cref{subsec:cts-glauber}, conditional on the environment, $\Prb_{\mathrm{dyn}}$ and $\Q_{\mathrm{dyn}}$ are the laws of single-site dynamics with spin update kernels $\nu_{\partial_iH_\phi(y)}(z)$ and $\nu_{\partial_iH_{\mathrm{null}}(y)}(z)$, respectively. For $i\in[N]$, $z\in\{\pm1\}$, $y\in\Cube$, and $t\le T$, define
\begin{align*}
    r_i(z;y)\doteq\frac{\nu_{\partial_iH_\phi(y)}(z)}{\nu_{\partial_iH_{\mathrm{null}}(y)}(z)},
    \qquad
    a_i(z,t)\doteq r_i(z;Y_{t^-}).
\end{align*}
Observe that, conditional on the environment, the $a_i(z,t)$ are positive, predictable with respect to the filtration of Glauber dynamics up to time $T$, and bounded since $\Cube$ is finite. We define the stochastic process $(L_t^a)_{0\le t\le T}$ and the tilted measure $\Q_{\mathrm{dyn}}^{T,a}$ as follows:
\begin{align}
\label{eq:dynamical-exponential-martingale}
    L_t^a
    \doteq\exp\left(\sum_{n:T_n\le t}\log a_{I_n}(Z_n,T_n)-\int_0^t\sum_{i=1}^N\sum_{z\in\{\pm1\}}\nu_{\partial_iH_{\mathrm{null}}(Y_{s^-})}(z)\bigl(a_i(z,s)-1\bigr)\,ds\right),
    \qquad
    \frac{d\Q_{\mathrm{dyn}}^{T,a}}{d\Q_{\mathrm{dyn}}^T}
    \doteq L_T^a.
\end{align}
Then, by \cref{lem:single-site-change-measure}, $L_t^a$ is a non-negative mean-one martingale under $\Q_{\mathrm{dyn}}^T$, and thus $\Q_{\mathrm{dyn}}^{T,a}$ is a well-defined law of a marked point process. Since $\sum_z\nu_{\partial_iH_{\mathrm{null}}(y)}(z)r_i(z;y)=\sum_z\nu_{\partial_iH_\phi(y)}(z)=1$, the integral term in \cref{eq:dynamical-exponential-martingale} vanishes. Moreover, by \cref{lem:single-site-change-measure}, updates with mark $(i,z)$ under $\Q_{\mathrm{dyn}}^{T,a}$ have intensity $a_i(z,t)\nu_{\partial_iH_{\mathrm{null}}(Y_{t^-})}(z)=\nu_{\partial_iH_\phi(Y_{t^-})}(z)$. Hence, $\Q_{\mathrm{dyn}}^{T,a}=\Prb_{\mathrm{dyn}}^T$, and consequently,
\begin{align}
\label{eq:dynamical-likelihood-ratio-product}
    Z_{\mathrm{dyn}}(t)
    =L_t^a
    =\prod_{n:T_n\le t}r_{I_n}(Z_n;Y_{T_n^-}),
    \qquad t\le T.
\end{align}
Since $L_t^a$ is a non-negative mean-one martingale under $\Q_{\mathrm{dyn}}$, so is $Z_{\mathrm{dyn}}(t)$.

We now control $Z_{\mathrm{dyn}}(t)^2$. To this end, define
\begin{align*}
    g_i(y)
    \doteq\sum_{z\in\{\pm1\}}\nu_{\partial_iH_{\mathrm{null}}(y)}(z)\bigl[r_i(z;y)-1\bigr]^2
    =\sum_{z\in\{\pm1\}}\nu_{\partial_iH_{\mathrm{null}}(y)}(z)\bigl[r_i(z;y)^2-1\bigr],
\end{align*}
where the equality uses $\sum_z\nu_{\partial_iH_{\mathrm{null}}(y)}(z)r_i(z;y)=1$. Using \cref{eq:dynamical-likelihood-ratio-product} and applying \cref{lem:single-site-change-measure} with the tilts $a_i(z,t)^2$, we conclude that the following process is a mean-one martingale under $\Q_{\mathrm{dyn}}$:
\begin{align}
\label{eq:dynamical-likelihood-ratio-square}
    \begin{aligned}
        L_t\doteq L_t^{a^2}
        &=\exp\left(2\sum_{n:T_n\le t}\log r_{I_n}(Z_n;Y_{T_n^-})-\int_0^t\sum_{i=1}^Ng_i(Y_{s^-})\,ds\right)\\
        &=Z_{\mathrm{dyn}}(t)^2\exp\left(-\int_0^t\sum_{i=1}^Ng_i(Y_{s^-})\,ds\right).
    \end{aligned}
\end{align}
Moreover, since $\nu_h(z)=\frac{e^{\beta zh}}{2\cosh(\beta h)}$, a direct computation shows that for any $h,\Delta\in\R$,
\begin{align*}
    \sum_{z\in\{\pm1\}}\nu_h(z)\left[\frac{\nu_{h+\Delta}(z)}{\nu_h(z)}\right]^2
    =\frac{\cosh(\beta h)\cosh\bigl(\beta(h+2\Delta)\bigr)}{\cosh^2\bigl(\beta(h+\Delta)\bigr)}
    =1+\frac{\sinh^2(\beta\Delta)}{\cosh^2\bigl(\beta(h+\Delta)\bigr)}
    \le1+\sinh^2(\beta\Delta),
\end{align*}
where the second equality uses $\cosh(a-b)\cosh(a+b)=\cosh^2a+\sinh^2b$. 

Since $H_\phi=H_{\mathrm{null}}+H_{\mathrm{sig}}$, applying this bound with $h=\partial_iH_{\mathrm{null}}(y)$ and $\Delta=\partial_iH_{\mathrm{sig}}(y)$, and using $\abs{\tanh x}\le\abs{x}$, we obtain the following whenever $\abs{\partial_iH_{\mathrm{sig}}(y)}\le D$:
\begin{align*}
    0\le g_i(y)
    \le\sinh^2\bigl(\beta\partial_iH_{\mathrm{sig}}(y)\bigr)
    =\cosh^2\bigl(\beta\partial_iH_{\mathrm{sig}}(y)\bigr)\tanh^2\bigl(\beta\partial_iH_{\mathrm{sig}}(y)\bigr)
    \le\beta^2\cosh^2(\beta D)\partial_iH_{\mathrm{sig}}(y)^2.
\end{align*}
Using the above inequality and our assumption on $\dgrad H_{\mathrm{sig}}(Y_t)$, we obtain
\begin{align}
\label{eq:dynamical-compensator-bound}
    \int_0^\tau\sum_{i=1}^Ng_i(Y_{s^-})\,ds\le\beta^2\cosh^2(\beta D)\int_0^\tau\norm{\dgrad H_{\mathrm{sig}}(Y_s)}^2\,ds.
\end{align}

Finally, we combine the two factors via an optional stopping argument. Since $Z_{\mathrm{env}}$ is a function of the environment, it is $\cG_0$-measurable. Moreover, recall that conditional on the environment, $Z_{\mathrm{dyn}}(t)$ and $L_t$ are mean-one martingales with respect to the filtration of Glauber dynamics. Hence, for any $0\le s\le t\le T$,
\begin{align*}
    \E_\Q\bigl[Z(t)\mid\cG_s\bigr]
    =Z_{\mathrm{env}}\E_\Q\bigl[Z_{\mathrm{dyn}}(t)\mid\cG_s\bigr]
    =Z_{\mathrm{env}}Z_{\mathrm{dyn}}(s)
    =Z(s),
    \qquad
    \E_\Q\bigl[L_t\mid\cG_s\bigr]
    =L_s,
\end{align*}
and thus, $Z(t)$ and $L_t$ are martingales with respect to $(\cG_t)$ under $\Q$. Since $\tau\le T$ is a $(\cG_t)$-stopping time, Doob's optional stopping theorem~\citep[Chapter~2]{ethier1986markov} implies that $\E_\Q[Z(T)\mid\cG_\tau]=Z(\tau)$ and $\E_\Q[L_\tau\mid\cG_0]=L_0=1$. Consequently, for any $E\in\cG_\tau$,
\begin{align}
\label{eq:adaptive-contiguity-optional-stopping}
    \begin{aligned}
        \E_\Q\bigl[Z(\tau)\ind{E}\bigr]
        &=\E_\Q\bigl[\E_\Q[Z(T)\mid\cG_\tau]\ind{E}\bigr]
        =\E_\Q\bigl[\E_\Q[Z(T)\ind{E}\mid\cG_\tau]\bigr]\\
        &=\E_\Q\bigl[Z(T)\ind{E}\bigr]
        =\E_\Q\left[\frac{d\Prb^T}{d\Q^T}\ind{E}\right]
        =\Prb(E).
    \end{aligned}
\end{align}
By \cref{eq:adaptive-contiguity-optional-stopping} and the Cauchy--Schwarz inequality, we obtain
\begin{align}
\label{eq:adaptive-contiguity-cauchy-schwarz}
    \begin{aligned}
        \Prb(E)
        &=\E_\Q\left[Z(\tau)\exp\left(-\frac12\int_0^\tau\sum_{i=1}^Ng_i(Y_{s^-})\,ds\right)\cdot\exp\left(\frac12\int_0^\tau\sum_{i=1}^Ng_i(Y_{s^-})\,ds\right)\ind{E}\right]\\
        &\le\E_\Q\left[Z(\tau)^2\exp\left(-\int_0^\tau\sum_{i=1}^Ng_i(Y_{s^-})\,ds\right)\right]^{\nicefrac{1}{2}}\E_\Q\left[\exp\left(\int_0^\tau\sum_{i=1}^Ng_i(Y_{s^-})\,ds\right)\ind{E}\right]^{\nicefrac{1}{2}}.
    \end{aligned}
\end{align}
Thus, it suffices to bound the two expectations on the right-hand side. For the first, since $Z(\tau)=Z_{\mathrm{env}}Z_{\mathrm{dyn}}(\tau)$ and $Z_{\mathrm{env}}$ is $\cG_0$-measurable, we infer from \cref{eq:environment-likelihood-second-moment,eq:dynamical-likelihood-ratio-square} that
\begin{align}
\label{eq:adaptive-contiguity-first-expectation}
    \begin{aligned}
        &\E_\Q\left[Z(\tau)^2\exp\left(-\int_0^\tau\sum_{i=1}^Ng_i(Y_{s^-})\,ds\right)\right]\\
        &\qquad=\E_\Q\left[Z_{\mathrm{env}}^2\,\E_\Q\left[Z_{\mathrm{dyn}}(\tau)^2\exp\left(-\int_0^\tau\sum_{i=1}^Ng_i(Y_{s^-})\,ds\right)\,\middle|\,\cG_0\right]\right]\\
        &\qquad=\E_\Q\bigl[Z_{\mathrm{env}}^2\,\E_\Q[L_\tau\mid\cG_0]\bigr]
        =\E_\Q\bigl[Z_{\mathrm{env}}^2\bigr]
        =\exp(N\lambda^2\sigma_N^2),
    \end{aligned}
\end{align}
where the third equality uses $\E_\Q[L_\tau\mid\cG_0]=1$.

For the second expectation in \cref{eq:adaptive-contiguity-cauchy-schwarz}, we infer from \cref{eq:dynamical-compensator-bound} that
\begin{align}
\label{eq:adaptive-contiguity-second-expectation}
    \E_\Q\left[\exp\left(\int_0^\tau\sum_{i=1}^Ng_i(Y_{s^-})\,ds\right)\ind{E}\right]
    \le\E_\Q\left[\exp\left(\beta^2\cosh^2(\beta D)\int_0^\tau\norm{\dgrad H_{\mathrm{sig}}(Y_s)}^2\,ds\right)\ind{E}\right].
\end{align}
Substituting \cref{eq:adaptive-contiguity-first-expectation,eq:adaptive-contiguity-second-expectation} into \cref{eq:adaptive-contiguity-cauchy-schwarz}, we conclude that
\begin{align*}
    \Prb(E)
    \le\exp\left(\frac12N\lambda^2\sigma_N^2\right)\E_\Q\left[\exp\left(\beta^2\cosh^2(\beta D)\int_0^\tau\norm{\dgrad H_{\mathrm{sig}}(Y_t)}^2\,dt\right)\ind{E}\right]^{\nicefrac{1}{2}},
\end{align*}
which concludes the proof.
\end{proof}

\begin{remark}
\label{rem:adaptive-contiguity-transfer}
If $\int_0^\tau\norm{\dgrad H_{\mathrm{sig}}(Y_t(e))}^2\,dt\le\Gamma$ holds on $E$ for some deterministic $\Gamma\ge0$, then \cref{thm:adaptive-contiguity} implies that $\Prb(E)\le\exp\bigl(\frac12N\lambda^2\sigma_N^2+\frac12\beta^2\cosh^2(\beta D)\Gamma\bigr)\Q(E)^{1/2}$. Thus, events that are (sufficiently) exponentially unlikely under the null law remain so under the planted law, albeit at an exponential cost due to the change in environment and the accumulated signal contribution along the trajectory.
\end{remark}

\subsection{Critical Recovery via the Reverse Planting Trick}
\label{subsec:critical-recovery-reverse-planting}

In this subsection, we shall use the reverse planting trick to establish \cref{thm:critical-recovery-warm-start}. Throughout, we retain the assumptions of \cref{thm:critical-recovery-warm-start}, and we take $N$ to be sufficiently large. Our proof will make use of the following regularity estimates, which we prove in \cref{app:regularity-critical-recovery}.

\begin{lemma}[Correlation of the Warm Start]
\label{lem:warm-start-correlation}
There exist constants $c_0,c>0$, depending only on $k$, such that the following holds under the planted law:
\begin{align*}
    \Prb\left[m_\phi(e)\ge c_0\lambda\sigma_N\ \ \forall\,e\in\cB_H\bigl(\mathbf{1},\lambda^2\sigma_N^2N^{-\vartheta}\bigr)\right]
    \ge1-\exp\bigl(-cN\lambda^2\sigma_N^2\bigr).
\end{align*}
\end{lemma}

Henceforth, we write $a_0\doteq c_0\lambda\sigma_N$, where $c_0$ is the constant in \cref{lem:warm-start-correlation}.

\begin{lemma}[Regularity of the Null Hamiltonian]
\label{lem:null-hamiltonian-regularity}
There exist constants $C_0,c>0$, depending only on $k$, such that the following holds under the planted law:
\begin{align*}
    \Prb\left[\max_{y\in\Cube}\norm{\dgrad H_{\mathrm{null}}(y)}^2\ge C_0N\right]
    \le\exp(-cN).
\end{align*}
\end{lemma}

\begin{lemma}[Bounds on the Signal Drift]
\label{lem:signal-drift-bounds}
There exist constants $c_1,C_1>0$, depending only on $k$, such that the following holds for every $y\in\Cube$ with $m_\phi(y)\ge a_0/2$ and every $i\in[N]$:
\begin{align}
\label{eq:signal-gradient-bounds}
    c_1\lambda m_\phi(y)^{k-1}
    \le\phi_i\partial_iH_{\mathrm{sig}}(y)
    \le C_1\lambda m_\phi(y)^{k-1}.
\end{align}
Moreover, for any $C>0$, there exists a constant $c_2>0$, depending only on $\beta$, $k$ and $C$, such that the following holds for every $y\in\Cube$ with $a_0/2\le m_\phi(y)\le Cr_N$ and $\norm{\dgrad H_{\mathrm{null}}(y)}^2\le C_0N$, where $C_0$ is the constant in \cref{lem:null-hamiltonian-regularity}:
\begin{align*}
    (\cL_{\mathrm{pl}}-\cL_{\mathrm{null}})m_\phi(y)
    \ge c_2\lambda m_\phi(y)^{k-1}.
\end{align*}
\end{lemma}

\begin{lemma}[{\cite[Lemma~5.1]{arous2020algorithmic}}]
\label{lem:power-law-gronwall}
Let $a,c,T>0$, and let $f:[0,T]\to\R$ be a bounded measurable function such that
\begin{align*}
    f(t)\ge a+c\int_0^tf(s)^{k-1}\,ds
    \qquad\forall\,t\in[0,T].
\end{align*}
Then, the following holds for every $t\in[0,T]$ with $(k-2)ca^{k-2}t<1$:
\begin{align*}
    f(t)\ge\bigl[a^{-(k-2)}-(k-2)ct\bigr]^{-\frac{1}{k-2}}.
\end{align*}
\end{lemma}

Equipped with the above regularity estimates, we now establish \cref{thm:critical-recovery-warm-start} by applying the adaptive contiguity bound up to a carefully chosen stopping time, thereby proving that the null drift is unlikely to overtake the signal drift before the dynamics achieves critical recovery.

\subsubsection{Proof of \texorpdfstring{\cref{thm:critical-recovery-warm-start}}{Theorem \ref{thm:critical-recovery-warm-start}}}
\label{subsubsec:proof-critical-recovery-warm-start}

\begin{proof}
Throughout this proof, we work under the planted law. We fix a deterministic initialization $e\in\cB_H(\mathbf{1},\lambda^2\sigma_N^2N^{-\vartheta})$, and write $Y_t=Y_t(e)$ and $m_t=m_\phi(Y_t)$. We also assume that $C_{\beta,k}$ is sufficiently large, depending on $\beta$, $k$ and $C$, where $C$ is the constant in the definition of $T_*(e)$.

We begin by fixing some constants. Let $C_1$ and $c_2$ denote the constants in \cref{lem:signal-drift-bounds}, and define
\begin{align*}
    \eta\doteq\frac{c_2}{2},
    \qquad
    D\doteq C_1C^{k-1},
    \qquad
    T_\lambda\doteq\frac{C_T}{\lambda a_0^{k-2}},
    \qquad
    C_T\doteq\frac{2^{k-1}}{(k-2)\eta}.
\end{align*}
Since $C_{\beta,k}$ is sufficiently large and
\begin{align*}
    \lambda a_0^{k-2}
    =c_0^{k-2}\lambda^{k-1}\sigma_N^{k-2}
    \asymp_kC_{\beta,k}^{k-1}N^{(k-1)(\gamma-\frac{k-2}{4})},
\end{align*}
we may assume that $\lambda a_0^{k-2}\ge1$, and hence $T_\lambda\le C_T$.

We now define the Dynkin martingale of $m_\phi$. By \cref{def:planted-law}, conditional on the environment, $(Y_t)_{t\ge0}$ is Glauber dynamics for $H_\phi$ at inverse temperature $\beta$. Hence, by \cref{lem:single-site-dynkin}, the following process is a martingale conditional on the environment:
\begin{align}
\label{eq:critical-recovery-dynkin}
    M_t\doteq m_t-m_0-\int_0^t(\cL_{\mathrm{pl}}m_\phi)(Y_s)\,ds.
\end{align}
Since $m_\phi$ is linear with $\partial_im_\phi=\phi_i/N$, we have $\norm{m_\phi}_{\cD^1}\le2$ and $\max_{y\in\Cube}\norm{\dgrad m_\phi(y)}_\infty=1/N$. Consider the events
\begin{align*}
    \cI\doteq\{m_0\ge a_0\},
    \qquad
    \cE_{\mathrm{reg}}\doteq\left\{\max_{y\in\Cube}\norm{\dgrad H_{\mathrm{null}}(y)}^2<C_0N\right\},
    \qquad
    \calM_{T_\lambda}\doteq\left\{\sup_{0\le t\le T_\lambda}\abs{M_t}\le\frac{a_0}{32}\right\}.
\end{align*}
Applying \cref{lem:warm-start-correlation,lem:null-hamiltonian-regularity}, and \cref{lem:dynkin-martingale-concentration} conditionally on the environment, we obtain the following:
\begin{align}
\label{eq:critical-recovery-event-probabilities}
    \begin{aligned}
        &\Prb(\cI^c)\le\exp\bigl(-cN\lambda^2\sigma_N^2\bigr),
        \qquad
        \Prb(\cE_{\mathrm{reg}}^c)\le\exp(-cN),\\
        &\Prb(\calM_{T_\lambda}^c)
        \le2\exp\left(-\frac{cN(a_0/32)^2}{4T_\lambda+a_0/32}\right)
        \le2\exp\bigl(-cN\lambda^2\sigma_N^2\bigr),
    \end{aligned}
\end{align}
where the last step uses $T_\lambda\le C_T$, $a_0\le1$ and $a_0=c_0\lambda\sigma_N$.

Define the following stopping times:
\begin{align*}
    T_*^Y(e)
    &\doteq\inf\bigl\{t\ge0:m_t\ge Cr_N\bigr\},\\
    \tau^{\downarrow}
    &\doteq\inf\Bigl\{t\ge0:m_t\le\frac{a_0}{2}\Bigr\},\\
    \tau^{\mathrm{bad}}
    &\doteq\inf\bigl\{t\ge0:\abs{(\cL_{\mathrm{null}}m_\phi)(Y_t)}\ge\eta\lambda m_t^{k-1}\bigr\},\\
    \tau
    &\doteq T_\lambda\wedge T_*^Y(e)\wedge\tau^{\downarrow}\wedge\tau^{\mathrm{bad}}.
\end{align*}

We next record two consequences of these definitions. For every $t<\tau$, we have $a_0/2<m_t<Cr_N$. Hence, since $\lambda r_N^{k-1}=1$, \cref{eq:signal-gradient-bounds} implies that
\begin{align}
\label{eq:critical-recovery-signal-gradient}
    \norm{\dgrad H_{\mathrm{sig}}(Y_t)}_\infty
    \le C_1\lambda(Cr_N)^{k-1}
    =D,
    \qquad
    \norm{\dgrad H_{\mathrm{sig}}(Y_t)}^2
    \le C_1^2N\lambda^2m_t^{2k-2}
    \qquad\forall\,t<\tau.
\end{align}
Moreover, $\cL_{\mathrm{pl}}m_\phi=\cL_{\mathrm{null}}m_\phi+(\cL_{\mathrm{pl}}-\cL_{\mathrm{null}})m_\phi$. Hence, on $\cE_{\mathrm{reg}}$, the second part of \cref{lem:signal-drift-bounds} and $t<\tau^{\mathrm{bad}}$ imply that
\begin{align}
\label{eq:critical-recovery-planted-drift}
    (\cL_{\mathrm{pl}}m_\phi)(Y_t)
    \ge2\eta\lambda m_t^{k-1}-\abs{(\cL_{\mathrm{null}}m_\phi)(Y_t)}
    \ge\eta\lambda m_t^{k-1}
    >0
    \qquad\forall\,t<\tau.
\end{align}

We now control the correlation up to time $\tau$. Define $\calM_\tau\doteq\{\sup_{0\le t\le\tau}\abs{M_t}\le a_0/32\}$. By \cref{eq:critical-recovery-dynkin,eq:critical-recovery-planted-drift}, the following holds on $\cI\cap\cE_{\mathrm{reg}}\cap\calM_\tau$ for any $0\le s\le t\le\tau$:
\begin{align*}
    m_t-m_s
    =\int_s^t(\cL_{\mathrm{pl}}m_\phi)(Y_u)\,du+M_t-M_s
    \ge M_t-M_s
    \ge-\frac{a_0}{16}.
\end{align*}
Since $M_0=0$ and $m_0\ge a_0$, we also have
\begin{align}
\label{eq:critical-recovery-overlap-bounds}
    m_t\ge m_0+M_t\ge\frac{31a_0}{32}
    \quad\forall\,t\le\tau,
    \qquad
    m_s\le m_t+\frac{a_0}{16}\le2m_t
    \quad\forall\,s\le t\le\tau.
\end{align}
Similarly, we infer from \cref{eq:critical-recovery-dynkin,eq:critical-recovery-planted-drift} that
\begin{align}
\label{eq:critical-recovery-integral-inequality}
    m_t
    =m_0+\int_0^t(\cL_{\mathrm{pl}}m_\phi)(Y_s)\,ds+M_t
    \ge\frac{31a_0}{32}+\eta\lambda\int_0^tm_s^{k-1}\,ds
    \qquad\forall\,t\le\tau.
\end{align}

We now show that the null drift is unlikely to dominate the signal drift before the dynamics achieves critical recovery. To this end, define the failure event
\begin{align*}
    \cE_{\mathrm{fail}}\doteq\cI\cap\cE_{\mathrm{reg}}\cap\calM_\tau\cap\bigl\{\tau=\tau^{\mathrm{bad}}<T_*^Y(e)\bigr\}.
\end{align*}
Since $\cI$ and $\cE_{\mathrm{reg}}$ depend only on the environment, and the remaining events are determined by the environment and the trajectory up to time $\tau$, we have $\cE_{\mathrm{fail}}\in\cG_\tau$. To bound $\Prb(\cE_{\mathrm{fail}})$ by \cref{thm:adaptive-contiguity}, we first control the accumulated signal drift on $\cE_{\mathrm{fail}}$. Note that \cref{eq:critical-recovery-signal-gradient,eq:critical-recovery-planted-drift,eq:critical-recovery-overlap-bounds} hold pathwise on $\cE_{\mathrm{fail}}$, irrespective of the underlying law. By \cref{eq:critical-recovery-signal-gradient,eq:critical-recovery-overlap-bounds}, the following holds on $\cE_{\mathrm{fail}}$:
\begin{align}
\label{eq:critical-recovery-energy-moment}
    \int_0^\tau\norm{\dgrad H_{\mathrm{sig}}(Y_t)}^2\,dt
    \le C_1^2N\lambda^2\int_0^\tau m_t^{2k-2}\,dt
    \le C_1^2N\lambda^2(2m_\tau)^{k-1}\int_0^\tau m_t^{k-1}\,dt.
\end{align}
Moreover, since $m_0+M_\tau\ge31a_0/32$ by \cref{eq:critical-recovery-overlap-bounds}, we infer the following from \cref{eq:critical-recovery-dynkin,eq:critical-recovery-planted-drift}
\begin{align}
\label{eq:critical-recovery-drift-integral}
    \eta\lambda\int_0^\tau m_t^{k-1}\,dt
    \le\int_0^\tau(\cL_{\mathrm{pl}}m_\phi)(Y_t)\,dt
    =m_\tau-m_0-M_\tau
    \le m_\tau-\frac{a_0}{2}.
\end{align}
Since $m_\tau\ge a_0/2$, we also have
\begin{align}
\label{eq:critical-recovery-power-gap}
    m_\tau^{2k-2}-\Bigl(\frac{a_0}{2}\Bigr)^{2k-2}
    \ge m_\tau^{2k-3}\Bigl(m_\tau-\frac{a_0}{2}\Bigr)
    \ge\Bigl(\frac{a_0}{2}\Bigr)^{k-2}m_\tau^{k-1}\Bigl(m_\tau-\frac{a_0}{2}\Bigr).
\end{align}
Combining \cref{eq:critical-recovery-energy-moment,eq:critical-recovery-drift-integral,eq:critical-recovery-power-gap}, we infer that the following holds on $\cE_{\mathrm{fail}}$:
\begin{align}
\label{eq:critical-recovery-energy-combined}
    \int_0^\tau\norm{\dgrad H_{\mathrm{sig}}(Y_t)}^2\,dt
    &\le\frac{2^{k-1}C_1^2N\lambda}{\eta}m_\tau^{k-1}\Bigl(m_\tau-\frac{a_0}{2}\Bigr)\nonumber\\
    &\le\frac{2^{k-1}C_1^2N}{\eta^3\lambda(a_0/2)^{k-2}}\Bigl[\eta^2\lambda^2m_\tau^{2k-2}-\eta^2\lambda^2\Bigl(\frac{a_0}{2}\Bigr)^{2k-2}\Bigr].
\end{align}
Moreover, since $\tau=\tau^{\mathrm{bad}}$ on $\cE_{\mathrm{fail}}$, we have
\begin{align}
\label{eq:critical-recovery-bad-time}
    \eta^2\lambda^2m_\tau^{2k-2}
    \le\bigl|(\cL_{\mathrm{null}}m_\phi)(Y_\tau)\bigr|^2
    \le\sup_{t\in[0,T_\lambda]}\bigl|(\cL_{\mathrm{null}}m_\phi)(Y_t)\bigr|^2.
\end{align}
Since $C_{\beta,k}$ is sufficiently large, we may assume that $\eta^3\lambda(a_0/2)^{k-2}\ge2^{k+3}\beta^2\cosh^2(\beta D)C_1^2$. Hence, writing $c_3\doteq2^{k-1}\beta^2\cosh^2(\beta D)C_1^2/\eta$ and combining \cref{eq:critical-recovery-energy-combined,eq:critical-recovery-bad-time}, we conclude that the following holds on $\cE_{\mathrm{fail}}$:
\begin{align}
\label{eq:critical-recovery-energy-bound}
    \beta^2\cosh^2(\beta D)\int_0^\tau\norm{\dgrad H_{\mathrm{sig}}(Y_t)}^2\,dt
    \le\frac{N}{16}\sup_{t\in[0,T_\lambda]}\bigl|(\cL_{\mathrm{null}}m_\phi)(Y_t)\bigr|^2-c_3N\lambda\Bigl(\frac{a_0}{2}\Bigr)^k.
\end{align}

We now transfer the failure event to the null law. Since $\cE_{\mathrm{fail}}\in\cG_\tau$ and $\tau\le T_\lambda$ is a bounded stopping time, applying \cref{thm:adaptive-contiguity} using \cref{eq:critical-recovery-signal-gradient} to $E=\cE_{\mathrm{fail}}$, we infer from \cref{eq:critical-recovery-energy-bound} that
\begin{align*}
    \Prb(\cE_{\mathrm{fail}})
    &\le\exp\Bigl(\frac12N\lambda^2\sigma_N^2\Bigr)\E_\Q\left[\exp\left(\beta^2\cosh^2(\beta D)\int_0^\tau\norm{\dgrad H_{\mathrm{sig}}(Y_t)}^2\,dt\right)\ind{\cE_{\mathrm{fail}}}\right]^{\nicefrac{1}{2}}\\
    &\le\exp\Bigl(\frac12N\lambda^2\sigma_N^2-\frac{c_3}{2}N\lambda\Bigl(\frac{a_0}{2}\Bigr)^k\Bigr)\E_\Q\left[\exp\left(\frac{N}{16}\sup_{t\in[0,T_\lambda]}\bigl|(\cL_{\mathrm{null}}m_\phi)(Y_t)\bigr|^2\right)\right]^{\nicefrac{1}{2}}\\
    &\le2^{\nicefrac{1}{4}}(1+NT_\lambda)^{\nicefrac{1}{2}}\exp\Bigl(\frac12N\lambda^2\sigma_N^2-\frac{c_3}{2}N\lambda\Bigl(\frac{a_0}{2}\Bigr)^k\Bigr),
\end{align*}
where the last step uses \cref{lem:null-trajectory-uniform-concentration}. Moreover, since $C_{\beta,k}$ is sufficiently large, we may assume that $c_3c_0^2\lambda(a_0/2)^{k-2}\ge16$, and hence
\begin{align*}
    \frac{c_3}{2}N\lambda\Bigl(\frac{a_0}{2}\Bigr)^k
    =\frac{c_3c_0^2}{8}\lambda\Bigl(\frac{a_0}{2}\Bigr)^{k-2}N\lambda^2\sigma_N^2
    \ge2N\lambda^2\sigma_N^2.
\end{align*}
Since $T_\lambda\le C_T$ and $N\lambda^2\sigma_N^2\gtrsim\sqrt N$, we conclude that
\begin{align}
\label{eq:critical-recovery-failure-probability}
    \Prb(\cE_{\mathrm{fail}})
    \le2^{\nicefrac{1}{4}}(1+NT_\lambda)^{\nicefrac{1}{2}}\exp\Bigl(-\frac32N\lambda^2\sigma_N^2\Bigr)
    \le\exp\bigl(-N\lambda^2\sigma_N^2\bigr).
\end{align}

We are now ready to show that the dynamics achieves critical recovery before time $T_\lambda$. Define the good event
\begin{align*}
    \cE_{\mathrm{good}}\doteq\cI\cap\cE_{\mathrm{reg}}\cap\calM_{T_\lambda}\cap\bigl(\{\tau<\tau^{\mathrm{bad}}\}\cup\{\tau=T_*^Y(e)\}\bigr).
\end{align*}
Observe that
\begin{align*}
    \bigl\{\tau=\tau^{\mathrm{bad}}<T_*^Y(e)\bigr\}^c
    =\{\tau<\tau^{\mathrm{bad}}\}\cup\{\tau=T_*^Y(e)\}.
\end{align*}
Moreover, since $\tau\le T_\lambda$, we have $\calM_{T_\lambda}\subseteq\calM_\tau$. Hence, $\cE_{\mathrm{good}}=\cI\cap\cE_{\mathrm{reg}}\cap\calM_{T_\lambda}\cap\cE_{\mathrm{fail}}^c$, and we infer from \cref{eq:critical-recovery-event-probabilities,eq:critical-recovery-failure-probability} and a union bound that
\begin{align*}
    \Prb(\cE_{\mathrm{good}})\ge1-5\exp\bigl(-cN\lambda^2\sigma_N^2\bigr).
\end{align*}
Henceforth, we work on $\cE_{\mathrm{good}}$. Since $\calM_{T_\lambda}\subseteq\calM_\tau$, \cref{eq:critical-recovery-overlap-bounds,eq:critical-recovery-integral-inequality} hold. We first turn \cref{eq:critical-recovery-integral-inequality} into an explicit lower bound on the correlation. To this end, define the blow-up time
\begin{align*}
    t_\infty
    \doteq\frac{(a_0/2)^{-(k-2)}}{(k-2)\eta\lambda}
    =\frac{2^{k-2}}{(k-2)\eta\lambda a_0^{k-2}}.
\end{align*}
By our choice of $C_T$, we have $t_\infty=T_\lambda/2<T_\lambda$. Since $31a_0/32\ge a_0/2$, applying \cref{lem:power-law-gronwall} to $f(t)=m_t$ on $[0,\tau]$ with $a=a_0/2$ and $c=\eta\lambda$, we infer from \cref{eq:critical-recovery-integral-inequality} that
\begin{align}
\label{eq:critical-recovery-power-law}
    m_t
    \ge\Bigl[\Bigl(\frac{a_0}{2}\Bigr)^{-(k-2)}-(k-2)\eta\lambda t\Bigr]^{-\frac{1}{k-2}}
    \qquad\forall\,t\in[0,\tau\wedge t_\infty).
\end{align}

We now examine which of the stopping times attains the minimum in the definition of $\tau$. By \cref{eq:critical-recovery-overlap-bounds}, we have $m_\tau\ge31a_0/32>a_0/2$, and hence $\tau\ne\tau^{\downarrow}$. Therefore, one of the following two cases holds.

\textbf{Case 1: $\tau=T_*^Y(e)$.} In this case, we have $T_*^Y(e)=\tau\le T_\lambda$, and we are done.

\textbf{Case 2: $\tau<T_*^Y(e)$.} By definition of $\cE_{\mathrm{good}}$, we must have $\tau<\tau^{\mathrm{bad}}$, and hence $\tau=T_\lambda$. Moreover, since $T_*^Y(e)>0$, we have $a_0\le m_0<Cr_N$. Define
\begin{align*}
    t_*\doteq\frac{(a_0/2)^{-(k-2)}-(2Cr_N)^{-(k-2)}}{(k-2)\eta\lambda}.
\end{align*}
Since $a_0/2<2Cr_N$, we have $0<t_*<t_\infty<T_\lambda=\tau$. Hence, \cref{eq:critical-recovery-power-law} holds at $t=t_*$, which implies that
\begin{align*}
    m_{t_*}
    \ge\Bigl[\Bigl(\frac{a_0}{2}\Bigr)^{-(k-2)}-(k-2)\eta\lambda t_*\Bigr]^{-\frac{1}{k-2}}
    =2Cr_N.
\end{align*}
Consequently, $T_*^Y(e)\le t_*<\tau$, which contradicts $\tau<T_*^Y(e)$. Hence, this case cannot occur.

Combining the two cases, we have $T_*^Y(e)\le T_\lambda$ on $\cE_{\mathrm{good}}$, and after adjusting constants, we conclude that
\begin{align}
\label{eq:critical-recovery-fixed-initialization}
    \Prb\bigl[T_*^Y(e)>T_\lambda\bigr]\le\exp\bigl(-cN\lambda^2\sigma_N^2\bigr).
\end{align}

Finally, we invert the gauge transformation and take a union bound over the initialization. Since $\sgn(u_H)=\xopt\odot\phi$ under the planted law, and the gauge transformation $x\mapsto\xopt\odot\phi\odot x$ is an isometric involution of $\Cube$, it maps $\cB_H(\sgn(u_H),r)$ isometrically onto $\cB_H(\mathbf{1},r)$ for every $r>0$. Moreover, $m(X_t(\xopt\odot\phi\odot e))=m_\phi(Y_t(e))$, and hence $T_*(\xopt\odot\phi\odot e)=T_*^Y(e)$ for every $e\in\Cube$. Consequently, we infer from \cref{eq:critical-recovery-fixed-initialization} that the following holds for every $e\in\cB_H(\mathbf{1},\lambda^2\sigma_N^2N^{-\vartheta})$:
\begin{align*}
    \Prb\bigl[T_*(\xopt\odot\phi\odot e)>T_\lambda\bigr]\le\exp\bigl(-cN\lambda^2\sigma_N^2\bigr).
\end{align*}
Since $\abs{\cB_H(\mathbf{1},r)}\le(3/r)^{rN}$ for $r\in(0,\nicefrac{1}{2}]$ and $\lambda^2\sigma_N^2\gtrsim N^{-\nicefrac{1}{2}}$, we have $\log\abs{\cB_H(\mathbf{1},\lambda^2\sigma_N^2N^{-\vartheta})}=o(N\lambda^2\sigma_N^2)$. Hence, a union bound over $e\in\cB_H(\mathbf{1},\lambda^2\sigma_N^2N^{-\vartheta})$ yields
\begin{align*}
    \Prb\Bigl[T_*(e)\le T_\lambda\ \ \forall\,e\in\cB_H\bigl(\sgn(u_H),\lambda^2\sigma_N^2N^{-\vartheta}\bigr)\Bigr]
    \ge1-\exp\bigl(-cN\lambda^2\sigma_N^2\bigr).
\end{align*}
The proof is concluded by noting that $T_\lambda=C_T/(\lambda a_0^{k-2})\lesssim N^{-(k-1)(\gamma-\frac{k-2}{4})}$, with implicit constants depending only on $\beta$, $k$ and $C$.
\end{proof}

\subsection{Proof of \texorpdfstring{\cref{thm:persistent-strong-recovery-algorithmic-threshold}}{Theorem~\ref{thm:persistent-strong-recovery-algorithmic-threshold}}}
\label{subsec:proof-persistent-strong-recovery-algorithmic-threshold}

\begin{proof}
Throughout the proof, we assume that $N$ is sufficiently large. Since $\beta_0^{-2}\lambda^{-2}N^{k+2\vartheta}\asymp_kN^{2k+4\vartheta}$, our choices of $\beta_0$ and $T_{\mathrm{av}}$ satisfy the assumptions of \cref{thm:tpca-averaged-differential-recovery}. By construction, $(X_t)_{t\ge0}$ is Glauber dynamics for $H$ at inverse temperature $\beta$, started from $X_0=\sgn(\Xbar_{T_{\mathrm{av}}})$. If $\lambda\ge C_{\beta,k}N^{\tfrac{k-2}{4}+\tfrac{1}{4k}}$, the claim follows from \cref{cor:averaged-glauber-persistent-strong-recovery}.

Henceforth, we assume that $\lambda=C_{\beta,k}N^\gamma$ with $\tfrac{k-2}{4}\le\gamma<\tfrac{k-2}{4}+\tfrac{1}{4k}$, and let $r_N\doteq\lambda^{-\tfrac{1}{k-1}}$, as in \cref{subsec:critical-recovery-warm-start}. We first show that the second phase attains critical correlation in constant time. We infer from \cref{thm:tpca-averaged-differential-recovery} that the following holds with probability at least $1-\exp(-c_kN^{(1-\vartheta)/2})$:
\begin{align*}
    X_0=\sgn(\Xbar_{T_{\mathrm{av}}})
    \in\cB_H\bigl(\sgn(u_H),\lambda^2\sigma_N^2N^{-\vartheta/2}\bigr).
\end{align*}
Since \cref{thm:critical-recovery-warm-start} holds simultaneously for all initializations in this Hamming ball, we may apply it with $\vartheta/2$ in place of $\vartheta$, and we obtain the following with probability at least $1-\exp(-c_kN^{(1-\vartheta)/2})-\exp(-cN\lambda^2\sigma_N^2)$ for some sufficiently large constant $C$:
\begin{align}
\label{eq:algorithmic-threshold-critical-time}
    T_*(X_0)=\inf\bigl\{t\ge0:m(X_t)\ge Cr_N\bigr\}\le C_1N^{-(k-1)\left(\gamma-\tfrac{k-2}{4}\right)}.
\end{align}
Let $\bar r_N\doteq Cr_N$. Since $\lambda r_N^{k-1}=1$, we have $\lambda\bar r_N^{k-1}=C^{k-1}\ge C_{\beta,k}$. Hence, applying \cref{thm:strong-recovery-after-critical-correlation} on the event in \cref{eq:algorithmic-threshold-critical-time} with $\tau_{\bar r_N}=T_*(X_0)$, and using $T_{\mathrm{str}}\ge T_*(X_0)+1+2\eta\log N$ for sufficiently large $N$, we obtain the following via a union bound:
\begin{align*}
    \Prb\left[\inf_{t\in[T_{\mathrm{str}},T]}m(X_t)<1-N^{-\eta}\right]
    \le\exp\bigl(-c_kN^{(1-\vartheta)/2}\bigr)+\exp\bigl(-cN\lambda^2\sigma_N^2\bigr)+3\exp\bigl(-cN\min\{\bar r_N^2,N^{-2\eta}\}\bigr).
\end{align*}
Finally, since $\gamma<\tfrac{k-2}{4}+\tfrac{1}{4k}$, we have $N\bar r_N^2\gtrsim N^{\tfrac{k+1}{2k}}$. As $N\lambda^2\sigma_N^2\gtrsim\sqrt N$ and $\tfrac{k+1}{2k}>\nicefrac{1}{2}>\tfrac{1-\vartheta}{2}$, the claim follows after adjusting constants appropriately.
\end{proof}

%% file: strong_recovery_every_order.tex
\section{Strong Recovery at the Algorithmic Threshold for Every Tensor Order}
\label{sec:strong-recovery-every-order}

In this section, we analyze \cref{alg:annealed-glauber-ergodic-second-moments} and show that it achieves persistent strong recovery above the algorithmic threshold for any $k\ge3$, irrespective of its parity. Specifically, we prove the following theorem, which is a formal version of \cref{thm:glauber-every-k-informal}. The complete proof is presented in \cref{subsec:proof-persistent-strong-recovery-every-order}.

\begin{theorem}[Persistent Strong Recovery above the Algorithmic Threshold for Every Tensor Order]
\label{thm:persistent-strong-recovery-every-order}
Let $k\ge3$, let $\beta>0$, and assume $\lambda\ge C_{\beta,k}N^{\tfrac{k-2}{4}}$, where $C_{\beta,k}$ is a sufficiently large constant. Fix $\vartheta\in(0,1)$ and $\eta\in(0,\nicefrac{1}{2})$, and let $(X_t)_{t\ge0}$ denote the trajectory of the second phase of \cref{alg:annealed-glauber-ergodic-second-moments}, initialized at $\tilX_0\sim\pi$ with inverse temperatures $\beta_0$ and $\beta$ for the two respective phases, where $\beta_0\asymp_k\lambda^{-1}N^{-\tfrac{k+1}{2}-\vartheta}$ and $T_{\mathrm{av}}\asymp_kN^{2k+2+4\vartheta}$. Set $T_{\mathrm{str}}\doteq\log N$. Then, for all sufficiently large $N$ and every deterministic $T$ satisfying $T_{\mathrm{str}}\le T\le\poly(N)$,
\begin{align*}
    \Prb\left[\inf_{t\in[T_{\mathrm{str}},T]}\abs{m(X_t)}\ge1-N^{-\eta}\right]
    \ge1-\exp\left(-c_{\beta,k}N^{\min\{(1-\vartheta)/2,\,1-2\eta\}}\right).
\end{align*}
Moreover, if $k$ is odd, the same bound holds with $m(X_t)$ in place of $\abs{m(X_t)}$.
\end{theorem}

\subsection{Properties of the Augmented Ergodic Second Moment}
\label{subsec:recovery-ergodic-second-moments}

In this section, we analyze the properties of the augmented ergodic second moment of the first phase of \cref{alg:annealed-glauber-ergodic-second-moments}, defined by $\bar M_T\doteq\frac1T\int_0^TF(\tilX_t)^{\otimes2}\,dt$, where $F(x)\doteq(x,1)\in\R^{N+1}$ and $(\tilX_t)_{t\ge0}$ denotes the trajectory of this phase. Recall from \cref{sec:averaged-discrete-differentials} that the average discrete gradient and the average discrete Hessian of $H$ are $u_H=\E_\pi[\dgrad H]$ and $A_H=\E_\pi[\dgrad^2H]$, respectively. Expanding $F(\tilX_t)^{\otimes2}$ into blocks, we obtain
\begin{align*}
    \bar M_T
    =\begin{pmatrix}
        \Sigmabar_T & \Xbar_T\\
        \Xbar_T^\top & 1
    \end{pmatrix},
\end{align*}
where $\Xbar_T\doteq\frac1T\int_0^T\tilX_t\,dt$ and $\Sigmabar_T\doteq\frac1T\int_0^T\tilX_t\tilX_t^\top\,dt$ denote the ergodic average and the ergodic second moment of the first phase, respectively. Consequently, by \cref{thm:glauber-recovers-average-gradient,thm:glauber-recovers-average-hessian}, at high temperature we expect $\bar M_{T_{\mathrm{av}}}$ to concentrate around $I_{N+1}+\beta_0\cA_H$, where
\begin{align*}
    \cA_H\doteq\begin{pmatrix}
        A_H & u_H\\
        u_H^\top & 0
    \end{pmatrix}.
\end{align*}

As discussed in \cref{subsec:sketch-any-order}, the parity of $k$ determines which block of $\cA_H$ carries the signal. To observe this, let us first consider the case when $k$ is odd. In this case, \cref{avg-discrete-grad-snr,avg-discrete-hess-snr} imply that $A_H=0$ and $u_H$ is a noisy observation of $\xopt$. Hence, the top eigenvector of $\cA_H$ is $\tfrac{1}{\sqrt2}(u_H/\norm{u_H},1)$, and rounding its first $N$ coordinates to signs yields $\sgn(u_H)$. Consequently, we expect the initial state of the second phase of \cref{alg:annealed-glauber-ergodic-second-moments} to lie close to $\sgn(u_H)$, which parallels the setting in \cref{alg:averaged-annealed-glauber}.

When $k$ is even, \cref{avg-discrete-grad-snr} implies that $u_H=0$, so $\cA_H$ is block diagonal, and its top eigenvector is obtained by appending a zero to a top eigenvector of $A_H$. By \cref{avg-discrete-hess-snr}, $A_H$ has a rank-one spike in the planted direction $\plant$, and hence, by \cref{rem:average-hessian-snr}, its top eigenvector achieves constant absolute correlation with $\plant$. Consequently, we expect the initial state of the second phase to weakly recover the planted signal up to a global sign. 

We note that \cref{alg:annealed-glauber-ergodic-second-moments} does not need to know the parity of $k$, since the same top eigenvector computation extracts the signal from whichever block of $\cA_H$ carries it. We formalize these heuristics in the following theorem, which can be interpreted as an analogue of \cref{thm:tpca-averaged-differential-recovery} for \cref{alg:annealed-glauber-ergodic-second-moments}.

\begin{theorem}[Recovery by the Augmented Ergodic Second Moment for Tensor PCA]
\label{thm:recovery-ergodic-second-moments}
Assume that $\lambda\ge C_kN^{\tfrac{k-2}{4}}$, and let $(\tilX_t)_{t\ge0}$ denote Glauber dynamics for the tensor PCA Hamiltonian $H$ at inverse temperature $\beta_0>0$, initialized at $\tilX_0\sim\pi$ independently of the Gaussian tensor. Define
\begin{align*}
    \bar M_T\doteq\frac{1}{T}\int_0^TF(\tilX_t)^{\otimes2}\,dt,
    \qquad
    F(x)\doteq(x,1)\in\R^{N+1},
\end{align*}
write $u_H=\E_\pi[\dgrad H]$, and let $(v,a)\in\R^N\times\R$ be a top unit eigenvector of $\bar M_{T_{\mathrm{av}}}$ with $a\ge0$. For every fixed $\vartheta\in(0,1)$ and all sufficiently large $N$, the following hold whenever $0<\beta_0\lesssim_k\lambda^{-1}N^{-\tfrac{k+1}{2}-\vartheta}$ and $T_{\mathrm{av}}\gtrsim_k\beta_0^{-2}\lambda^{-2}N^{k+1+2\vartheta}$, with sufficiently small and large implicit constants, respectively:
\begin{enumerate}
    \item If $k$ is odd, then, with probability at least $1-\exp(-c_kN^{(1-\vartheta)/2})$, 
    \begin{align*}
        \frac{1}{N}d_H\bigl(\sgn(v),\sgn(u_H)\bigr)
        \le\min\{N^{-\vartheta},\lambda^2\sigma_N^2N^{-\vartheta/2}\},
    \end{align*}
    where $\sigma_N\asymp_k N^{-(k-1)/4}$ is as defined in \cref{avg-discrete-grad-snr}.

    \item If $k$ is even, then, with probability at least $1-\exp(-c_kN)$,
    \begin{align*}
        \bigl|m(\sgn(v))\bigr|
        \ge 1-C_k\left(\frac{N^{\tfrac{k-2}{2}}}{\lambda^2}+N^{-1-2\vartheta}\right).
    \end{align*}
\end{enumerate}
\end{theorem}

\begin{proof}
Throughout the proof, we assume that $N$ is sufficiently large. We first show that the rescaled ergodic second moment $\widehat\cA\doteq\beta_0^{-1}(\bar M_{T_{\mathrm{av}}}-I_{N+1})$ is close to $\cA_H$ in operator norm. Since $\beta_0>0$, $(v,a)$ is also a top unit eigenvector of $\widehat\cA$. In the notation of \cref{thm:tpca-averaged-differential-recovery}, we have
\begin{align*}
    \bar M_{T_{\mathrm{av}}}
    =\begin{pmatrix}
        \Sigmabar_{T_{\mathrm{av}}} & \Xbar_{T_{\mathrm{av}}}\\
        \Xbar_{T_{\mathrm{av}}}^\top & 1
    \end{pmatrix}.
\end{align*}
Since the bottom-right entries of $\bar M_{T_{\mathrm{av}}}-I_{N+1}$ and $\cA_H$ both vanish, we obtain
\begin{align*}
    \widehat\cA-\cA_H
    =\begin{pmatrix}
        \beta_0^{-1}(\Sigmabar_{T_{\mathrm{av}}}-I)-A_H & 0\\
        0 & 0
    \end{pmatrix}
    +\begin{pmatrix}
        0 & \beta_0^{-1}\Xbar_{T_{\mathrm{av}}}-u_H\\
        (\beta_0^{-1}\Xbar_{T_{\mathrm{av}}}-u_H)^\top & 0
    \end{pmatrix}.
\end{align*}
The first matrix on the right-hand side is block diagonal, and hence its operator norm equals that of its top-left block. The second matrix has zero diagonal blocks, and its operator norm equals the Euclidean norm of its top-right block. Consequently, the triangle inequality implies
\begin{align}
\label{eq:second-moment-blockwise-error}
    \norm{\widehat\cA-\cA_H}_{\mathrm{op}}
    \le\norm{\beta_0^{-1}(\Sigmabar_{T_{\mathrm{av}}}-I)-A_H}_{\mathrm{op}}
    +\norm{\beta_0^{-1}\Xbar_{T_{\mathrm{av}}}-u_H}.
\end{align}

By \cref{eq:tpca-averaging-normalized-errors} in the proof of \cref{thm:tpca-averaged-differential-recovery}, the following event has probability at least $1-\exp(-c_kN)$.
\begin{align*}
    \cE\doteq\left\{\norm{\beta_0^{-1}\Xbar_{T_{\mathrm{av}}}-u_H}\le\eps_N\right\}
    \cap\left\{\norm{\beta_0^{-1}(\Sigmabar_{T_{\mathrm{av}}}-I)-A_H}_{\mathrm{op}}\le\sqrt N\,\eps_N\right\}
\end{align*}
Since our assumptions on $\beta_0$ and $T_{\mathrm{av}}$ are those of \cref{thm:tpca-averaged-differential-recovery} with $\vartheta+\nicefrac{1}{2}$ in place of $\vartheta$, the computation in \cref{eq:tpca-averaging-error-scale} also shows that $\sqrt N\eps_N\lesssim_k\lambda N^{-\frac{k-1}{2}-\vartheta}$. Henceforth, we work on $\cE$. By \cref{eq:second-moment-blockwise-error} and the definition of $\cE$, we obtain
\begin{align}
\label{eq:second-moment-augmented-error}
    \norm{\widehat\cA-\cA_H}_{\mathrm{op}}
    \le(1+\sqrt N)\eps_N
    \lesssim_k\lambda N^{-\tfrac{k-1}{2}-\vartheta}.
\end{align}

We first consider odd $k$. By \cref{avg-discrete-grad-snr,avg-discrete-hess-snr}, $A_H=0$ and $u_H=\lambda\sigma_N^2\xopt+\zeta$, where $\zeta\sim\cN(0,\sigma_N^2I)$ and $\sigma_N^2\asymp_kN^{-\frac{k-1}{2}}$. As in the proof of \cref{thm:tpca-averaged-differential-recovery}, let $\rho_N\doteq\min\{N^{-\vartheta},\lambda^2\sigma_N^2N^{-\vartheta/2}\}$, and let $h_N\doteq\tfrac18\sigma_N\min\{\lambda\sigma_NN^{-\vartheta},\lambda^2\sigma_N^2N^{-\vartheta/2}\}$. Consider the events
\begin{align*}
    \cE_{\mathrm{sig}}\doteq\left\{\dotp{u_H}{\xopt}\ge\tfrac12\lambda\sigma_N^2N\right\},
    \qquad
    \cE_{\mathrm{margin}}\doteq\left\{\sum_{i=1}^N\ind{\abs{(u_H)_i}\le h_N}\le\tfrac12N\rho_N\right\}.
\end{align*}
Since $\dotp{u_H}{\xopt}-\lambda\sigma_N^2N=\dotp{\zeta}{\xopt}\sim\cN(0,\sigma_N^2N)$ and $\lambda\sigma_N\gtrsim_kN^{-\nicefrac14}$, a Gaussian tail bound implies $\Prb[\cE_{\mathrm{sig}}^c]\le\exp(-\lambda^2\sigma_N^2N/8)\le\exp(-c_k\sqrt N)$. Moreover, as shown in the proof of \cref{thm:tpca-averaged-differential-recovery}, every coordinate of $u_H$ satisfies $\Prb[\abs{(u_H)_i}\le h_N]\le\rho_N/4$. Since these coordinates are independent, a Chernoff bound implies $\Prb[\cE_{\mathrm{margin}}^c]\le\exp(-cN\rho_N)$. Henceforth, we work on $\cE\cap\cE_{\mathrm{sig}}\cap\cE_{\mathrm{margin}}$.

We next show that $v$ is close to a positive multiple of $u_H$. Since we work on $\cE_{\mathrm{sig}}$, we have
\begin{align}
\label{eq:second-moment-gradient-norm}
    \norm{u_H}
    \ge\frac{\dotp{u_H}{\xopt}}{\sqrt N}
    \ge\frac12\lambda\sigma_N^2\sqrt N
    >0.
\end{align}
Moreover, since $A_H=0$, we have $\cA_H=\bigl(\begin{smallmatrix}0&u_H\\u_H^\top&0\end{smallmatrix}\bigr)$, which has top unit eigenvector $w\doteq\tfrac{1}{\sqrt2}(u_H/\norm{u_H},1)$ and eigengap $\norm{u_H}$, since its remaining eigenvalues are $0$ and $-\norm{u_H}$. Applying the Davis--Kahan theorem~\citep[Theorem~4.5.5]{vershynin2018high} to $\widehat\cA$ and $\cA_H$, we obtain
\begin{align}
\label{eq:second-moment-davis-kahan}
    \min_{s\in\{\pm1\}}\norm{(v,a)-sw}
    \lesssim\frac{\norm{\widehat\cA-\cA_H}_{\mathrm{op}}}{\norm{u_H}}
    \lesssim_k\frac{\lambda N^{-\frac{k-1}{2}-\vartheta}}{\lambda\sigma_N^2\sqrt N}
    \lesssim_kN^{-\nicefrac12-\vartheta},
\end{align}
where the second inequality uses \cref{eq:second-moment-augmented-error,eq:second-moment-gradient-norm}, and the last one uses $\sigma_N^2\asymp_kN^{-\frac{k-1}{2}}$.

Since the Davis--Kahan theorem determines $(v,a)$ only up to a global sign, we next show that the minimum in \cref{eq:second-moment-davis-kahan} is attained at $s=1$. Indeed, since $a\ge0$ and the last coordinate of $w$ equals $1/\sqrt2$, the last coordinate of $(v,a)+w$ is at least $1/\sqrt2$, and hence
\begin{align*}
    \norm{(v,a)+w}
    \ge a+\frac{1}{\sqrt2}
    \ge\frac{1}{\sqrt2}.
\end{align*}
In contrast, the right-hand side of \cref{eq:second-moment-davis-kahan} is thus smaller than $1/\sqrt2$ for all sufficiently large $N$. Consequently, the minimum in \cref{eq:second-moment-davis-kahan} is attained at $s=1$. Since the first $N$ coordinates of $w$ are $u_H/(\sqrt2\norm{u_H})$, we conclude from \cref{eq:second-moment-davis-kahan} that
\begin{align}
\label{eq:second-moment-eigenvector-error}
    \norm{v-\frac{u_H}{\sqrt2\norm{u_H}}}
    \le\norm{(v,a)-w}
    \lesssim\frac{\norm{\widehat\cA-\cA_H}_{\mathrm{op}}}{\norm{u_H}}.
\end{align}

We now compare the signs of $v$ and $u_H$. Fix $i\in[N]$. If $\sgn(v_i)\ne\sgn((u_H)_i)$, then
\begin{align}
\label{eq:second-moment-sign-mismatch}
    \left|v_i-\frac{(u_H)_i}{\sqrt2\norm{u_H}}\right|
    =\abs{v_i}+\frac{\abs{(u_H)_i}}{\sqrt2\norm{u_H}}
    \ge\frac{\abs{(u_H)_i}}{\sqrt2\norm{u_H}}.
\end{align}
Since we work on $\cE_{\mathrm{margin}}$, we obtain the following via \cref{eq:second-moment-sign-mismatch} and $\ind{\abs{z}\ge t}\le z^2/t^2$ for $t>0$:
\begin{align}
\label{eq:second-moment-hamming-bound}
    \frac1Nd_H\bigl(\sgn(v),\sgn(u_H)\bigr)
    &=\frac1N\sum_{i=1}^N\ind{\sgn(v_i)\ne\sgn((u_H)_i)}\nonumber\\
    &=\frac1N\sum_{i=1}^N\ind{\sgn(v_i)\ne\sgn((u_H)_i)}\ind{\abs{(u_H)_i}\le h_N}\nonumber\\
    &\quad+\frac1N\sum_{i=1}^N\ind{\sgn(v_i)\ne\sgn((u_H)_i)}\ind{\abs{(u_H)_i}>h_N}\nonumber\\
    &\le\frac1N\sum_{i=1}^N\ind{\abs{(u_H)_i}\le h_N}
      +\frac1N\sum_{i=1}^N\ind{\left|v_i-\frac{(u_H)_i}{\sqrt2\norm{u_H}}\right|>\frac{h_N}{\sqrt2\norm{u_H}}}\nonumber\\
    &\le\frac{\rho_N}{2}
      +\frac{2\norm{u_H}^2}{Nh_N^2}\sum_{i=1}^N\left|v_i-\frac{(u_H)_i}{\sqrt2\norm{u_H}}\right|^2\nonumber\\
    &=\frac{\rho_N}{2}
      +\frac{2\norm{u_H}^2}{Nh_N^2}\norm{v-\frac{u_H}{\sqrt2\norm{u_H}}}^2.
\end{align}
Moreover, we obtain the following via \cref{eq:second-moment-eigenvector-error,eq:second-moment-augmented-error}:
\begin{align*}
    \frac{2\norm{u_H}^2}{Nh_N^2}\norm{v-\frac{u_H}{\sqrt2\norm{u_H}}}^2
    &\lesssim\frac{\norm{\widehat\cA-\cA_H}_{\mathrm{op}}^2}{Nh_N^2}
    \lesssim_k\frac{\lambda^2\sigma_N^4N^{-2\vartheta}}{N\sigma_N^2\min\{\lambda^2\sigma_N^2N^{-2\vartheta},\lambda^4\sigma_N^4N^{-\vartheta}\}}\\
    &=\max\left\{\frac1N,\frac{N^{-1-\vartheta}}{\lambda^2\sigma_N^2}\right\}
    \lesssim_kN^{-1}+N^{-\nicefrac12-\vartheta}
    \le\frac{\rho_N}{2}.
\end{align*}
Substituting this bound into \cref{eq:second-moment-hamming-bound}, we obtain
\begin{align*}
    \frac1Nd_H\bigl(\sgn(v),\sgn(u_H)\bigr)\le\rho_N.
\end{align*}
Finally, since $N\rho_N\gtrsim_kN^{(1-\vartheta)/2}$, a union bound over $\cE$, $\cE_{\mathrm{sig}}$ and $\cE_{\mathrm{margin}}$ proves the claim for odd $k$ after adjusting constants appropriately.

We now consider even $k$. By \cref{avg-discrete-grad-snr,avg-discrete-hess-snr}, $u_H=0$ and $A_H=A_{\mathrm{sig}}+W_N$, where $A_{\mathrm{sig}}\doteq\lambda\gamma_N(\plant(\plant)^\top-I/N)$ with $\gamma_N\asymp_kN^{-\frac{k-2}{2}}$. Moreover, by \cref{avg-discrete-hess-snr}, $\norm{W_N}_{\mathrm{op}}\le C_kN^{-\frac{k-2}{4}}$ with probability at least $1-\exp(-cN)$. Henceforth, we also work on this event. Since $\cA_H=\bigl(\begin{smallmatrix}A_{\mathrm{sig}}+W_N&0\\0&0\end{smallmatrix}\bigr)$, we infer from \cref{eq:second-moment-augmented-error} and the triangle inequality that
\begin{align*}
    \norm{\widehat\cA-\begin{pmatrix}A_{\mathrm{sig}}&0\\0&0\end{pmatrix}}_{\mathrm{op}}
    \le\norm{\widehat\cA-\cA_H}_{\mathrm{op}}+\norm{W_N}_{\mathrm{op}}
    \lesssim_k\lambda N^{-\tfrac{k-1}{2}-\vartheta}+N^{-\tfrac{k-2}{4}}.
\end{align*}
The matrix $\bigl(\begin{smallmatrix}A_{\mathrm{sig}}&0\\0&0\end{smallmatrix}\bigr)$ has top unit eigenvector $(\plant,0)$ and eigengap $\lambda\gamma_N(1-\nicefrac{1}{N})$, since its remaining eigenvalues are $0$ and $-\lambda\gamma_N/N$. Applying the Davis--Kahan theorem to $\widehat\cA$ and this matrix, we obtain
\begin{align*}
    \min_{s\in\{\pm1\}}\norm{v-s\plant}
    \le\min_{s\in\{\pm1\}}\norm{(v,a)-s(\plant,0)}
    \lesssim_k\frac{\lambda N^{-\frac{k-1}{2}-\vartheta}+N^{-\frac{k-2}{4}}}{\lambda N^{-\frac{k-2}{2}}}
    =N^{-\nicefrac12-\vartheta}+\frac{N^{\tfrac{k-2}{4}}}{\lambda}.
\end{align*}
Fix $s\in\{\pm1\}$. If $\sgn(v_i)\ne s\xopt_i$, then $\abs{v_i-s\plant_i}\ge\abs{\plant_i}=1/\sqrt N$. Squaring and summing over $i\in[N]$, we obtain $\frac1Nd_H(\sgn(v),s\xopt)\le\norm{v-s\plant}^2$. Since $\abs{m(x)}=1-\frac2N\min_{s\in\{\pm1\}}d_H(x,s\xopt)$ for every $x\in\Cube$, we conclude that
\begin{align*}
    1-\bigl|m(\sgn(v))\bigr|
    \le2\min_{s\in\{\pm1\}}\norm{v-s\plant}^2
    \lesssim_k\frac{N^{\tfrac{k-2}{2}}}{\lambda^2}+N^{-1-2\vartheta}.
\end{align*}
A union bound proves the claim for even $k$.
\end{proof}

\begin{remark}
\label{rem:second-moment-parameter-cost}
Compared to \cref{thm:tpca-averaged-differential-recovery}, \cref{thm:recovery-ergodic-second-moments} requires an inverse temperature $\beta_0$ that is smaller by a factor of $\sqrt N$, and hence an averaging time that is longer by a factor of $N^2$. This is because, by \cref{eq:tpca-averaging-normalized-errors} in the proof of \cref{thm:tpca-averaged-differential-recovery}, the estimation error of the ergodic second moment is larger than that of the ergodic average by a factor of $\sqrt N$. Moreover, the top eigenvector of $\bar M_{T_{\mathrm{av}}}$ inherits the estimation error of the ergodic second moment, irrespective of the parity of $k$.
\end{remark}

\subsection{Boosting Weak Recovery to Persistent Strong Recovery}
\label{subsec:boosting-weak-to-strong-recovery}

The following theorem shows that, for even $k$, once Glauber dynamics has weakly recovered the planted signal upto a global sign, running the process for an additional $\log N$ time suffices to achieve $1-\poly(\nicefrac{1}{N})$ correlation with high probability. Beyond that, the correlation stays above this level throughout any $\poly(N)$ time horizon. 

The proof relies on the fact that the algorithmic threshold satisfies $N^{\tfrac{k-2}{4}}\ge\sqrt N$ for even $k$, while the discrete partial derivatives of $H_0$ are uniformly bounded by $O(\sqrt N)$ with high probability by \cref{lem:kspin-sobolev-norm-concentration}. Consequently, once the correlation is bounded away from zero, the signal term dominates the disorder at every coordinate, and with high probability, each update of Glauber dynamics aligns the updated spin with the planted vector up to a global sign.

\begin{theorem}[Persistent Strong Recovery after Weak Recovery]
\label{thm:strong-recovery-after-weak-recovery}
Let $k$ be even, let $(X_t)_{t\ge0}$ denote Glauber dynamics for $H$ at inverse temperature $\beta>0$, and fix arbitrary constants $\rho\in(0,1)$ and $\eta\in(0,1)$. Assume $\lambda\ge C_{\beta,k,\rho}N^{\tfrac{k-2}{4}}$, where $C_{\beta,k,\rho}$ is a sufficiently large constant, and set $T_{\mathrm{str}}\doteq\log N$. Then, the following holds for all sufficiently large $N$, every deterministic $T$ satisfying $T_{\mathrm{str}}\le T\le\poly(N)$, and any initial state $X_0$ with $\abs{m(X_0)}\ge\rho$, possibly depending on $H$,
\begin{align*}
    \Prb\left[\inf_{t\in[T_{\mathrm{str}},T]}\abs{m(X_t)}\ge1-N^{-\eta}\right]
    \ge1-\exp(-c_kN)-\exp(-cN^{1-\eta})-NT\exp\left(-k\beta\lambda\left(\frac{\rho}{2}\right)^{k-1}\right).
\end{align*}
\end{theorem}

\begin{proof}
Throughout the proof, we assume that $N$ is sufficiently large. Let $X_0$ be any initial state with $\abs{m(X_0)}\ge\rho$, possibly depending on $H$, and define
\begin{align*}
    \widetilde x^\star\doteq\sgn\bigl(m(X_0)\bigr)\xopt,
    \qquad
    \widetilde m(x)\doteq\frac1N\dotp{\widetilde x^\star}{x}.
\end{align*}
Observe that $\widetilde m(X_0)=\abs{m(X_0)}\ge\rho$ and $\abs{m(x)}\ge\widetilde m(x)$ for every $x\in\Cube$. Moreover, since $k$ is even, $H=N\lambda\widetilde m^k+H_0$. Consider the event
\begin{align*}
    \cE\doteq\left\{\max_{x\in\Cube}\norm{\dgrad H_0(x)}^2\le C_1N\right\}.
\end{align*}
By \cref{lem:kspin-sobolev-norm-concentration}, we can choose $C_1$ such that $\Prb[\cE]\ge1-\exp(-c_kN)$. Henceforth, we work on $\cE$.

We first show that each update from a state $x$ with $\widetilde m(x)\ge\rho$ sets the updated spin to its value in $\widetilde x^\star$ with high probability. For $i\in[N]$, define $\widetilde m_{-i}(x)\doteq N^{-1}\sum_{j\ne i}\widetilde x^\star_jx_j$. If $\widetilde m(x)\ge4/N$, then $\widetilde m_{-i}(x)\pm1/N\ge\widetilde m(x)-2/N\ge\widetilde m(x)/2$. Hence, repeating the steps in \cref{eq:strong-recovery-signal-field} in the proof of \cref{thm:strong-recovery-after-critical-correlation} with $\widetilde x^\star$ and $\widetilde m$, we obtain
\begin{align}
\label{eq:weak-to-strong-signal-field}
    N\lambda\widetilde x^\star_i\partial_i(\widetilde m^k)(x)
    =\frac{N\lambda}{2}\left[
        \left(\widetilde m_{-i}(x)+\frac1N\right)^k
        -\left(\widetilde m_{-i}(x)-\frac1N\right)^k
    \right]
    \ge k\lambda\left(\frac{\widetilde m(x)}{2}\right)^{k-1}.
\end{align}
Since $k\ge4$, we have $\lambda\ge C_{\beta,k,\rho}N^{\tfrac{k-2}{4}}\ge C_{\beta,k,\rho}\sqrt N$. Hence, the following holds on $\cE$ after setting $C_{\beta,k,\rho}$ to be a sufficiently large constant:
\begin{align}
\label{eq:weak-to-strong-disorder-bound}
    \abs{\partial_iH_0(x)}
    \le\norm{\dgrad H_0(x)}
    \le\sqrt{C_1N}
    \le\frac{k}{2}\lambda\left(\frac{\rho}{2}\right)^{k-1}
    \qquad\forall\,x\in\Cube,\ i\in[N].
\end{align}
Since $H=N\lambda\widetilde m^k+H_0$ and $\rho\ge4/N$, we infer from \cref{eq:weak-to-strong-signal-field,eq:weak-to-strong-disorder-bound} that the following holds for every $x\in\Cube$ with $\widetilde m(x)\ge\rho$ and every $i\in[N]$:
\begin{align}
\label{eq:weak-to-strong-derivative-bound}
    \widetilde x^\star_i\partial_iH(x)
    \ge k\lambda\left(\frac{\widetilde m(x)}{2}\right)^{k-1}-\abs{\partial_iH_0(x)}
    \ge k\lambda\left(\frac{\rho}{2}\right)^{k-1}-\sqrt{C_1N}
    \ge\frac{k}{2}\lambda\left(\frac{\rho}{2}\right)^{k-1}.
\end{align}
For $x\in\Cube$ and $i\in[N]$, let $q_i(x)$ denote the probability that an update of coordinate $i$ at state $x$ sets its spin to $-\widetilde x^\star_i$. Since the site update kernel of Glauber dynamics is $q(z\mid i,t,x)=\nicefrac{1}{2}(1+z\tanh(\beta\partial_iH(x)))$ by \cref{subsec:cts-glauber}, we infer from \cref{eq:weak-to-strong-derivative-bound} that the following holds for every $x\in\Cube$ with $\widetilde m(x)\ge\rho$ and every $i\in[N]$:
\begin{align}
\label{eq:weak-to-strong-bad-update}
    q_i(x)
    =\frac{1}{1+\exp\bigl(2\beta\widetilde x^\star_i\partial_iH(x)\bigr)}
    \le\exp\left(-k\beta\lambda\left(\frac{\rho}{2}\right)^{k-1}\right)
    \doteq p.
\end{align}

We now realize the trajectory of Glauber dynamics using auxiliary uniform random variables. Let $(T_n)_{n\ge1}$ denote the arrival times of a rate $N$ Poisson process, with marks $(I_n)_{n\ge1}\iidsim\Uniform([N])$ and $(\omega_n)_{n\ge1}\iidsim\Uniform([0,1])$. The arrival times and marks are mutually independent, and independent of $H$ and $X_0$. At time $T_n$, the process sets coordinate $I_n=i$ of $X_{T_n}$ to $-\widetilde x^\star_i$ if $\omega_n\le q_i(X_{T_n^-})$, and to $\widetilde x^\star_i$ otherwise. By \cref{subsec:cts-glauber}, the resulting trajectory is Glauber dynamics for $H$ at inverse temperature $\beta$, initialized at $X_0$.

We call the arrival $T_n$ bad if $\omega_n\le p$, and let $B_T$ denote the number of bad arrivals in $[0,T]$. Since $p$ is deterministic, Poisson thinning implies $B_T\sim\Poisson(NTp)$, and hence
\begin{align*}
    \Prb[B_T>0]
    =1-e^{-NTp}
    \le NTp.
\end{align*}
Moreover, let $U_t$ denote the number of coordinates that are not updated during $[0,t]$. Precisely, define
\begin{align*}
    N_i(t)\doteq\#\{n:T_n\le t,\ I_n=i\},
    \qquad
    U_t\doteq\sum_{i=1}^N\ind{N_i(t)=0}.
\end{align*}
By Poisson thinning, $(N_i(t))_{t\ge0}$ are independent rate one Poisson processes for $i\in[N]$, and hence $U_t\sim\bin(N,e^{-t})$. In particular, since $T_{\mathrm{str}}=\log N$, we have $U_{T_{\mathrm{str}}}\sim\bin(N,\nicefrac{1}{N})$, and a Chernoff bound implies
\begin{align*}
    \Prb\left[U_{T_{\mathrm{str}}}>\frac12N^{1-\eta}\right]
    \le\exp(-cN^{1-\eta}).
\end{align*}
Henceforth, we work on the event $\cE\cap\{B_T=0\}\cap\{U_{T_{\mathrm{str}}}\le\tfrac12N^{1-\eta}\}$.

We now show by induction that each update during $[0,T]$ sets the updated spin to its value in $\widetilde x^\star$. Consider any $n$ with $T_n\le T$, and suppose that the claim holds for all earlier updates. Since such updates can only increase $\widetilde m$, we have
\begin{align*}
    \widetilde m(X_{T_n^-})
    \ge\widetilde m(X_0)
    \ge\rho.
\end{align*}
Hence, by \cref{eq:weak-to-strong-bad-update} and since we work on $\{B_T=0\}$, we obtain $q_{I_n}(X_{T_n^-})\le p<\omega_n$. Therefore, coordinate $I_n$ is set to $\widetilde x^\star_{I_n}$, which completes the induction.

Consequently, for every $t\in[0,T]$, each coordinate that is updated during $[0,t]$ agrees with $\widetilde x^\star$ at time $t$, and hence $d_H(X_t,\widetilde x^\star)\le U_t$. Since $U_t$ is nonincreasing in $t$, we conclude that
\begin{align*}
    \abs{m(X_t)}
    \ge\widetilde m(X_t)
    =1-\frac{2}{N}d_H(X_t,\widetilde x^\star)
    \ge1-\frac{2}{N}U_{T_{\mathrm{str}}}
    \ge1-N^{-\eta}
    \qquad\forall\,t\in[T_{\mathrm{str}},T].
\end{align*}
The proof is concluded via a union bound over $\cE^c$, $\{B_T>0\}$ and $\{U_{T_{\mathrm{str}}}>\tfrac12N^{1-\eta}\}$.
\end{proof}

\subsection{Proof of \texorpdfstring{\cref{thm:persistent-strong-recovery-every-order}}{Theorem~\ref{thm:persistent-strong-recovery-every-order}}}
\label{subsec:proof-persistent-strong-recovery-every-order}

\begin{proof}
Throughout the proof, we assume that $N$ is sufficiently large. Our choices of $\beta_0$ and $T_{\mathrm{av}}$ satisfy the assumptions of \cref{thm:recovery-ergodic-second-moments}, since $\beta_0^{-2}\lambda^{-2}N^{k+1+2\vartheta}\asymp_kN^{2k+2+4\vartheta}$. By construction, $(X_t)_{t\ge0}$ is Glauber dynamics for $H$ at inverse temperature $\beta$, started from $X_0=\sgn(v)$.

We first analyze the case when $k$ is odd, wherein we prove the claim with $m(X_t)$ in place of $\abs{m(X_t)}$. Without loss of generality, let $\lambda=C_{\beta,k}N^\gamma$ with $\gamma\ge\tfrac{k-2}{4}$. We infer from \cref{thm:recovery-ergodic-second-moments} that the following holds with probability at least $1-\exp(-c_kN^{(1-\vartheta)/2})$:
\begin{align}
\label{eq:every-order-warm-start}
    X_0=\sgn(v)
    \in\cB_H\bigl(\sgn(u_H),\min\{N^{-\vartheta},\lambda^2\sigma_N^2N^{-\vartheta/2}\}\bigr).
\end{align}

We first consider the regime $\gamma\ge\tfrac{k-2}{4}+\tfrac{1}{4k}$. In this regime, we show that the correlation of $X_0$ is already large enough to apply \cref{thm:strong-recovery-after-critical-correlation}. Define $a\doteq\tfrac{k-1}{4k}$, and note that $\gamma-\tfrac{k-1}{4}\ge-a$, so that $\lambda\sigma_N\gtrsim_kC_{\beta,k}N^{-a}$. Combining \cref{eq:every-order-warm-start} with \cref{cor:avg-discrete-grad-correlation}, and using $\sigma_N\asymp_kN^{-\tfrac{k-1}{4}}$ and $\min\{N^{-\vartheta},\lambda^2\sigma_N^2N^{-\vartheta/2}\}\le N^{-\vartheta/2}\min\{1,\lambda\sigma_N\}$, we obtain the following with probability at least $1-\exp(-c_kN^{(1-\vartheta)/2})$, after choosing $C_{\beta,k}$ sufficiently large:
\begin{align}
\label{eq:every-order-initial-correlation}
    m(X_0)
    &\ge m(\sgn(u_H))-\frac{2}{N}d_H\bigl(X_0,\sgn(u_H)\bigr)\nonumber\\
    &\ge\bigl(c_k-2N^{-\vartheta/2}\bigr)\min\{1,\lambda\sigma_N\}
    \ge\frac{c_k}{2}\min\{1,\lambda\sigma_N\}
    \ge N^{-a}.
\end{align}
Let $\bar r_N\doteq N^{-a}$. Since $\gamma\ge(k-1)a$, we have $\lambda\bar r_N^{k-1}\ge C_{\beta,k}$. Hence, applying \cref{thm:strong-recovery-after-critical-correlation} on the event in \cref{eq:every-order-initial-correlation} with $C=1$ and $\tau_{\bar r_N}=0$, and using $T_{\mathrm{str}}\ge1+2\eta\log N$ for sufficiently large $N$, we obtain the following via a union bound:
\begin{align*}
    \Prb\left[\inf_{t\in[T_{\mathrm{str}},T]}m(X_t)<1-N^{-\eta}\right]
    \le\exp\bigl(-c_kN^{(1-\vartheta)/2}\bigr)+3\exp\bigl(-cN\min\{\bar r_N^2,N^{-2\eta}\}\bigr).
\end{align*}
Finally, since $1-2a=\tfrac{k+1}{2k}>\nicefrac{1}{2}>\tfrac{1-\vartheta}{2}$, the claim follows after adjusting constants appropriately.

We now consider the regime $\tfrac{k-2}{4}\le\gamma<\tfrac{k-2}{4}+\tfrac{1}{4k}$, and let $r_N\doteq\lambda^{-\tfrac{1}{k-1}}$, as in \cref{subsec:critical-recovery-warm-start}. In this regime, our argument closely mirrors the proof of \cref{thm:persistent-strong-recovery-algorithmic-threshold}. We first show that the second phase attains critical correlation in constant time. By \cref{eq:every-order-warm-start}, we have $X_0\in\cB_H(\sgn(u_H),\lambda^2\sigma_N^2N^{-\vartheta/2})$. Since \cref{thm:critical-recovery-warm-start} holds simultaneously for all initializations in this Hamming ball, we may apply it with $\vartheta/2$ in place of $\vartheta$, and we obtain the following with probability at least $1-\exp(-c_kN^{(1-\vartheta)/2})-\exp(-cN\lambda^2\sigma_N^2)$ for some sufficiently large constant $C$:
\begin{align}
\label{eq:every-order-critical-time}
    T_*(X_0)=\inf\bigl\{t\ge0:m(X_t)\ge Cr_N\bigr\}\le C_1N^{-(k-1)\left(\gamma-\tfrac{k-2}{4}\right)}.
\end{align}
Let $\bar r_N\doteq Cr_N$. Since $\lambda r_N^{k-1}=1$, we have $\lambda\bar r_N^{k-1}=C^{k-1}\ge C_{\beta,k}$. Hence, applying \cref{thm:strong-recovery-after-critical-correlation} on the event in \cref{eq:every-order-critical-time} with $\tau_{\bar r_N}=T_*(X_0)$, and using $T_{\mathrm{str}}\ge T_*(X_0)+1+2\eta\log N$ for sufficiently large $N$, we obtain the following via a union bound:
\begin{align*}
    \Prb\left[\inf_{t\in[T_{\mathrm{str}},T]}m(X_t)<1-N^{-\eta}\right]
    \le\exp\bigl(-c_kN^{(1-\vartheta)/2}\bigr)+\exp\bigl(-cN\lambda^2\sigma_N^2\bigr)+3\exp\bigl(-cN\min\{\bar r_N^2,N^{-2\eta}\}\bigr).
\end{align*}
Finally, since $\gamma<\tfrac{k-2}{4}+\tfrac{1}{4k}$, we have $N\bar r_N^2\gtrsim N^{\tfrac{k+1}{2k}}$. As $N\lambda^2\sigma_N^2\gtrsim\sqrt N$ and $\tfrac{k+1}{2k}>\nicefrac{1}{2}>\tfrac{1-\vartheta}{2}$, the claim follows after adjusting constants appropriately.

Finally, we analyze the case when $k$ is even. Since $C_{\beta,k}$ is sufficiently large, we infer from \cref{thm:recovery-ergodic-second-moments} that the following holds with probability at least $1-\exp(-c_kN)$:
\begin{align}
\label{eq:every-order-weak-recovery}
    \abs{m(X_0)}=\bigl|m(\sgn(v))\bigr|\ge\frac12.
\end{align}
Hence, applying \cref{thm:strong-recovery-after-weak-recovery} on the event in \cref{eq:every-order-weak-recovery} with $\rho=\nicefrac{1}{2}$ and $C_{\beta,k}$ sufficiently large, we obtain the following via a union bound:
\begin{align*}
    \Prb\left[\inf_{t\in[T_{\mathrm{str}},T]}\abs{m(X_t)}<1-N^{-\eta}\right]
    \le2\exp(-c_kN)+\exp(-cN^{1-\eta})+NT\exp\bigl(-k\beta\lambda4^{1-k}\bigr).
\end{align*}
Moreover, since $k\ge4$, we have $\lambda\ge C_{\beta,k}\sqrt N$, so $NT\exp(-k\beta\lambda4^{1-k})\le\exp(-c_{\beta,k}\sqrt N)$ for $T\le\poly(N)$. Finally, since $\nicefrac{1}{2}>\tfrac{1-\vartheta}{2}$ and ${1-\eta>1-2\eta}$, the claim follows after adjusting constants appropriately.
\end{proof}

%% file: appendix.tex
\section{Proof of Auxiliary Lemmas}
\label[appendix]{app:auxiliary-lemmas}

\subsection{Sobolev Norm Concentration for the Spin Glass Hamiltonian}
\label[appendix]{app:kspin-sobolev-proof}

\begin{proof}[Proof of \cref{lem:kspin-sobolev-norm-concentration}]
Since $G$ is the symmetrization of a standard Gaussian $k$-tensor $\widetilde G$, $\dotp{G}{x^{\otimes k}}=\dotp{\widetilde G}{x^{\otimes k}}$. Thus, \cref{eq:tensor-pca-overlap-noise} gives
\begin{align*}
    H_0(x)=N^{-\tfrac{k-1}{2}}\sum_{i_1,\ldots,i_k\in[N]}\widetilde G_{i_1,\ldots,i_k}\cdot\prod_{r=1}^k x_{i_r}.
\end{align*}

Fix $\ell\in\{1,\ldots,k\}$. By repeated applications of \cref{eq:partial-fourier-character} we obtain the following for any $S\subseteq[N]$,
\begin{align}
\label{eq:kspin-character-derivatives}
    \partial_{i_1}\cdots\partial_{i_\ell}\chi_S(x)
    =\begin{cases}
        \chi_{S\setminus\{i_1,\ldots,i_\ell\}}(x), & \text{if } i_1,\ldots,i_\ell\in S\text{ are distinct},\\
        0, & \text{otherwise}.
    \end{cases}
\end{align}

For any multi-index $I=(i_1,\ldots,i_k)\in[N]^k$, define
\begin{align*}
    A(I)=\left\{j\in[N]:\#\{r:i_r=j\}\text{ is odd}\right\}.
\end{align*}
Since $x_j^2=1$ for every $j\in[N]$, the factors with even multiplicity cancel, and hence
\begin{align*}
    \prod_{r=1}^k x_{i_r}=\prod_{j\in A(I)}x_j=\chi_{A(I)}(x).
\end{align*}
Using the above identity to write $H_0$ as a sum of Fourier characters:
\begin{align*}
    H_0(x)=N^{-\tfrac{k-1}{2}}\sum_{I\in[N]^k}\widetilde G_I \cdot \chi_{A(I)}(x).
\end{align*}
Applying \cref{eq:kspin-character-derivatives} term by term, we obtain the following for distinct $i_1,\ldots,i_\ell\in[N]$,
\begin{align}
    \partial_{i_1}\cdots\partial_{i_\ell}H_0(x)
    &=N^{-\tfrac{k-1}{2}}\sum_{I\in[N]^k}\widetilde G_I\cdot
    \partial_{i_1}\cdots\partial_{i_\ell}\chi_{A(I)}(x)\nonumber\\
    &=N^{-\tfrac{k-1}{2}}\sum_{I\in[N]^k}\widetilde G_I\cdot
    \chi_{A(I)\setminus\{i_1,\ldots,i_\ell\}}(x)\cdot
    \ind{\{i_1,\ldots,i_\ell\}\subseteq A(I)}.
\label{eq:kspin-partial-derivatives}
\end{align}

Consider any $a^{(1)},\ldots,a^{(\ell)}\in\bS^{N-1}$ and let $Z=\dgrad^\ell H_0(x)[a^{(1)},\ldots,a^{(\ell)}]$. From the above identity, we obtain
\begin{align*}
    Z
    &=\sum_{j_1,\ldots,j_\ell\in[N]}
    \partial_{j_1}\cdots\partial_{j_\ell}H_0(x)\cdot\prod_{r=1}^{\ell}a_{j_r}^{(r)}
    =N^{-\tfrac{k-1}{2}}\sum_{I\in[N]^k}\widetilde G_I\cdot C_I,
\end{align*}
where
\begin{align*}
    C_I
    =\sum_{\substack{j_1,\ldots,j_\ell\in A(I)\\ j_1,\ldots,j_\ell\ \mathrm{distinct}}}
    \chi_{A(I)\setminus\{j_1,\ldots,j_\ell\}}(x)\cdot\prod_{r=1}^{\ell}a_{j_r}^{(r)}.
\end{align*}
Since $|\chi_A(x)|=1$ and $|A(I)|\le k$, we obtain the following by Cauchy--Schwarz:
\begin{align*}
    |C_I|^2
    &\le\#\{(j_1,\ldots,j_\ell):j_1,\ldots,j_\ell\in A(I)\text{ distinct}\}
    \sum_{\substack{j_1,\ldots,j_\ell\in A(I)\\ j_1,\ldots,j_\ell\ \mathrm{distinct}}}
    \prod_{r=1}^{\ell}(a_{j_r}^{(r)})^2\\
    &\le k^\ell
    \sum_{\substack{j_1,\ldots,j_\ell\in A(I)\\ j_1,\ldots,j_\ell\ \mathrm{distinct}}}
    \prod_{r=1}^{\ell}(a_{j_r}^{(r)})^2.
\end{align*}

Since $\widetilde G_I\iidsim\cN(0,1)$ for $I\in[N]^k$,
\begin{align*}
    \Var(Z)
    &=N^{-(k-1)}\sum_{I\in[N]^k}|C_I|^2\\
    &\le k^\ell N^{-(k-1)}
    \sum_{\substack{j_1,\ldots,j_\ell\in[N]\\ j_1,\ldots,j_\ell\ \mathrm{distinct}}}
    \#\{I\in[N]^k:j_1,\ldots,j_\ell\in A(I)\}
    \prod_{r=1}^{\ell}(a_{j_r}^{(r)})^2\\
    &\le k^\ell N^{-(k-1)}\cdot N^{k-\ell}\cdot k(k-1)\cdots(k-\ell+1)
    \sum_{\substack{j_1,\ldots,j_\ell\in[N]\\ j_1,\ldots,j_\ell\ \mathrm{distinct}}}
    \prod_{r=1}^{\ell}(a_{j_r}^{(r)})^2\\
    &\le k^{2\ell}N^{1-\ell}.
\end{align*}
Therefore, by Gaussian concentration, for any fixed $x\in\Cube$ and $a^{(1)},\ldots,a^{(\ell)}\in\bS^{N-1}$,
\begin{align*}
    \Prb\left[
        \left|\dgrad^\ell H_0(x)[a^{(1)},\ldots,a^{(\ell)}]\right|\ge u
    \right]
    \le2\exp\left(-\frac{c_{k,\ell}u^2}{N^{1-\ell}}\right),
    \qquad u\ge0.
\end{align*}

Set $\eps=(4\ell)^{-1}$ and let $\cN$ be a fixed $\eps$-net of $\bS^{N-1}$ with $|\cN|\le(C\ell)^N$. By the standard $\eps$-net comparison for $\ell$-linear forms, we obtain
\begin{align*}
    \norm{\dgrad^\ell H_0(x)}_{\mathrm{op}}
    \le2\sup_{a^{(1)},\ldots,a^{(\ell)}\in\cN}
    \left|\dgrad^\ell H_0(x)[a^{(1)},\ldots,a^{(\ell)}]\right|.
\end{align*}
Taking a union bound over $x\in\Cube$ and over $\cN^\ell$, and choosing the constant $C_{k,\ell}$ sufficiently large, we conclude that, for every $t\ge0$,
\begin{align}
\label{eq:kspin-discrete-derivative-concentration}
    \Prb\left[
        \max_{x\in\Cube}\norm{\dgrad^\ell H_0(x)}_{\mathrm{op}}
        \ge C_{k,\ell}N^{1-\nicefrac{\ell}{2}}\left(1+\sqrt{\frac{t}{N}}\right)
    \right]
    \le\exp\left(-c_{k,\ell}(N+t)\right).
\end{align}
For $\ell=0$, we have $H_0(x)\sim\cN(0,N)$ for every $x\in\Cube$ by \cref{eq:kspin-covariance}. Hence, by a union bound, we obtain
\begin{align*}
    \Prb\left[\max_{x\in\Cube}|H_0(x)|\ge u\right]
    \le2^{N+1}\exp\left(-\frac{u^2}{2N}\right),
    \qquad u\ge0.
\end{align*}
Choosing the constants appropriately, we conclude that \cref{eq:kspin-discrete-derivative-concentration} holds for $\ell=0$ as well. Moreover, since $|A(I)|\le k$, we infer from repeated applications of \cref{eq:partial-fourier-character} that $\dgrad^\ell H_0=0$ for $\ell>k$.

Taking a union bound over $0\le\ell\le\min(n,k)$ in \cref{eq:kspin-discrete-derivative-concentration} and adjusting constants appropriately, we conclude from \cref{eq:discrete-sobolev-norm} that the following holds for every $t\ge0$,
\begin{align*}
    \Prb\left[
        \norm{H_0}_{\cD^n}
        \ge C_{k,n}N\left(1+\sqrt{\frac{t}{N}}\right)
    \right]
    \le\exp\left(-c_{k,n}(N+t)\right).
\end{align*}
\end{proof}

\subsection{Sobolev Norm Bounds for the Signal Term}
\label[appendix]{app:signal-sobolev-proof}

\begin{proof}[Proof of \cref{lem:overlap-sobolev-bounds}]
Since $|m(x)|\le1$, $\norm{\dgrad m(x)}=\norm{\nicefrac{\xopt}{N}}=N^{-\nicefrac{1}{2}}$, and $\dgrad^\ell m=0$ for $\ell\ge2$, the first bound follows from \cref{eq:discrete-sobolev-norm}.

For the second bound, let $g(t)=t^k$. For any $A\subseteq[N]$, define $m_{-A}(x)=\nicefrac{1}{N}\sum_{i\notin A}\xopt_i x_i$, and write $m_{-i}=m_{-\{i\}}$. Since $m(x)=m_{-i}(x)+\tfrac{\xopt_i x_i}{N}$, we have
\begin{align*}
    \partial_i(m^k)(x)
    &=\frac{1}{2}\left[g\left(m_{-i}(x)+\frac{\xopt_i}{N}\right)
    -g\left(m_{-i}(x)-\frac{\xopt_i}{N}\right)\right]\\
    &=\frac{\xopt_i}{2N}\int_{-1}^1g'\left(m_{-i}(x)+t\frac{\xopt_i}{N}\right)\,dt.
\end{align*}

More generally, fix $\ell\ge1$ and distinct $i_1,\ldots,i_\ell\in[N]$, and let $A=\{i_1,\ldots,i_\ell\}$. Applying the fundamental theorem of calculus successively in the remaining coordinates, we obtain
\begin{align*}
    \partial_{i_1}\cdots\partial_{i_\ell}(m^k)(x)
    =\frac{\prod_{r=1}^{\ell}\xopt_{i_r}}{(2N)^\ell}
    \int_{[-1,1]^\ell}
    g^{(\ell)}\left(m_{-A}(x)+\frac{1}{N}\sum_{r=1}^{\ell}t_r\xopt_{i_r}\right)
    \,dt_1\cdots dt_\ell.
\end{align*}
For any $t_1,\ldots,t_\ell\in[-1,1]$, the argument of $g^{(\ell)}$ lies in $[-1,1]$. Hence,
\begin{align*}
    |\partial_{i_1}\cdots\partial_{i_\ell}(m^k)(x)|
    \le\frac{1}{N^\ell}\max_{t\in[-1,1]}|g^{(\ell)}(t)|
    \le C_{k,\ell}N^{-\ell}.
\end{align*}
When two indices coincide, the derivative vanishes because the discrete partial derivatives commute and $\partial_i^2(m^k)=0$, so the same bound holds.

Consider any $a^{(1)},\ldots,a^{(\ell)}\in\bS^{N-1}$. By the triangle inequality and Cauchy--Schwarz, we obtain
\begin{align*}
    \left|\dgrad^\ell(m^k)(x)[a^{(1)},\ldots,a^{(\ell)}]\right|
    &\le\sum_{i_1,\ldots,i_\ell\in[N]}
    |\partial_{i_1}\cdots\partial_{i_\ell}(m^k)(x)|\cdot
    \prod_{r=1}^{\ell}|a_{i_r}^{(r)}|\\
    &\le C_{k,\ell}N^{-\ell}
    \prod_{r=1}^{\ell}\left(\sum_{i=1}^N|a_i^{(r)}|\right)\\
    &\le C_{k,\ell}N^{-\nicefrac{\ell}{2}}.
\end{align*}
Taking the supremum over $a^{(1)},\ldots,a^{(\ell)}\in\bS^{N-1}$, we conclude that
\begin{align}
\label{eq:overlap-power-discrete-derivative-bound}
    \max_{x\in\Cube}\norm{\dgrad^\ell(m^k)(x)}_{\mathrm{op}}
    \le C_{k,\ell}N^{-\nicefrac{\ell}{2}},
    \qquad \ell\ge1.
\end{align}
Since $|m(x)|\le1$, summing over the derivative orders in \cref{eq:discrete-sobolev-norm} and applying \cref{eq:overlap-power-discrete-derivative-bound}, we obtain the following after adjusting constants appropriately.
\begin{align*}
    \norm{m^k}_{\cD^n}=\max_{x\in\Cube}|m(x)|^k+\sum_{\ell=1}^n N^{\nicefrac{\ell}{2}}\max_{x\in\Cube}\norm{\dgrad^\ell(m^k)(x)}_{\mathrm{op}}\le C_{k,n}.
\end{align*}
\end{proof}

\subsection{Properties of the Hessian Observable}
\label[appendix]{app:hessian-observable-proof}

\begin{proof}[Proof of \cref{lem:hessian-observable-properties}]
The first claim follows directly from the hypercube Stein's lemma~\cref{lem:hypercube-stein} and \cref{eq:hessian-recovery-observable}. To prove the Sobolev norm bound, fix $w\in\bS^{N-1}$ and define
\begin{align}
\label{eq:hessian-observable-quadratic-form}
    f_w(x)\doteq w^\top Q_V(x)w
    =\dotp{w}{x}\dotp{w}{\dgrad V(x)}
      -\sum_{j=1}^N w_j^2x_j\partial_jV(x).
\end{align}
Since $\norm{V}_{\cD^2}\le BN$, we have $|f_w(x)|\le(\sqrt N+1)\norm{\dgrad V(x)}\le2BN$ for every $x\in\Cube$.
Write $w^{\odot2}=(w_1^2,\ldots,w_N^2)$ and $y=x\odot w^{\odot2}$. By the discrete product rule~\cref{lem:discrete-product-rule} and $\partial_i^2V=0$, we have
\begin{align}
\label{eq:hessian-observable-product-derivatives}
    \begin{aligned}
        \partial_i\bigl(\dotp{w}{x}\dotp{w}{\dgrad V(x)}\bigr)
        &=w_i\dotp{w}{\dgrad V(x)}
          +\dotp{w}{x}[\dgrad^2V(x)w]_i
          -2x_iw_i[\dgrad^2V(x)w]_i,\\
        \partial_i\left(\sum_{j=1}^N w_j^2x_j\partial_jV(x)\right)
        &=w_i^2\partial_iV(x)+[\dgrad^2V(x)y]_i.
    \end{aligned}
\end{align}
Combining \cref{eq:hessian-observable-quadratic-form,eq:hessian-observable-product-derivatives}, we obtain
\begin{align}
\label{eq:hessian-observable-gradient-decomposition}
    \dgrad f_w(x)
    &=\dotp{w}{\dgrad V(x)}w+\dotp{w}{x}\dgrad^2V(x)w\nonumber\\
    &\qquad-2\diag(x\odot w)\dgrad^2V(x)w
      -w^{\odot2}\odot\dgrad V(x)-\dgrad^2V(x)y.
\end{align}
Since $\norm{w}=1$, we have $\norm{w^{\odot2}}\le1$ and $\norm{y}\le1$. Hence, by \cref{eq:hessian-observable-gradient-decomposition} and $\norm{V}_{\cD^2}\le BN$,
\begin{align*}
    \norm{\dgrad f_w(x)}
    \le2\norm{\dgrad V(x)}+(\sqrt N+3)\norm{\dgrad^2V(x)}_{\mathrm{op}}
    \lesssim B\sqrt N.
\end{align*}
Combining the bounds on $|f_w|$ and $\norm{\dgrad f_w}$, we conclude that $\norm{f_w}_{\cD^1}\lesssim BN$.
\end{proof}

\subsection{Regularity Estimates from \texorpdfstring{\cref{subsec:critical-recovery-reverse-planting}}{Section \ref{subsec:critical-recovery-reverse-planting}}}
\label[appendix]{app:regularity-critical-recovery}

\subsubsection{Proof of \texorpdfstring{\cref{lem:warm-start-correlation}}{Lemma \ref{lem:warm-start-correlation}}}
\label[appendix]{app:proof-warm-start-correlation}

\begin{proof}
Since $\sigma_N^2\asymp_kN^{-\tfrac{k-1}{2}}$ and $\tfrac{k-2}{4}\le\gamma\le\tfrac{k-2}{4}+\tfrac{1}{4k}$, we have $N^{-\nicefrac{1}{4}}\lesssim\lambda\sigma_N\lesssim N^{-\tfrac14+\tfrac{1}{4k}}$. Hence,
\begin{align*}
    \min\bigl\{1,\lambda N^{-\tfrac{k-1}{4}}\bigr\}\asymp_k\lambda\sigma_N,
    \qquad
    \min\bigl\{N,\lambda^2N^{-\tfrac{k-3}{2}}\bigr\}\asymp_kN\lambda^2\sigma_N^2\to\infty.
\end{align*}
Moreover, since $\xopt\odot\phi=\sgn(u_H)$, we have $m_\phi(\mathbf{1})=m(\sgn(u_H))$ under the planted law. Therefore, adjusting constants appropriately, we infer from \cref{cor:avg-discrete-grad-correlation} that
\begin{align*}
    \Prb\bigl[m_\phi(\mathbf{1})\ge2c_0\lambda\sigma_N\bigr]
    \ge1-\exp\bigl(-cN\lambda^2\sigma_N^2\bigr).
\end{align*}
Finally, since $\abs{m_\phi(x)-m_\phi(y)}\le\frac{2}{N}d_H(x,y)$ for all $x,y\in\Cube$, every $e\in\cB_H(\mathbf{1},\lambda^2\sigma_N^2N^{-\vartheta})$ satisfies $m_\phi(e)\ge m_\phi(\mathbf{1})-2\lambda^2\sigma_N^2N^{-\vartheta}$. The proof is concluded by noting that $2\lambda^2\sigma_N^2N^{-\vartheta}\le c_0\lambda\sigma_N$ for all sufficiently large $N$.
\end{proof}

\subsubsection{Proof of \texorpdfstring{\cref{lem:null-hamiltonian-regularity}}{Lemma \ref{lem:null-hamiltonian-regularity}}}
\label[appendix]{app:proof-null-hamiltonian-regularity}

\begin{proof}
By \cref{lem:tpca-hamiltonian-under-gauge}, $\dgrad H_{\mathrm{null}}(y)=R+\dgrad W(y)$ for every $y\in\Cube$. By \cref{def:planted-law}, $R=\abs{U}$ for some $U\sim\cN(\lambda\sigma_N^2\mathbf{1},\sigma_N^2I)$, and $W$ has the same law as $H_0^\perp$, which equals $H_0-\dotp{\zeta}{\cdot}$ with $\zeta\sim\cN(0,\sigma_N^2I)$ by \cref{lem:kspin-gaussian-conditioning}. Consequently, $\max_{y\in\Cube}\norm{\dgrad H_{\mathrm{null}}(y)}$ is at most $\norm{U}+\max_{y\in\Cube}\norm{\dgrad W(y)}$, where the second term has the same law as $\max_{x\in\Cube}\norm{\dgrad H_0(x)-\zeta}\le\max_{x\in\Cube}\norm{\dgrad H_0(x)}+\norm{\zeta}$. Since $\lambda\sigma_N^2\le1$ and $\sigma_N\le1$, the claim follows from \cref{lem:kspin-sobolev-norm-concentration}, Gaussian concentration~\citep{vershynin2018high}, and a union bound.
\end{proof}

\subsubsection{Proof of \texorpdfstring{\cref{lem:signal-drift-bounds}}{Lemma \ref{lem:signal-drift-bounds}}}
\label[appendix]{app:proof-signal-drift-bounds}

\begin{proof}
Fix $y\in\Cube$ with $m\doteq m_\phi(y)\ge a_0/2$, fix $i\in[N]$, and let $m_{-i}\doteq m-\phi_iy_i/N$. Since $H_{\mathrm{sig}}=N\lambda m_\phi^k-N\lambda\sigma_N^2m_\phi$ and $\phi_i\partial_im_\phi=1/N$, we obtain the following as in \cref{eq:strong-recovery-signal-field}:
\begin{align*}
    \phi_i\partial_iH_{\mathrm{sig}}(y)
    &=\frac{N\lambda}{2}\left[\left(m_{-i}+\frac1N\right)^k-\left(m_{-i}-\frac1N\right)^k\right]-\lambda\sigma_N^2\\
    &=\frac{k\lambda}{2}\int_{-1}^{1}\left(m_{-i}+\frac{s}{N}\right)^{k-1}\,ds-\lambda\sigma_N^2.
\end{align*}
Since $a_0\asymp\lambda\sigma_N\gtrsim N^{-\nicefrac{1}{4}}$, we have $m\ge4/N$, and hence $m/2\le m_{-i}+s/N\le2m$ for every $s\in[-1,1]$. Moreover, since $\sigma_N^2\asymp_kN^{-\tfrac{k-1}{2}}$, we have $\sigma_N^2\le\frac{k}{2}(m/2)^{k-1}$. Consequently, \cref{eq:signal-gradient-bounds} holds with $c_1=k2^{-k}$ and $C_1=k2^{k-1}$.

We next prove the lower bound on the signal drift. By \cref{eq:glauber-generator}, the following holds for any $V:\Cube\to\R$:
\begin{align*}
    (\cL_{V,\beta}m_\phi)(y)
    =\frac{1}{N}\sum_{i=1}^N\phi_i\bigl[\tanh\bigl(\beta\partial_iV(y)\bigr)-y_i\bigr].
\end{align*}
Since $H_\phi=H_{\mathrm{null}}+H_{\mathrm{sig}}$, it follows that
\begin{align*}
    (\cL_{\mathrm{pl}}-\cL_{\mathrm{null}})m_\phi(y)
    &=\frac{1}{N}\sum_{i=1}^N\phi_i\bigl[\tanh\bigl(\beta\partial_iH_\phi(y)\bigr)-\tanh\bigl(\beta\partial_iH_{\mathrm{null}}(y)\bigr)\bigr]\\
    &=\frac{\beta}{N}\sum_{i=1}^N\phi_i\partial_iH_{\mathrm{sig}}(y)\int_0^1\operatorname{sech}^2\bigl(\beta\bigl(\partial_iH_{\mathrm{null}}(y)+u\,\partial_iH_{\mathrm{sig}}(y)\bigr)\bigr)\,du.
\end{align*}
By \cref{eq:signal-gradient-bounds}, every summand is non-negative. Now suppose that $m\le Cr_N$ and $\norm{\dgrad H_{\mathrm{null}}(y)}^2\le C_0N$, and let $\calJ\doteq\{i\in[N]:\abs{\partial_iH_{\mathrm{null}}(y)}\le2\sqrt{C_0}\}$. Since
\begin{align*}
    4C_0\bigl(N-\abs{\calJ}\bigr)
    \le\sum_{i\notin\calJ}\bigl(\partial_iH_{\mathrm{null}}(y)\bigr)^2
    \le\norm{\dgrad H_{\mathrm{null}}(y)}^2
    \le C_0N,
\end{align*}
we have $\abs{\calJ}\ge3N/4$. Moreover, since $m\le Cr_N$ and $\lambda r_N^{k-1}=1$, \cref{eq:signal-gradient-bounds} implies that $\abs{\partial_iH_{\mathrm{sig}}(y)}\le C_1C^{k-1}$. Hence, for every $i\in\calJ$ and $u\in[0,1]$,
\begin{align*}
    \operatorname{sech}^2\bigl(\beta\bigl(\partial_iH_{\mathrm{null}}(y)+u\,\partial_iH_{\mathrm{sig}}(y)\bigr)\bigr)
    \ge\operatorname{sech}^2\bigl(\beta\bigl(2\sqrt{C_0}+C_1C^{k-1}\bigr)\bigr).
\end{align*}
Discarding the summands with $i\notin\calJ$ and applying the lower bound in \cref{eq:signal-gradient-bounds}, we conclude that the claim holds with $c_2=\frac34\beta c_1\operatorname{sech}^2(\beta(2\sqrt{C_0}+C_1C^{k-1}))$.
\end{proof}

\subsubsection{Proof of \texorpdfstring{\cref{lem:power-law-gronwall}}{Lemma \ref{lem:power-law-gronwall}}}
\label[appendix]{app:proof-power-law-gronwall}

\begin{proof}
Since the proof of \cite[Lemma~5.1]{arous2020algorithmic} assumes continuity of $f$, we present a short proof for bounded measurable $f$ for completeness. Define
\begin{align*}
    g(t)\doteq a+c\int_0^tf(s)^{k-1}\,ds.
\end{align*}
Since $k$ is odd, we have $f(t)\ge g(t)\ge a>0$ for every $t\in[0,T]$. Moreover, since $f$ is bounded and measurable, $g$ is absolutely continuous, and for almost every $t\in[0,T]$,
\begin{align*}
    g'(t)=cf(t)^{k-1}\ge cg(t)^{k-1}.
\end{align*}
Since $g$ is bounded away from zero, $g^{-(k-2)}$ is also absolutely continuous, and the following holds for almost every $t\in[0,T]$:
\begin{align*}
    \frac{d}{dt}g(t)^{-(k-2)}
    =-(k-2)g(t)^{-(k-1)}g'(t)
    \le-(k-2)c.
\end{align*}
Integrating and using $g(0)=a$, we obtain
\begin{align*}
    g(t)^{-(k-2)}\le a^{-(k-2)}-(k-2)ct.
\end{align*}
If $(k-2)ca^{k-2}t<1$, the right-hand side is positive, and raising both sides to the power $-\frac{1}{k-2}$ yields
\begin{align*}
    g(t)\ge\bigl[a^{-(k-2)}-(k-2)ct\bigr]^{-\frac{1}{k-2}}.
\end{align*}
The proof is concluded by noting that $f(t)\ge g(t)$.
\end{proof}